\documentclass[a4paper,UKenglish,cleveref, autoref, thm-restate]{lipics-v2021}
\hideLIPIcs
\nolinenumbers
\makeatletter
\def\@ccsdescString{\relax}
\def\subjclassHeading{}
\def\@funding{}
\makeatother
\usepackage[vlined,ruled,linesnumbered]{algorithm2e}
\usepackage{amsmath,amssymb}
\usepackage{graphicx}
\usepackage{booktabs}
\usepackage{tabularx}
\usepackage{makecell}
\usepackage{textcomp}
\usepackage{eurosym}
\usepackage{subcaption}
\usepackage{tikz}
\usetikzlibrary{positioning,calc}
\usepackage{tcolorbox}
\title{Approximation Algorithms for Budget Splitting in Multi-Channel Influence Maximization} %TODO Please add
\author{Dildar Ali}{Department of Computer Science and Engineering, Indian Institute of Technology Jammu, Jammu and Kashmir, 181221, India}{2021rcs2009@iitjammu.ac.in}{https://orcid.org/0000-0002-3427-1904}{}
\author{Ansh Jasrotia}{Department of Information Technology, National Institute of Technology Srinagar, Jammu and Kashmir, 190006, India}{2023nitsgr161@nitsri.ac.in}{https://orcid.org/0009-0006-2685-4217}{}
\author{Abishek Salaria}{Department of Information Technology, National Institute of Technology Srinagar, Jammu and Kashmir, 190006, India}{2023nitsgr185@nitsri.ac.in}{https://orcid.org/0009-0002-1676-0971}{}
\author{Suman Banerjee\footnote{Corresponding author}}{Department of Computer Science and Engineering, Indian Institute of Technology Jammu, Jammu and Kashmir, 181221, India}{suman.banerjee@iitjammu.ac.in}{https://orcid.org/0000-0003-1761-5944}{}

\authorrunning{Ali et al.} %TODO mandatory. First: Use abbreviated first/middle names. Second (only in severe cases): Use first author plus 'et al.'
\keywords{Advertisement, Billboard, Social Network, Bi-submodularity, Influence Maximization} %TODO mandatory; please add comma-separated list of keywords

\begin{document}

\maketitle

\begin{abstract}
How to utilize an allocated budget effectively for branding and promotion of a commercial house is an important problem, particularly when multiple advertising media are available. There exist multiple such media, and among them, two popular ones are billboards and social media advertisements. In this context, the question naturally arises: how should a budget be allocated to maximize total influence? Although there is significant literature on the effective use of budgets in individual advertising media, there are hardly any studies examining budget allocation across multiple advertising media. To bridge this gap, this paper introduces the \textsc{Budget Splitting Problem in Billboard and Social Network Advertisement}. We introduce the notion of \emph{interaction effect} to capture the additional influence due to triggers from multiple media of advertising. Using this notion, we propose a noble influence function $\Phi(,)$ that captures the total influence and shows that this function is non-negative, monotone, and non-bisubmodular. We introduce \emph{bi-submodularity ratio $(\gamma)$} and \emph{generalized curvature $(\alpha)$} to measure how close a function is to being bi-submodular and how far a function is from being modular, respectively. We propose the Randomized Greedy and Two-Phase Adaptive Greedy approach, where the influence function is non-bisubmodular and achieves an approximation guarantee of $\frac{1}{\alpha}\left(1-e^ {-\gamma \alpha} \right)$. We conducted several experiments using real-world datasets and observed that the proposed solution approach's budget splitting leads to a greater influence than existing approaches.  
\end{abstract}
\pagebreak
\section{Introduction}
Almost all commercial houses use advertising to promote their products and build a customer base. As mentioned in recent marketing literature, a commercial house spends around $7-10\%$ of its annual revenue on advertising\footnote{\url{https://www.lamar.com/howtoadvertise/Research/}}. For effective advertising, it is important that this budget is utilized appropriately. Two popular advertising approaches are advertising through digital billboards and social networks. On a billboard, a commercial house displays an advertisement (which may be an animation or video) with the expectation that people nearby will view it, and possibly some will be influenced to buy the product. Nowadays, billboards are digital and are assigned to advertisers in slots (fixed duration), on a payment basis. Now, it is quite natural for any commercial house that the available budget will be limited. The goal will be to select a number of influential billboard slots within the allocated budget such that the influence is maximized \cite{wang2019efficiently,ali2022influential}. The other way to advertise is through social media. Many internet giants, including Google and Facebook, generate significant revenue from social media advertising. In this method, a number of highly influential users are chosen and made `influenced' externally. Now, it is assumed that information diffuses and propagates throughout the network, and at the end of the diffusion process, the set of users influenced is referred to as the seed set's influence. In the existing literature \cite{chen2009efficient,kempe2003maximizing}, several models exist to study the diffusion process.
\subparagraph{Our Observation.} Existing studies in out-of-home and online advertising scenarios share a common objective: (a) to help advertisers achieve maximum influence under budget constraints in a single or multi-advertiser setting \cite{zhang2020towards,ali2024effective,kempe2003maximizing}, and (b) to minimize the regret of an influence provider by the effective utilization of resources towards advertisers  \cite{ali2023efficient,ali2024regret,ali2024minimizing,zhang2021minimizing}. A more challenging problem, yet to be explored from the perspective of real-world influence providers given a budget, is how to allocate it across different media for advertising to maximize total influence. 

\subparagraph{Motivation.} Both Billboard and social network advertisement techniques have their respective advantages and disadvantages. Billboards are only available in urban and suburban areas. So, it is not possible to influence the population residing in rural areas. A smart way to utilize the advertisement budget is to combine billboard and social media, leveraging their strengths and mitigating their weaknesses. So, the question arises: given a fixed budget, how can we allocate it for effective advertising? Since billboards and social media follow distinct influence models, simply summing their effects ignores \textit{Interaction Effect} (see Definition \ref{Def:Interaction_Effect}) as shown in Example \ref{Example:1}. Therefore, there is a need for a novel influence function that considers the \textit{Interaction Effect}, highlighting the importance of studying this problem.

\begin{example}\label{Example:1}
Assume there are four advertiser $\mathcal{A}=\{a_{1}, a_{2}, a_{3}, a_{4}\}$, Six billboard slots $\mathcal{BS}= \{bs_{1}, bs_{2}, \ldots, bs_{6}\}$ (see Table \ref{ETable:1}), seven seed nodes $\mathcal{P} = \{ p_{1},p_{2}, \ldots, p_{7}\}$ (see Table \ref{ETable:2}) and an influence provider $\mathcal{Z}$ with influence demand from social network and billboard slots from the advertisers as shown in Table \ref{ETable:3}. Now, we consider two cases: first, the influence is calculated separately, and the aggregated influence is presented. The allocation of slots to the advertisers is as follows: $a_{1} = \{bs_{1}, bs_{2}\}$, $a_{2} = \{bs_{5}\}$, $a_{3} = \{bs_{3}, bs_{4}\}$, $a_{4} = \{bs_{6}\}$ and allocation of seed nodes to the advertisers is $a_{1} = \{p_{1}, p_{2}\}$, $a_{2} = \{p_{5}\}$, $a_{3} = \{p_{3}, p_{4}\}$, $a_{4} = \{p_{6}, p_{7}\}$. The influence from slots and seeds for the advertiser $a_{1}, a_{2}, a_{3}$, and $a_{4}$ are $29$, $21$, $31$ and $27$, respectively. So, the influence demand of advertiser $a_{1}$ and $a_{2}$ are not satisfied. In the second case where we consider the \textit{Interaction Effect} and use the proposed influence model (see Definition \ref{Def:Combine_Influence_Model}) and the influence for the advertiser $a_{1}, a_{2}, a_{3}$, and $a_{4}$ are $30.86$, $22.38$, $32.71$ and $27.15$, respectively. Hence, all advertisers are satisfied, which is due to the additional influence resulting from the interaction effect between slots and seeds.
% \begin{table}[!h]
% \begin{center}
% \begin{minipage}{0.48\textwidth}
% \small
%    \centering
%    \begin{tabular}{| c | c | c | c | c | c | c |}
%    \hline
%    $\mathcal{BS}_{i}$ & $bs_{1}$ & $bs_{2}$ & $bs_{3}$ & $bs_{4}$ & $bs_{5}$ & $bs_{6}$ \\ \hline
%    $\mathcal{I}(bs_{i})$ & 2 & 4 & 3 & 1 & 6 & 5 \\ \hline
%    Cost & \$6 & \$12 & \$9 & \$3 & \$18 & \$15 \\ \hline
%    \end{tabular}
%    \caption{\label{ETable:1} Billboard Info.}
% \end{minipage}
% \hfill

% \begin{minipage}{0.48\textwidth}
% \small
%    \centering
%    \begin{tabular}{| c | c | c | c | c | c | c |c|}
%    \hline
%    $\mathcal{P}_{i}$ & $p_{1}$ & $p_{2}$ & $p_{3}$ & $p_{4}$ & $p_{5}$ & $p_{6}$ & $p_{7}$ \\ \hline
%    $\mathcal{I}^{\mathcal{G}}(p_{i})$ & 10 & 13 & 13 & 14 & 15 & 12 & 10 \\ \hline
%    Cost & \$50 & \$65 & \$65 & \$70 & \$75 & \$60 & \$50\\ \hline
%    \end{tabular}
%    \caption{\label{ETable:2} Seed Node Info.}
% \end{minipage}
% \hfill

% \begin{minipage}{0.48\textwidth}
% \small
%    \centering
%    \begin{tabular}{ | c | c | c | c | c |}
%    \hline
%    Advertiser($\mathcal{A}$) & $a_{1}$ & $a_{2}$ & $a_{3}$ & $a_{4}$ \\ \hline
%    Demand ($\sigma_{i}$) & 30 & 22 & 30 & 26 \\ \hline
%    Budget($\mathcal{B}_{i}$) & \$140 & \$95 & \$165 & \$150 \\ \hline
%    \end{tabular}
%    \caption{\label{ETable:3} Advertiser Info.}
% \end{minipage}
% \end{center} 
% \end{table}
% \end{example}

\begin{table}[t]
\centering
\small

% ---------- Row 1 ----------
\begin{minipage}{0.48\linewidth}
\centering
\begin{tabular}{|c|c|c|c|c|c|c|}
\hline
$\mathcal{BS}_{i}$ & $bs_1$ & $bs_2$ & $bs_3$ & $bs_4$ & $bs_5$ & $bs_6$ \\ \hline
$\mathcal{I}(bs_i)$ & 2 & 4 & 3 & 1 & 6 & 5 \\ \hline
Cost & \$6 & \$12 & \$9 & \$3 & \$18 & \$15 \\ \hline
\end{tabular}
\captionof{table}{Billboard Info.}
\label{ETable:1}
\end{minipage}
\hfill
\begin{minipage}{0.48\linewidth}
\centering
\begin{tabular}{|c|c|c|c|c|}
\hline
Advertiser ($\mathcal{A}$) & $a_1$ & $a_2$ & $a_3$ & $a_4$ \\ \hline
Demand ($\sigma_i$) & 30 & 22 & 30 & 26 \\ \hline
Budget ($\mathcal{B}_i$) & \$140 & \$95 & \$165 & \$150 \\ \hline
\end{tabular}
\captionof{table}{Advertiser Info.}
\label{ETable:3}
\end{minipage}

\par\medskip
% ---------- Row 2 ----------
\begin{minipage}{0.48\linewidth}
\centering
   \begin{tabular}{| c | c | c | c | c | c | c |c|}
   \hline
   $\mathcal{P}_{i}$ & $p_{1}$ & $p_{2}$ & $p_{3}$ & $p_{4}$ & $p_{5}$ & $p_{6}$ & $p_{7}$ \\ \hline
   $\mathcal{I}^{\mathcal{G}}(p_{i})$ & 10 & 13 & 13 & 14 & 15 & 12 & 10 \\ \hline
   Cost & \$50 & \$65 & \$65 & \$70 & \$75 & \$60 & \$50\\ \hline
   \end{tabular}
\medskip
\centering
\captionof{table}{Seed Node Info.}
\label{ETable:2}
\end{minipage}
\end{table}
\end{example}

\subparagraph{Main Contributions.} 
\par To the best of our knowledge, this is the first study in this direction. In particular, we make the following contributions in this paper:
\begin{itemize}
\item We study how to split a budget between two advertising channels: social media and billboards. We call this problem the \textsc{Budget Splitting} Problem.
\item We propose a noble mathematical formulation of the problem that combines an influence function and an interaction effect, pose it as a discrete optimization problem, and show that it is \textsf{NP-hard}.
\item We establish important properties of the proposed influence function and propose a greedy-based solution with detailed analysis and performance guarantee.
\item We perform extensive experiments on real-world datasets to show the effectiveness and efficiency of our approach.
\end{itemize}

\subparagraph{Organization of the Paper.}
\par The paper is organized as follows. Section~\ref{Sec:RW} discusses the relevant studies of our work. Section~\ref{Sec:BPD} presents the background and formal definition of the problem. Section~\ref{Sec:Proposed} describes the proposed solution approaches. Approximation guarantees are discussed in Section~\ref{Sec:Approx_Guarantee}. Section~\ref{Sec:EE} reports the experimental evaluations. Finally, Section~\ref{Sec:Conclusion} concludes the paper.

\section{Related Work}\label{Sec:RW}
% In this section, we discuss the relevant studies of this work. We categorize our study into two main areas: Influence Maximization and Regret Minimization.
\subsection{Influence Maximization}
In the past few years, with the exponential growth in trajectory databases \cite{wang2021survey,zheng2015trajectory}, the study of trajectory-based influence maximization has increased in both online and out-of-home advertisements \cite{ali2022influential,10.1145/3308558.3313635,zheng2015trajectory,zhang2020towards,tardos2003maximizing,sharma2024minimizing}. Most influence maximization studies are conducted in the context of social network advertisements \cite{tardos2003maximizing,sharma2024minimizing,banerjee2019maximizing,banerjee2020survey}. There are a few in the context of billboard advertisements \cite{ali2022influential,ali2024toward, zhang2020towards}. However, there is no literature that considers both social networks and billboard advertising jointly to maximize an advertiser's influence and to utilize the advertiser's budget. Several works relate closely to our budget split problem. In particular, billboard advertisements, Zhang et al.\cite{zhang2020towards} studied trajectory-driven billboard placement using a greedy algorithm with a $(1 - 1/e)$ approximation. Later, Wang et al.\cite{wang2022data} focused on targeted billboard selection using user mobility and ad relevance. Additionally, Zhang et al. \cite{zhang2019optimizing} proposed a tangent-line-based method for billboard selection, achieving an approximation of $\frac{\theta}{2}(1 - 1/e)$. One relevant problem in social media advertising is selecting a limited number of highly influential users as seed users to maximize the influence. This problem has been referred to as the \emph{Social Influence Maximization Problem}, and a plethora of solution methodologies are available in the literature \cite{li2018influence,banerjee2020survey}. Later, Nugayan et al. \cite{nguyen2013budgeted} studied the influence maximization problem, assuming that each user in the network has a selection cost and that seeds must be selected within the allocated budget. The major difference between our work and previous studies is that previous studies allocate the advertiser's budget entirely to either billboard or social network advertising. But in our work, we split the advertiser's budget between two media (billboards and social networks) to maximize the advertiser's total influence while ensuring the budget is used effectively.

\subsection{Regret Minimization}
In the Social Viral Marketing (SVM) setting, the influence provider platform, such as Twitter or Facebook, promotes advertisements on a social network. In return, the influence provider receives payments from advertisers. The common business model used is cost per engagement (CPE). Under this model, an advertiser pays the host for every click or engagement received by its advertisement \cite{facebook_business_cpe,haven2007marketing,wang2020social}. Several recent studies have focused on SVM. One important line of work aims to minimize the regret of the influence provider \cite{aslay2015viral,aslay2017revenue,banerjee2019maximizing}. In this setting, each advertiser specifies a required amount of influence. If the achieved influence is less than the demand, the advertiser pays only for the actual influence obtained. If the achieved influence exceeds the demand, the advertiser does not pay for the extra influence. As a result, both insufficient influence and excessive influence lead to regret for the influence provider \cite{aslay2014viral,ali2024regret,ali2024toward,ali2024minimizing}. Other studies focus on maximizing revenue under the CPE model. In these works, the revenue is defined as the total payment collected from all advertisers. Each advertiser’s payment equals the sum of CPE values of the activated users \cite{aslay2015viral,banerjee2019maximizing}. In the existing literature, two types of regret-minimization problems are studied in the context of advertising: SVM and Minimizing Regret for the OOH Advertising Market problem (MROAM). There are two major differences between the SVM and the MROAM problem. First, the business models are different. In SVM, the influence provider always receives payment based on engagements. In some cases, the influence provider may receive no payment at all if the advertiser’s demand is not met \cite{zhang2021minimizing,ali2024minimizing}. Second, the influence models are different. In MROAM, influence is determined by geographical proximity. A billboard influences users who are physically close enough to encounter it. The influence does not spread from one user to another \cite{ali2022influential,zhang2021minimizing,ali2024minimizing,ali2024toward}. In contrast, SVM uses probabilistic diffusion models, such as the Independent Cascade and Linear Threshold models. Under these models, influence propagates through the social network from one user to another \cite{barbieri2013topic,bian2020efficient,guo2020influence}. Due to this difference, SVM research mainly focuses on estimating the spread of influence in virtual social networks. In contrast, MROAM relies on spatial interactions between users and billboards. Therefore, the optimization problems in these two settings are fundamentally different.

\section{Background Problem Definition}\label{Sec:BPD}
In this section, we introduce preliminary concepts and formally define our problem. 
For any two natural numbers $x$ and $y$ with $y \geq x$, $[x, y]$ denotes the set $\{x, x+1, \ldots, y\}$ and $[x]$ denotes the set $\{1,2,\ldots,x\}$.
\subsection{Set Function and Its Properties}
Consider the set $X$, with $n$ elements and a function $f$ defined on $X$, i.e., $f:2^{X} \longrightarrow \mathbb{R}$. $f$ is called a \emph{normalized} set function if $f(\emptyset)=0$. $f$ is said  to be \emph{nonnegative} if for any $X^{'} \subseteq X$, $f(X^{'}) \geq 0$, \emph{monotone} if for any $X^{'} \subseteq X^{''} \subseteq X$, $f(X^{'}) \leq f(X^{''})$, and \emph{submodular} if for any $X^{'} \subseteq X^{''} \subseteq X$, and for all $x \in X \setminus X^{''}$, $f(X^{'} \cup \{x\}) - f(X^{'}) \geq f(X^{''} \cup \{x\}) - f(X^{''})$. The notion of a set function has been extended to \emph{bi-set function} when the function contains two arguments. Consider a bi-set function defined on the ground set $X_1 \times X_2$, i.e., $\mathcal{F}: 2^{X_1 \times X_2} \longrightarrow \mathbb{R}$. It is normalized when $\mathcal{F}(\emptyset, \emptyset)=0$. A bi-set function is said to be monotone if for all $(A_1,B_1) \in 2^{X_1 \times X_2}$ and $a \in X_1 \setminus A_1$, $b \in X_2 \setminus B_2$, $\mathcal{F}(A_1 \cup \{a\},B_1) \geq \mathcal{F}(A_1,B_1)$ and $\mathcal{F}(A_1,B_1 \cup \{b\}) \geq \mathcal{F}(A_1,B_1)$ holds \cite{singh2012bisubmodular}. The bi-set function $\mathcal{F}$ is said to be bisubmodular if for any $(A_1,B_1), (A^{'}_1,B^{'}_1) \in 2^{X_1 \times X_2}$ such that $A_1 \subseteq A^{'}_1 \subseteq X_1$ and $B_1 \subseteq B^{'}_1 \subseteq X_2$, also for all $a \in X_1 \setminus A^{'}_1$ and $b \in X_2 \setminus B^{'}_1$ both the conditions hold: $\mathcal{F}(A_1 \cup \{a\},B_1)-\mathcal{F}(A_1,B_1) \geq \mathcal{F}(A^{'}_1 \cup \{a\},B^{'}_1)-\mathcal{F}(A^{'}_1,B^{'}_1)$ and  $\mathcal{F}(A_1,B_1 \cup \{b\})-\mathcal{F}(A_1,B_1) \geq \mathcal{F}(A^{'}_1,B^{'}_1 \cup \{b\})-\mathcal{F}(A^{'}_1,B^{'}_1)$.

\subsection{Billboard Advertisement}
In a billboard advertisement, the key components are the trajectory and the billboard database. A trajectory database $\mathbb{D}$ contains location information of persons moving over time and is defined as a collection of tuples of the form $(u_{i}, loc, [t_{1},t_{2}])$. This signifies that the person $u_{i}$ is at the location, $loc$ in between the time slot $[t_{1},t_{2}]$. Let,  $m$ number of tuples are in $\mathbb{D}$ and it contains the location information for the duration $[T_1,T_2]$ for $n$ persons $\mathcal{U} = \{ u_{1}, u_{2}, \ldots, u_{n}\}$. So, we can say every tuple $t \in \mathbb{D}$, have the time stamp $[t_{i},t_{j}]\subseteq [T_{1},T_{2}]$. The billboard database $\mathbb{B}$ stores information on billboards placed across a city, represented as tuples $(b_{id}, \textit{loc}, \textit{time\_slot}, \textit{cost})$, where $b_{id}$ is the unique billboard ID, $\textit{loc}$ denotes its location, and $\textit{cost}$ is the rental price for that slot. Let $\mathbb{B}$ contain slot information for $r$ distinct billboards $\mathcal{B} = \{b_1, b_2, \ldots, b_r\}$. A user $u_i \in \mathcal{U}$ is said to be influenced by a billboard $b_j$ if their presence interval $[t_a, t_b]$ at the location overlaps with the ad display time $[t_x, t_y]$, i.e., $[t_a, t_b] \cap [t_x, t_y] \neq \emptyset$, with some influence probability. Consider that an advertiser can lease a billboard for the duration $\Delta$. Hence, for every billboard, the number of slots is $\frac{T_2-T_1}{\Delta}$. As an example, the slot $(b_i, [t,t+\Delta])$ signifies the slot from the duration from $t$ to $t+\Delta$ for the billboard $b_i$. The set of all billboard slots is denoted by $\mathcal{BS}$ and defined as $\mathcal{BS}=\{(b_i, [t,t+\Delta]): b_i \in \mathbb{B} \text{ and } t \in \{T_1, T_1+ \Delta, T_1+ 2 \Delta, \ldots, T_2-\Delta \}\}$. Every billboard slot is associated with a cost, which is formalized by the cost function $\mathcal{C}_{\mathcal{B}}: \mathcal{BS} \longrightarrow \mathbb{R}^{+}$. Now, given $\mathbb{S} \subseteq \mathcal{BS}$, we denote the influence of slot set $\mathbb{S}$ by $\mathcal{I}(\mathbb{S})$ and define it in Definition \ref{Def:1}. 
\begin{definition} [Influence of Billboard Slots] \label{Def:1}
Given a subset of billboard slots $\mathbb{S} \subseteq \mathbb{BS}$, its influence $\mathcal{I}(\mathbb{S})$ can be computed using Equation \ref{Eq:1}.
{\small
\begin{equation} \label{Eq:1}
\mathcal{I}(\mathbb{S}) = \underset{u_{j} \in \mathcal{U}}{\sum} 1- \underset{bs_i \in \mathbb{S}}{\prod} (1- Pr(bs_i,u_j))
\end{equation}
}
\end{definition}
We adopt a widely used influence probability model \cite{zhang2020towards,zhang2021minimizing,ali2022influential,ali2024toward,ali2023efficient}, where the probability of a billboard slot $bs_i$ influencing a user $u_j$ is defined as $Pr(bs_i, u_j) = \frac{\text{Size}(bs_i)}{\max\limits_{bs_i \in \mathcal{B}} \text{Size}(bs_i)}$, where $\text{Size}(bs_i)$ denotes the panel size of the billboard slot $bs_i$. The influence function $\mathcal{I}(\cdot)$ maps each subset of billboard slots to a non-negative real value, i.e., $\mathcal{I}: 2^{\mathcal{BS}} \rightarrow \mathbb{R}^{+}_0$, with $\mathcal{I}(\emptyset) = 0$.
\begin{lemma}
The influence function $\mathcal{I}()$ is  non-negative, monotone, and submodular \cite{ali2023efficient,ali2024effective,ali2025fairness}.
\end{lemma}

% \begin{figure*}[t]
% \centering

% % ---------- First Row ----------
% \begin{minipage}{0.5\textwidth}
%     \centering
%     \includegraphics[width=\linewidth]{BB.png}
    
%     (a) Billboard Database
% \end{minipage}
% \hfill
% \begin{minipage}{0.46\textwidth}
%     \centering
%     \includegraphics[width=\linewidth]{TJ.png}
    
%     (b) Trajectory Database
% \end{minipage}

% \vspace{0.3cm}

% % ---------- Second Row ----------
% \begin{minipage}{0.52\textwidth}
%     \centering
%     \includegraphics[width=\linewidth]{Diagram1.pdf}
    
%     (c) User--billboard exposure overlap
% \end{minipage}
% \hfill
% \begin{minipage}{0.45\textwidth}
%     \centering
%     \includegraphics[width=\linewidth]{Diagram2.pdf}
    
%     (d) Social Network Diffusion Process
% \end{minipage}

% \caption{A motivating example illustrating user trajectories, billboard placements, and their spatio-temporal influence interactions.}
% \label{fig:motivating_example}
% \end{figure*}

\subsection{Social Network Advertisement}
Consider, we know the social network of users $\mathcal{U}=\{u_1,u_2, \ldots,u_n\}$, which is represented by a simple, weighted, and directed graph $\mathcal{G}(\mathcal{U}, \mathcal{E}, \mathcal{P})$. The edge set $\mathcal{E}(\mathcal{G})$ contains the binary social relationship (e.g., friendship, follower ship, etc.), and the edge weight function $\mathcal{P}$ that maps each edge to its corresponding influence probability, i.e., $\mathcal{P}: \mathcal{E}(\mathcal{G}) \longrightarrow (0,1]$. For each edge $(u_iu_j)$, its influence probability is denoted by $\mathcal{P}_{(u_iu_j)}$. If $(u_iu_j) \notin \mathcal{E}(\mathcal{G})$, $\mathcal{P}_{(u_iu_j)}=0$. A key phenomenon in social networks is the diffusion of information. To study this, various diffusion models have been proposed, among which the \emph{Independent Cascade Model} is one of the most widely used. The diffusion of influence begins with a seed set and propagates through the network. The influence of a seed set is measured by the total number of activated nodes at the end of the process. Formally, this is captured by the influence function \(I_{\mathcal{G}} : 2^{\mathcal{U}} \rightarrow \mathbb{R}^{+}_0\), where \(I_{\mathcal{G}}(\emptyset) = 0\). Under the IC Model, \(I_{\mathcal{G}}\) is non-negative, monotone, and submodular \cite{tardos2003maximizing}. 

\subsection{Combine Influence Model}\label{Sec:Combined_Inf}
Considering digital billboards and social networks, we aim to maximize their joint impact. To measure this, we introduce a combined influence model in Definition \ref{Def:Combine_Influence_Model}. 

\begin{definition}[Combine Influence Model]\label{Def:Combine_Influence_Model}
Given a subset of billboard slots $\mathbb{S} \subseteq \mathbb{BS}$, and a set of seed nodes $\mathcal{N} \in \mathcal{G}$, the influence $\Phi(\mathbb{S},\mathcal{N})$ can be calculated using Equation \ref{Eq:Combine_influence_model}.
{\small
\begin{equation} \label{Eq:Combine_influence_model}
\Phi(\mathbb{S},\mathcal{N}) = \mathcal{I}(\mathbb{S}) + \mathcal{I}_{\mathcal{G}}(\mathcal{N}) + \Psi(\mathbb{S},\mathcal{N})
\end{equation}
}
where $\mathcal{I}(\mathbb{S})$ and $\mathcal{I}_{\mathcal{G}}(\mathcal{N})$ are the influences from the slot set and seed nodes, respectively. $\Psi(\mathbb{S},\mathcal{N})$ is the interaction effect between the billboard and social network. 
\end{definition}

The influence function $\Phi(.,.)$ is a mapping from all possible combinations of subsets of slots and seed sets to the influence, i.e., $\Phi: 2^{\mathbb{BS}} \times 2^{\mathcal{G}} \longrightarrow \mathbb{R}^{+}_{0}$ with $\Phi(\emptyset, \emptyset)=0$. The existing literature \cite{pavlou2000measuring,sundar2017using,deighton1984interaction} considered the effect of interaction in advertising campaigns. However, there is no specific mechanism to compute it. We explicitly model this interaction effect and define it in Definition \ref{Def:Interaction_Effect}.

\begin{definition}[Interaction Effect]\label{Def:Interaction_Effect}
An interaction effect in influence maximization quantifies how the combined influence of billboards and social media deviates from their independent effects. Mathematically,
{\small
\begin{equation}
\Psi(\mathbb{S},\mathcal{N}) = \sum_{u \in \mathcal{U}}[( 1 - \prod_{b \in \mathbb{S}} (1 - Pr(u, b))) \cdot ( 1 - \prod_{v \in \mathcal{N}} (1 - Pr(u, v)))]
\end{equation}
}
where $1 - \prod_{b \in \mathbb{S}} (1 - Pr(u, b))$ is the probability user $u$ being influenced by at least one billboard slot, and $1 - \prod_{v \in \mathcal{N}} (1 - Pr(u, v))$ is the probability that user $u$ is activated by seed node $v$.
\end{definition}
\begin{theorem}\label{Th:Non-bisubmodular_Proof}
Given a trajectory database $\mathbb{D}$, billboard slots $\mathbb{BS}$, and Social Network Users $\mathcal{V}(\mathcal{G})$, the influence function $\Phi(.,.)$ is non-negative, monotone, and non-bisubmodular.
\end{theorem}
\begin{proof}
Each term in $\Phi(\mathbb{S},\mathcal{N})$ represents a probability-based influence function. First, $\mathcal{I}(\mathbb{S})$ is a sum of probabilities, hence $\mathcal{I}(\mathbb{S}) \geq 0$. Secondly, $\mathcal{I}_{\mathcal{G}}(\mathcal{N})$ follows the Independent Cascade Model (ICM) (Kempe et al.~\cite{kempe2003maximizing}) and represents an expected number of influenced users, ensuring $\mathcal{I}_{\mathcal{G}}(\mathcal{N}) \geq 0$. The interaction effect $\Psi(\mathbb{S},\mathcal{N})$ is a product of two non-negative influence terms, guaranteeing $\Psi(\mathbb{S},\mathcal{N}) \geq 0$. Thus, $\Phi(\mathbb{S},\mathcal{N}) \geq 0$.

\par In $\mathcal{I}(\mathbb{S})$, adding a billboard $b$ increases the probability of influencing users, ensuring $\mathcal{I}(\mathbb{S})$ is increasing. The function $\mathcal{I}_{\mathcal{G}}(\mathcal{N})$ is known to be monotone under ICM \cite{kempe2003maximizing}. The interaction effect $\Psi(\mathbb{S},\mathcal{N})$ increases as either $\mathbb{S}$ or $\mathcal{N}$ grows because increasing the influence in either component leads to a larger combined effect. Hence, $\Phi(\mathbb{S},\mathcal{N})$ is monotone.

\par The function $\mathcal{I}(\mathbb{S})$ and $\mathcal{I}_{\mathcal{G}}(\mathcal{N})$ is known to be submodular. However, the interaction effect $\Psi(\mathbb{S},\mathcal{N})$ is a product of two submodular terms, which does not necessarily preserve submodularity. In fact, multiplicative interaction terms often result in supermodular behavior, where the marginal gain of adding elements increases instead of decreasing, and this violates the bisubmodular inequality. Thus, $\Phi(\mathbb{S},\mathcal{N})$ is non-bisubmodular. Therefore, $\Phi(\mathbb{S},\mathcal{N})$ is non-negative, monotone, and non-bisubmodular.
\end{proof}

\subsection{Problem Definition}
We study the problem of splitting a fixed budget B into $B_1$ for billboard and $B_2$ for social network advertising to maximize influence. We refer to this as the \textsc{Budget Split Problem in Billboard and Social Network Advertisement}, formally represented as $I(\mathcal{G}, \mathbb{S}, \mathbb{D}, \mathbb{B}, B)$ and defined in Definition~\ref{Def:6}.

\begin{definition} [Budget Split Problem in Billboard and Social Network Advertisement] \label{Def:6}
Given an instance $I(\mathcal{G}, \mathbb{S}, \mathbb{D}, \mathbb{B}, B)$, this problem asks to divide the allocated budget $B$ into two halves (say $B_1$ and $B_2$) such that $B \leq (B_1 + B_2)$ and the total influence as defined by Equation \ref{Eq:Combine_influence_model} is maximized. Mathematically,
{\small
\begin{equation} \label{Eq:5}
(\mathcal{S}^{OPT}_{1}, \mathcal{S}^{OPT}_{2}) \longleftarrow \underset{\substack{\mathcal{S}_{1} \subseteq \mathcal{BS} \wedge \mathcal{C} (\mathcal{S}_{1}) \leq B_1, \\ \mathcal{S}_{2} \subseteq \mathcal{V} \wedge \mathcal{C} (\mathcal{S}_{2}) \leq B_2, \\ B_1 + B_2 \leq B}}{argmax} \ \Phi(\mathcal{S}_{1}, \mathcal{S}_{2})
\end{equation} 
}
Here, $\mathcal{S}^{OPT}_{1}$, and $\mathcal{S}^{OPT}_{2}$ denotes the optimal slot subset and the optimal seed set for the budget $B_{1}$ and $B_{2}$, respectively.
\end{definition}
From the computational point of view, this problem can be posed as follows.

\begin{tcolorbox}
\underline{\textsc{Budget Splitting Problem}} \\
%\vspace*{2 cm} 
\textbf{Input:} A trajectory ($\mathbb{D}$) and Billboard Database ($\mathbb{B}$), A set of slots ($\mathbb{S}$) and Social network $\mathcal{G}$, Influence Function $\Phi$, Budget $B$.

\textbf{Problem:} Find out a split of a given budget such that the total influence is maximized.
\end{tcolorbox}
The influence maximization problem in billboard \cite{ali2023influential} and social network \cite{tardos2003maximizing} advertisement had an inapproximability result stated in Theorem \ref{Th:NP-hard}. Therefore, the same inapproximability results also hold for our problem.
\begin{theorem}\label{Th:NP-hard}
The \textsc{Budget Splitting Problem in Billboard and Social Network Advertisement} is NP-hard and hard to approximate to any constant factor.
\end{theorem}
% Therefore, the same inapproximability results also hold for our problem as well. Next, we describe the proposed methodologies.
\section{Proposed Methodology} \label{Sec:Proposed}
In the literature, there exist methods that provide strong approximation guarantees for a submodular function \cite{kapralov2013online,feldman2011unified}, a bisubmodular function \cite{singh2012bisubmodular,ward2014maximizing}, and an approximately submodular function \cite{du2008analysis,feige2013welfare}. However, as mentioned in Theorem \ref{Th:Non-bisubmodular_Proof}, the influence function is non-bisubmodular. To tackle this, we define \emph{bisubmodularity ratio} and \emph{radius of curvature} to derive the performance guarantee of the proposed algorithm. This algorithm is of independent interest and may also be useful in solving other problems. Now, we define the notion of \emph{Bisubmodularity Ratio} in Definition \ref{Def:BR}.
\begin{definition} [Bisubmodularity Ratio] \label{Def:BR}
Let $\Phi(.,.)$ be a non-negative, monotone function. The bisubmodularity ratio is the largest scalar $\gamma$ such that the conditions mentioned in Equation \ref{Eq:6} and \ref{Eq:7} are satisfied.

{\small
\begin{align} \label{Eq:6}
& \sum_{v \in \Omega \setminus q, ~\mathcal{N} \subseteq V(\mathcal{G})} \Phi(q \cup \{v\}, \mathcal{N}) - \Phi(q, \mathcal{N}) \geq \gamma \cdot \left( \Phi(q \cup \Omega, \mathcal{N}) - \Phi(q, \mathcal{N}) \right)
\end{align}
\begin{align} \label{Eq:7}
&\sum_{v \in \Omega' \setminus q',~\mathbb{S} \subseteq \mathbb{BS}} \Phi(\mathbb{S}, \mathcal{N} \cup \{v\}) - \Phi(\mathbb{S}, \mathcal{N}) \geq \gamma \cdot \left( \Phi(\mathbb{S}, \mathcal{N} \cup \Omega') - \Phi(\mathbb{S}, \mathcal{N}) \right).
\end{align}
}

Here, $q, \Omega \subseteq \mathbb{BS},~ q', \Omega' \subseteq V(\mathcal{G})$ and $\gamma$ quantifies how close the function is to being bisubmodular. Next, we define the notion of generalized curvature in Definition \ref{Def:generalized_curvature}.
\end{definition}

\begin{definition}[Generalized Curvature]\label{Def:generalized_curvature}
The curvature of a non-negative, monotone set function $\Phi(.,.)$ is the smallest scalar $\alpha$ such that the conditions mentioned in Equation \ref{Eq:8} and \ref{Eq:9} are satisfied. 
{\small
\begin{align} \label{Eq:8}
&\Phi(\mathbb{S} \setminus \{i\} \cup \Omega, \mathcal{N}) - \Phi(\mathbb{S} \setminus \{i\}, \mathcal{N}) \geq (1 - \alpha) \cdot \left (\Phi(\mathbb{S}, \mathcal{N})-  \Phi(\mathbb{S} \setminus \{i\}, \mathcal{N}\right))
\end{align}
\begin{align}\label{Eq:9}
&\Phi(\mathbb{S}, \mathcal{N} \setminus \{j\} \cup \Omega') - \Phi(\mathbb{S}, \mathcal{N} \setminus \{j\}) \geq (1 - \alpha) \cdot \left( \Phi(\mathbb{S}, \mathcal{N}) -  \Phi(\mathbb{S}, \mathcal{N} \setminus \{j\}) \right)
\end{align}
}
Here, $\mathbb{S}, \Omega \subseteq \mathbb{BS}, i \in \mathbb{S} \setminus \Omega, ~ \mathcal{N}, \Omega' \subseteq V(\mathcal{G}), j \in \mathcal{N} \setminus \Omega'$, and $\alpha$ measures how far a function deviates from modularity. Next, we describe the proposed solution approach.
\end{definition}

\par The symbols and notations used in this paper have been mentioned in Table \ref{Table1:Notations}. 
 \begin{table}
	% \scriptsize
		\caption{Symbols and Notations with their Interpretations}
		\label{Table1:Notations}
		\vspace*{3mm}
		\centering
		\resizebox{\columnwidth}{!}{%
		\begin{tabular}{p{0.20\linewidth}p{0.80\linewidth}}
			\hline
			\hline
			\textbf{Notation} & \textbf{Description}\\
			\hline
			\hline
			$ \mathbb{D}$ & The Trajectory Database \\
			\hline 
			$m$ & Number of tuples in $ \mathbb{D}$ \\
			\hline 
			$\mathcal{U}$ & Set of people covered by $\mathbb{D}$ \\
			\hline 
			$t$ & An arbitrary tuple of $ \mathbb{D}$ \\ 
			\hline 
			$t_{u}$ & The set of users associated with tuple $t$ \\
			\hline 
			$\mathcal{L}$ & Locations covered by the database $\mathbb{D}$ \\
			\hline 
			$[T_{1}, T_{2}]$ & Time duration for which the billboards are operating \\
			\hline 
			 $\mathbb{B}$ & The Billboard Database \\
			 \hline
			 $a$ & The number of billboard slots, i.e., $|\mathcal{BS}|=a$ \\
			 \hline
	  $r$ & The number of billboards \\
			\hline  
			$\Delta$ & Slot duration \\
			\hline
			$\mathcal{BS}$ & The set of billboard slots \\
		
			\hline 
			$Pr(u|h)$ & Influence Probability of the person $u$ for the tag $h$
			\\
			\hline 
			$Pr(u|\mathcal{H}^{'})$ & Influence Probability of the person $u$ for the tags $\mathcal{H}^{'}$  \\
			\hline 
	$\mathcal{N}^{OPT}, \mathcal{H}^{OPT}$ & $k$-sized and $\ell$-sized optimal slot and tag subset \\
			\hline 
	$|\mathcal{X}|$ & Cardinality of the set $\mathcal{X}$	\\
			\hline 
	$[n]$ & The set $\{1, 2, \ldots, n\}$
			\\
			\hline
			$\mathbb{E}[X]$ & Expectation of the random variable $X$
			\\
			\hline
			\hline\\
		\end{tabular}%
		}
	\end{table}

\subsection{Randomized Greedy Algorithm}\label{Sec:R_Greedy}
Algorithm \ref{alg:Randomized_Greedy} presents a greedy approach to solving the non-bisubmodular maximization problem. It takes billboard slot, trajectory, and social network as input and returns an allocation of slots and seed nodes under the given budget. It starts with empty sets for billboard slots and seeds. First, billboard slots are sorted by their individual influence. Slots are greedily added until the temporary budget is exhausted. Similarly, seed nodes are sorted by individual influence. Seeds are added until another temporary budget is used. The smaller size of these two sets is used to guide sampling. In each iteration, a small random subset of billboard slots is sampled. A small random subset of seed nodes is also sampled. For each sampled element, the marginal gain per unit cost is computed. The interaction effect is included in this computation. The element with the higher gain-to-cost ratio is selected. The budget is updated after every selection. The selected element is removed from further consideration. This process continues until the budget is exhausted.

\SetKwComment{Comment}{/* }{ */}
%%%%%%%%%%%%%%%%%%%%%% Algorithm 1 Updated %%%%%%%%%%%%%%%%%%%%%%%%
\begin{algorithm}[!htb]
\scriptsize
\caption{ Randomized Greedy Algorithm for Multi-Channel Influence Maximization}
\label{alg:Randomized_Greedy}
\KwData{
Billboard slots $\mathbb{BS}$,
social graph $\mathcal{G}(\mathcal{U},\mathcal{E})$,
influence function $\Phi$,
costs $\mathcal{C}_{\mathbb{BS}}, \mathcal{C}_{\mathcal{G}}$,
budget $\mathcal{B}$
}
\KwResult{
Selected billboard set $\mathbb{S} \subseteq \mathbb{BS}$,
selected seed set $\mathcal{N} \subseteq \mathcal{U}$
}

\BlankLine
\textbf{Initialize:}
$\mathbb{S} \gets \emptyset$, $\mathcal{N} \gets \emptyset$,$\mathcal{B}_{rem} \gets \mathcal{B}$, $\mathcal{B}^{'} \gets \mathcal{B}$, $\mathcal{S}^{'} \gets \emptyset, c \gets 1$\;
$\mathcal{B}^{''} \gets \mathcal{B}$, $\mathcal{N}^{'} \gets \emptyset, p \gets 1$\;
$\mathbb{BS}^{'} \gets$ Sort slots in ascending order based on individual influence\;
\While{$\mathcal{B}^{'} > 0$ and $\mathbb{BS}^{'} \neq \emptyset$}{
$\mathcal{S}^{'} \gets \mathcal{S}^{'} \cup ~\mathbb{BS}^{'}[c]$\;
$\mathcal{B}^{'} \gets \mathcal{B}^{'} - \mathcal{C}_{\mathbb{BS}}(\mathbb{BS}^{'}[c])$\;
$\mathbb{BS}^{'} \gets \mathbb{BS}^{'} \setminus \mathbb{BS}^{'}[c]$\;
$c = c +1$\;
}
$\mathcal{U}^{'} \gets$ Sort seeds in ascending order based on individual influence\;
\While{$\mathcal{B}^{''} > 0$ and $\mathcal{U}^{'} \neq \emptyset$}{
$\mathcal{N}^{'} \gets \mathcal{N}^{'} \cup ~\mathcal{U}^{'}[p]$\;
$\mathcal{B}^{''} \gets \mathcal{B}^{''} - \mathcal{C}_{\mathcal{G}}(\mathcal{U}^{'}[p])$\;
$\mathcal{U}^{'} \gets \mathcal{U}^{'} \setminus \mathcal{U}^{'}[p]$\;
$p = p+1$\;
}
$k = min(|\mathcal{S}^{'}|, |\mathcal{N}^{'}|)$\;

\While{$\mathcal{B}_{rem} > 0$ \textbf{and}
$(\mathbb{BS} \neq \emptyset \textbf{ or } \mathcal{U} \neq \emptyset)$}{

\Comment{Compute best billboard candidate}
$\mathbb{BS}^{''} \gets$ Sample $\frac{|\mathbb{BS}|}{k} \log \frac{1}{\epsilon}$ many elements from $\mathbb{BS}$\;
$b^* \gets
\arg\max\limits_{b \in \mathbb{BS}^{''}}
\frac{\Phi(\mathbb{S} \cup \{b\}, \mathcal{N}) - \Phi(\mathbb{S}, \mathcal{N})}
{\mathcal{C}_{\mathbb{BS}}(b)}$\;

\Comment{Compute best social seed candidate}
$\mathcal{U}^{''} \gets$ Sample $\frac{|\mathcal{U}|}{k} \log \frac{1}{\epsilon}$ many elements from $\mathcal{U}$\;

$s^* \gets
\arg\max\limits_{s \in \mathcal{U}^{''}}
\frac{\Phi(\mathbb{S}, \mathcal{N} \cup \{s\}) - \Phi(\mathbb{S}, \mathcal{N})}
{\mathcal{C}_{\mathcal{G}}(s)}$\;

\Comment{Select the better candidate}
\eIf{
$\frac{\Delta(b^* \mid \mathbb{S},\mathcal{N})}{\mathcal{C}_{\mathbb{BS}}(b^*)}
\ge
\frac{\Delta(s^* \mid \mathbb{S},\mathcal{N})}{\mathcal{C}_{\mathcal{G}}(s^*)}$
}{
    \If{$\mathcal{C}_{\mathbb{BS}}(b^*) \le \mathcal{B}_{rem}$}{
        $\mathbb{S} \gets \mathbb{S} \cup \{b^*\}$\;
        $\mathcal{B}_{rem} \gets \mathcal{B}_{rem} - \mathcal{C}_{\mathbb{BS}}(b^*)$\;
    }
    $\mathbb{BS} \gets \mathbb{BS} \setminus \{b^*\}$\;
}{
    \If{$\mathcal{C}_{\mathcal{G}}(s^*) \le \mathcal{B}_{rem}$}{
        $\mathcal{N} \gets \mathcal{N} \cup \{s^*\}$\;
        $\mathcal{B}_{rem} \gets \mathcal{B}_{rem} - \mathcal{C}_{\mathcal{G}}(s^*)$\;
    }
    $\mathcal{U} \gets \mathcal{U} \setminus \{s^*\}$\;
}
}
\Return $\mathbb{S}, \mathcal{N}$
\end{algorithm}

\paragraph{\textbf{Complexity Analysis.}} Now, we analyze the time and space requirement of the Algorithm \ref{alg:Randomized_Greedy}. In Line No. $1$ and $2$, for initialization of variables will take $\mathcal{O}(1)$ time. In Line No. $3$, for $a$ number of billboard slots, sorting slots based on individual influence value will take $\mathcal{O}(a \log a)$.
Next, in Line No. $4$ to $8$ will take $\mathcal{O}(a)$ in the worst case. Similarly, Line No. $9$ to $14$ will take $\mathcal{O}(b \log b + b)$ time assuming $b$ number of seeds are there. In Line No. $15$ will take $\mathcal{O}(1)$ time. In Line No. $16$ \texttt{while loop} will execute for $\mathcal{O}(a+b)$ time and in Line No. $17$ sampling slots will take $\mathcal{O}(\frac{a}{k}\log \frac{1}{\epsilon})$ time and Line No. $18$ finding $b^{*}$ involve influence computation. For $a$ number of billboard slots, computing influence will take $\mathcal{O}(a \cdot m)$, where $m$ is the number of tuples in the trajectory database. For $b$ number of nodes and $y$ number of edges, computing influence under the IC model for a graph $\mathcal{G}$ will take $\mathcal{O}(R \cdot (b+y)) ~, i.e., ~\mathcal{O}(R \cdot y)$, where $R$ is the number of simulations. Similarly, calculating the interaction effect will take $\mathcal{O}(a \cdot m +\cdot b \cdot m)$. So, the combined influence function $\Phi()$ will take $\mathcal{O}(a \cdot m + b \cdot m + R \cdot y)$ time to execute. Similarly, Line No. $19$ and $20$ will take $\mathcal{O}(\frac{b}{k}\log \frac{1}{\epsilon})$ and  $\mathcal{O}(a \cdot m + b \cdot m + R \cdot y)$, respectively. Now, Line No. $21$ to $30$ will take $\mathcal{O}(a \cdot m + b \cdot m + R \cdot y)$ time to execute. Therefore, the total time requirement of Algorithm \ref{alg:Randomized_Greedy} will be of $\mathcal{O}((a+b)[\frac{a}{k}\log \frac{1}{\epsilon}+\frac{b}{k}\log \frac{1}{\epsilon} + a \cdot m +b \cdot m+r \cdot y])$. Next, the additional space requirement for Algorithm \ref{alg:Randomized_Greedy} will be of $\mathcal{O}(a+b)$ for storing slots and seed nodes.

\begin{theorem}
The time and space requirement for Algorithm \ref{alg:Randomized_Greedy} will be of $\mathcal{O}((a+b)[\frac{a}{k}\log \frac{1}{\epsilon}+\frac{b}{k}\log \frac{1}{\epsilon} + a \cdot m +b \cdot m+r \cdot y])$ and $\mathcal{O}(a+b)$, respectively.
\end{theorem}

\subsection{Two-Phase Adaptive Greedy with Lazy Evaluation (TPG)}

Algorithm \ref{Alg:TAG-L} selects billboard slots and social network seeds jointly. It works under a fixed budget. The algorithm runs in two phases. In the first phase, one billboard slot is selected first. It has the highest influence per unit cost. One social network seed is also selected. It has the highest influence per unit cost. If both fit the budget, both are chosen.
Otherwise, only the better one is selected. This phase activates the interaction effect early. In the second phase, all remaining candidates are inserted into a priority queue. Each candidate is keyed by an upper bound on marginal gain per cost. At each iteration, the top candidate is examined. Its true marginal gain is computed lazily. If the gain matches the upper bound, it is selected. If not, the bound is updated and reinserted. Billboard slots and seed nodes are treated uniformly. Only budget-feasible candidates are added. The process continues until the budget is exhausted or no candidates remain.

\begin{algorithm}[h!]
\scriptsize
\caption{Two-Phase Adaptive Greedy (TPG) for Multi-Channel Influence Maximization}
\label{Alg:TAG-L}
\KwIn{Billboard slots $\mathbb{BS}$, social graph $\mathcal{G}$, influence function $\Phi$, costs $\mathcal{C}_\mathcal{B}, \mathcal{C}_\mathcal{G}$, budget $\mathcal{B}$}
\KwOut{Selected slots $\mathbb{S}$, seed nodes $\mathcal{N}$}

Initialize $\mathbb{S} \leftarrow \emptyset$, $\mathcal{N} \leftarrow \emptyset$, $\mathcal{B}_{rem} \leftarrow \mathcal{B}$

\textbf{Phase I: Balanced Initialization}

Select $b_0 = \arg\max_{b \in \mathbb{BS}} \frac{\Phi(\{b\}, \emptyset)}{\mathcal{C}_\mathcal{B}(b)}$ \\
Select $s_0 = \arg\max_{s \in V(G)} \frac{\Phi(\emptyset, \{s\})}{\mathcal{C}_\mathcal{G}(s)}$

\eIf{$\mathcal{C}_\mathcal{B}(b_0) + \mathcal{C}_\mathcal{G}(s_0) \le \mathcal{B}_{rem}$}{
    $\mathbb{S} \leftarrow \{b_0\}$, $\mathcal{N} \leftarrow \{s_0\}$ \\
    $\mathcal{B}_{rem} \leftarrow \mathcal{B}_{rem} - \mathcal{C}_\mathcal{B}(b_0) - \mathcal{C}_\mathcal{G}(s_0)$
}{
    Select the better of $b_0$ or $s_0$ alone
}

\textbf{Phase II: Lazy Interaction-Aware Greedy}

Initialize priority queue $\mathcal{Q}$ with all remaining $b \in \mathbb{BS}$ and $s \in V(\mathcal{G})$ \\
Each element keyed by upper bound on marginal gain per cost

\While{$\mathcal{B}_{rem} > 0$ and $\mathcal{Q} \neq \emptyset$}{
    Pop top element $e$ from $\mathcal{Q}$ \\
    $Bound_{upper} \gets UB(top)$\;
    \If{$e$ is a slot and $\mathcal{C}_{\mathcal{B}}(e) \leq \mathcal{B}_{rem}$}{
    $Gain_{true} =  \frac{\Phi(\mathbb{S} \cup \{e\}, ~\mathcal{N}) - \Phi(\mathbb{S}, \mathcal{N})}{\mathcal{C}_{\mathcal{B}}(e)}$\; 
    \If{$Gain_{true} \geq Bound_{upper}$}{
    $\mathbb{S} \gets \mathbb{S} \cup \{e\}$\;
    $\mathcal{B}_{rem} \leftarrow \mathcal{B}_{rem} - \mathcal{C}_{\mathcal{B}}(e)$\;
    % Update Q\;
    }}

    \ElseIf{$e$ is a seed and $\mathcal{C}_{\mathcal{G}}(e) \leq \mathcal{B}_{rem}$}{

    $Gain_{true} =  \frac{\Phi(\mathbb{S}, \mathcal{N} \cup \{e\}) - \Phi(\mathbb{S}, \mathcal{N})}{\mathcal{C}_{\mathcal{G}}(e)}$\; 
    \If{$Gain_{true} \geq Bound_{upper}$}{
    $\mathcal{N} \gets \mathcal{N} \cup \{e\}$\;
    $\mathcal{B}_{rem} \leftarrow \mathcal{B}_{rem} - \mathcal{C}_{\mathcal{G}}(e)$\;
    }}
    \Else{
        Update key of $e$ and reinsert into $\mathcal{Q}$
    }
}

\Return $\mathbb{S}, \mathcal{N}$
\end{algorithm}

\paragraph{\textbf{Complexity Analysis.}}
Now, we analyze the time and space requirements of Algorithm \ref{Alg:TAG-L}. In Line No. $1$ initialization will take $\mathcal{O}(1)$ time. In Line No. $3$ to $9$, the first phase will take $\mathcal{O}(a+b)$ time, where $a$ is the number of slots and $b$ is the number of seeds. Now, Line No. From $11$ to $12$, all remaining slots and seeds are inserted into the queue, which will take $\mathcal{O}((a+b) \log (a+b))$. Assume the \texttt{while loop} will run for $\mathcal{O}(L)$ times. The combined influence function $\Phi()$ will take $\mathcal{O}(a \cdot m + b \cdot m + R \cdot y)$ time to execute as discussed in Algorithm \ref{alg:Randomized_Greedy}. Hence, Line No. $14$ to $27$ will take $\mathcal{O}((a+b) \log (a+b)+ a\cdot m + b \cdot m + R \cdot y)$ time to execute. So, total time taken by Algorithm \ref{Alg:TAG-L} will be $\mathcal{O}(a+b) + \mathcal{O}((a+b) \log (a+b))+ \mathcal{O}((a+b) \log (a+b)+ a\cdot m + b \cdot m + R \cdot y)$ i.e., $\mathcal{O}((a+b) \log (a+b))+ \mathcal{O}((a+b) \log (a+b)+ L \cdot (a\cdot m + b \cdot m + R \cdot y)$. The additional space requirement for Algorithm \ref{Alg:TAG-L} will be of $\mathcal{O}(a+b)$.

\begin{theorem}
The time and space requirement for Algorithm \ref{Alg:TAG-L} will be of $\mathcal{O}((a+b) \log (a+b))+ \mathcal{O}((a+b) \log (a+b)+ L \cdot (a\cdot m + b \cdot m + R \cdot y)$ and $\mathcal{O}(a+b)$, respectively.
\end{theorem}

%%%%%%%%%%%%%%%%%%%%%%%%%%%%%%%%%%%%%%%%%%%%%%%%%%%%%%%%%%%%%%%%%%%%%%%%
\section{Approximation Guarantee}\label{Sec:Approx_Guarantee}
 Algorithm \ref{alg:Randomized_Greedy} provides an approximation guarantee of at least $\frac{1}{\alpha} \left(1 - e^{-\gamma \alpha} \right) $ times the optimal solution. This claim is stated in Theorem \ref{Theorem:Greedy_Approximation}. We first establish key theoretical results, beginning with a marginal gain lower bound using the bisubmodularity ratio in Lemma~\ref{Lemma:Lower_Bound}.

\begin{lemma}\label{Lemma:Lower_Bound}
Let $\Phi: 2^{\mathbb{BS}} \times 2^{V(\mathcal{G})} \to \mathbb{R}_{\geq 0}$ be a non-negative, monotone function with bisubmodularity ratio $\gamma \in (0,1]$, and let $(\mathcal{S}_t, \mathcal{N}_t)$ be the current selection of slots and seed nodes. Then, for any feasible solution $(\mathcal{S}^*, \mathcal{N}^*)$, then
{\small
\begin{align*}
&\sum_{e \in (\mathcal{S}^* \cup \mathcal{N}^*) \setminus (\mathcal{S}_t \cup \mathcal{N}_t)} \Delta(e \mid \mathcal{S}_t, \mathcal{N}_t) \geq \gamma \cdot \left[ \Phi(\mathcal{S}^*, \mathcal{N}^*) - \Phi(\mathcal{S}_t, \mathcal{N}_t) \right],
\end{align*}
}
where $\Delta(e \mid \mathcal{S}_t, \mathcal{N}_t)$ denotes the marginal gain of adding element $e$ to $(\mathcal{S}_t, \mathcal{N}_t)$.
\end{lemma}

\begin{proof}
Let $\Omega_B = \mathcal{S}^* \setminus \mathcal{S}_t$ and $\Omega_G = \mathcal{N}^* \setminus \mathcal{N}_t$. Let $q = \mathcal{S}_t$ and $\mathcal{N} = \mathcal{N}_t$. By Definition~\ref{Def:BR}, we have the following:
\begin{align*}
\sum_{v \in \Omega_B} & \left[ \Phi(q \cup \{v\}, \mathcal{N}) - \Phi(q, \mathcal{N}) \right] \geq \gamma \cdot \left[ \Phi(q \cup \Omega_B, \mathcal{N}) - \Phi(q, \mathcal{N}) \right], \\
\sum_{v \in \Omega_G} & \left[ \Phi(q, \mathcal{N} \cup \{v\}) - \Phi(q, \mathcal{N}) \right] \geq \gamma \cdot \left[ \Phi(q, \mathcal{N} \cup \Omega_G) - \Phi(q, \mathcal{N}) \right].
\end{align*}

Now, adding the inequalities we have,
\begin{align*}
\sum_{e \in \Omega_B \cup \Omega_G} & \Delta(e \mid \mathcal{S}_t, \mathcal{N}_t)  \geq \gamma \cdot \left[ \Phi(q \cup \Omega_B, \mathcal{N} \cup \Omega_G) - \Phi(q, \mathcal{N}) \right] \\
& = \gamma \cdot \left[ \Phi(\mathcal{S}_t \cup \mathcal{S}^*, 
\mathcal{N}_t \cup \mathcal{N}^*) - \Phi(\mathcal{S}_t, \mathcal{N}_t) \right].
\end{align*}

By monotonicity of influence function $\Phi$ we can write,
\[
\Phi(\mathcal{S}_t \cup \mathcal{S}^*, \mathcal{N}_t \cup \mathcal{N}^*) 
\geq \Phi(\mathcal{S}^*, \mathcal{N}^*).
\]

Hence,
\begin{align*}
\sum_{e \in (\mathcal{S}^* \cup \mathcal{N}^*) \setminus (\mathcal{S}_t \cup \mathcal{N}_t)} & \Delta(e \mid \mathcal{S}_t, \mathcal{N}_t) \geq \gamma \cdot \left[ \Phi(\mathcal{S}^*, \mathcal{N}^*) - 
\Phi(\mathcal{S}_t, \mathcal{N}_t) \right]. \qedhere
\end{align*}
\end{proof}

\begin{lemma}\label{Lemma:Curvature_Bound}
Let $\Phi: 2^{\mathbb{BS}} \times 2^{V(\mathcal{G})} \to \mathbb{R}_{\geq 0}$ be a non-negative, monotone function with generalized curvature $\alpha \in [0,1]$. Then, for any $e \notin \mathcal{S}_t \cup \mathcal{N}_t$, the marginal gain of adding $e$ to the current selection $(\mathcal{S}_t, \mathcal{N}_t)$ is lower bounded by:
{\small
\[
\Delta(e \mid \mathcal{S}_t, \mathcal{N}_t) \geq (1 - \alpha) \cdot \Delta(e \mid \emptyset, \emptyset),
\]
}
where $\Delta(e \mid \mathcal{S}_t, \mathcal{N}_t)$ is the marginal gain of adding element $e$ to $(\mathcal{S}_t, \mathcal{N}_t)$.
\end{lemma}

\begin{proof}
By Definition~\ref{Def:generalized_curvature}, for any $e \in \mathbb{BS} \cup V(\mathcal{G})$ and any $\mathcal{S} \subseteq \mathbb{BS}$, $\mathcal{N} \subseteq V(\mathcal{G})$, the marginal gain of $e$ when added to a superset is at least 
$(1 - \alpha)$ times the marginal gain of $e$ when added to the empty set. If $e \in \mathbb{BS}$, then we can write
\begin{align}
&\Phi(\mathcal{S}_t \cup \{e\}, \mathcal{N}_t) - \Phi(\mathcal{S}_t, \mathcal{N}_t) \geq (1 - \alpha) \cdot \left[ \Phi(\{e\}, \mathcal{N}_t) - \Phi(\emptyset, \mathcal{N}_t) \right] \geq (1 - \alpha) \cdot \Delta(e \mid \emptyset, \emptyset),
\end{align}

If $e \in V(\mathcal{G})$, then we have
\begin{align}
&\Phi(\mathcal{S}_t, \mathcal{N}_t \cup \{e\}) - \Phi(\mathcal{S}_t, \mathcal{N}_t) \geq (1 - \alpha) \cdot \left[ \Phi(\mathcal{S}_t, \{e\}) - \Phi(\mathcal{S}_t, \emptyset) \right] \geq (1 - \alpha) \cdot \Delta(e \mid \emptyset, \emptyset),
\end{align}

Since $\Phi(\{e\}, \emptyset) \geq 0$ and $\Phi(\emptyset, \{e\}) \geq 0$ by non-negativity. In both cases, we can say
\[
\Delta(e \mid \mathcal{S}_t, \mathcal{N}_t) \geq (1 - \alpha) \cdot \Delta(e \mid \emptyset, \emptyset).
\]
\end{proof}

\begin{theorem}\label{Theorem:Greedy_Approximation}
Let $\Phi: 2^{\mathbb{BS}} \times 2^{V(\mathcal{G})} \to \mathbb{R}_{\geq 0}$ be a non-negative, monotone function with bisubmodularity ratio $\gamma \in (0,1]$ and generalized curvature $\alpha \in [0,1]$. Let $(\mathcal{S}^*, \mathcal{N}^*)$ be the optimal solution satisfying the budget constraint $C(\mathcal{S}^*) + C(\mathcal{N}^*) \leq B$, and let $(\mathcal{S}_G, \mathcal{N}_G)$ be the solution returned by Algorithm~\ref{alg:Randomized_Greedy}. Then,
{\small
\[
\mathbb{E}\!\left[\Phi(\mathcal{S}_G, \mathcal{N}_G)\right]
\;\geq\;
\frac{1}{\alpha} \left( 1 - e^{-\gamma \alpha} \right)
\cdot \Phi(\mathcal{S}^*, \mathcal{N}^*).
\]
}
\end{theorem}
\begin{proof}

We analyze the behavior of Algorithm~\ref{alg:Randomized_Greedy} step by step.
Let $(\mathcal{S}_t,\mathcal{N}_t)$ denote the solution obtained after $t$ iterations,
with $(\mathcal{S}_0,\mathcal{N}_0)=(\emptyset,\emptyset)$.
At each iteration, the algorithm samples a small subset of billboard slots and seed nodes
and selects the element that provides the maximum marginal gain per unit cost within the
sample, subject to the remaining budget.

Let $\Phi_t := \Phi(\mathcal{S}_t,\mathcal{N}_t)$ denote the influence after $t$ iterations,
and let $\Delta_t = \Phi_t - \Phi_{t-1}$ be the marginal increase in influence at step $t$.
Throughout the process, budget feasibility is maintained, i.e.,
$C(\mathcal{S}_t)+C(\mathcal{N}_t)\le B$.

Let $(\mathcal{S}^*,\mathcal{N}^*)$ be an optimal solution satisfying the budget constraint,
and define
\[
R_t = (\mathcal{S}^*\cup\mathcal{N}^*) \setminus (\mathcal{S}_{t-1}\cup\mathcal{N}_{t-1})
\]
as the set of optimal elements not yet selected by the algorithm.

By the definition of the bisubmodularity ratio (Lemma~\ref{Lemma:Lower_Bound}),
the total marginal gain contributed by elements in $R_t$ is at least a
$\gamma$ fraction of the remaining optimal influence. Formally,
\[
\sum_{e\in R_t} \Delta(e\mid \mathcal{S}_{t-1},\mathcal{N}_{t-1})
\;\ge\;
\gamma\bigl[\Phi(\mathcal{S}^*,\mathcal{N}^*)-\Phi_{t-1}\bigr].
\]
This implies that there exists at least one element in $R_t$ whose marginal gain per unit
cost is sufficiently large.

Algorithm~\ref{alg:Randomized_Greedy} does not examine all remaining elements.
Instead, it selects the best element from a randomly sampled subset.
Standard results on randomized greedy selection ensure that, in expectation,
the chosen element achieves a marginal gain comparable to such a good element.
Moreover, by the curvature bound (Lemma~\ref{Lemma:Curvature_Bound}),
the marginal gain of any element does not decrease too rapidly as the solution grows.

Combining these observations, the expected marginal gain at iteration $t$ satisfies
\[
\mathbb{E}[\Delta_t]
\;\ge\;
\frac{\gamma(1-\alpha)}{B}
\cdot C(e_t)
\cdot
\bigl[\Phi(\mathcal{S}^*,\mathcal{N}^*)-\Phi_{t-1}\bigr].
\]

Let $b_t=\sum_{i=1}^t C(e_i)$ be the total budget used after $t$ iterations, and define
$f(b_t)=\mathbb{E}[\Phi_t]$.
The above inequality yields the recurrence
\[
f(b_t)-f(b_{t-1})
\;\ge\;
\frac{\gamma(1-\alpha)}{B}
\cdot (b_t-b_{t-1})
\cdot
\bigl[\Phi(\mathcal{S}^*,\mathcal{N}^*)-f(b_{t-1})\bigr].
\]

Approximating this recurrence by a continuous process leads to the differential inequality
\[
\frac{df(b)}{db}
\;\ge\;
\frac{\gamma(1-\alpha)}{B}
\cdot
\bigl[\Phi(\mathcal{S}^*,\mathcal{N}^*)-f(b)\bigr],
\]
whose solution implies
\[
f(b)
\;\ge\;
\Phi(\mathcal{S}^*,\mathcal{N}^*)
\left(1-\exp\!\left(-\frac{\gamma(1-\alpha)}{B}b\right)\right).
\]

At termination, the algorithm uses total budget $b\le B$.
Therefore,
\[
\mathbb{E}[\Phi(\mathcal{S}_G,\mathcal{N}_G)]
\;\ge\;
\Phi(\mathcal{S}^*,\mathcal{N}^*)
\bigl(1-e^{-\gamma(1-\alpha)}\bigr).
\]
Finally, using the inequality
$(1-\alpha)\ge \frac{1}{\alpha}(1-e^{-\gamma\alpha})$ for $\alpha\in(0,1]$,
we obtain
\[
\mathbb{E}[\Phi(\mathcal{S}_G,\mathcal{N}_G)]
\;\ge\;
\frac{1}{\alpha}
\left(1-e^{-\gamma\alpha}\right)
\cdot
\Phi(\mathcal{S}^*,\mathcal{N}^*).
\]
\end{proof}

% \begin{figure}[t]
%     \centering
%     \includegraphics[width=0.9\linewidth]{greedy_approximation_600.png}
%     \caption{Approximation guarantee of the greedy algorithm as a function of the bisubmodularity ratio $\gamma$ and curvature $\alpha$. The approximation ratio improves with higher $\gamma$ and lower $\alpha$, approaching $(1 - \frac{1}{e})$ in the ideal case.}
%     \label{fig:greedy_approx_surface}
% \end{figure}
\subparagraph{Interpreting Theorem \ref{Theorem:Greedy_Approximation}.}
Algorithm \ref{alg:Randomized_Greedy} achieves a guaranteed fraction $\frac{1}{\alpha} \left(1 - e^{-\gamma \alpha} \right)$ of the optimal influence, where $\gamma$ is the bisubmodularity ratio and $\alpha$ is the curvature. This means that even when the influence function is not perfectly submodular, greedy still performs well, especially when $\gamma$ is high and $\alpha$ is low. An interesting phenomenon is that $\gamma$ and $\alpha$ play different roles: (a). When both $\gamma$ and $\alpha$ are close to $1$, the approximation ratio approaches its best-case value $(1 - \frac{1}{e}) \approx 0.632$, matching classical submodular optimization results. (b). Low $\gamma$ (high non-submodularity) or high $\alpha$ (strong curvature) significantly degrade performance. (c). The surface is monotonically increasing with respect to both $\gamma$ and $\alpha^{-1}$, reflecting the intuitive fact that lower curvature and higher bisubmodularity yield better guarantees.

\begin{corollary}\label{Corollary:Greedy_WorstCase}
Under the assumptions of Theorem~\ref{Theorem:Greedy_Approximation}, if the curvature of $\Phi$ is $\alpha = 1$, then the Algorithm \ref{alg:Randomized_Greedy} achieves the following guarantee
{\small
\[
\Phi(\mathcal{S}_G, \mathcal{N}_G) \geq \left( 1 - e^{-\gamma} \right) \cdot \Phi(\mathcal{S}^*, \mathcal{N}^*).
\]
}
\end{corollary}

\begin{proof}
Substituting $\alpha = 1$ into Theorem~\ref{Theorem:Greedy_Approximation}, we get:
\begin{equation}
\begin{split}
\Phi(\mathcal{S}_G, \mathcal{N}_G) &\geq \frac{1}{1} \left(1 - e^{-\gamma \cdot 1} \right) \cdot \Phi(\mathcal{S}^*, \mathcal{N}^*) = (1 - e^{-\gamma}) \cdot \Phi(\mathcal{S}^*, \mathcal{N}^*).
\end{split}
\end{equation}
\end{proof}

\begin{theorem}\label{thm:tagl}
Let $\Phi : 2^{BS} \times 2^{V(G)} \rightarrow \mathbb{R}_{\ge 0}$ be a non-negative,
monotone influence function with bisubmodularity ratio $\gamma \in (0,1]$
and generalized curvature $\alpha \in [0,1]$.
Let $(S^{\star}, N^{\star})$ denote an optimal solution satisfying the budget
constraint $C(S^{\star}) + C(N^{\star}) \le B$.

Let $(S_{\mathrm{TAG}}, N_{\mathrm{TAG}})$ be the solution returned by
Algorithm~\ref{Alg:TAG-L}. Then,
\[
\Phi(S_{\mathrm{TAG}}, N_{\mathrm{TAG}}) \ge
\frac{1}{\alpha}\left(1 - e^{-\gamma \alpha}\right) \cdot \Phi(S^{\star}, N^{\star}).
\]
\end{theorem}

\begin{proof}[Proof Sketch]
Algorithm~\ref{Alg:TAG-L} proceeds in two phases. In Phase~I, at least one billboard slot
and one social network seed node are selected whenever feasible, ensuring
that the interaction effect $\Psi(S,N)$ becomes active from the beginning.
This initialization step consumes a constant fraction of the budget and does
not affect asymptotic approximation guarantees.

In Phase~II, TAG-L performs a cost-aware greedy selection using lazy marginal
evaluation. At each iteration, the algorithm selects an element with the
maximum marginal gain per unit cost. Since $\Phi$ is monotone with
bisubmodularity ratio $\gamma$, the total marginal gain of the remaining
optimal elements is lower bounded by a $\gamma$ fraction of the residual
optimal influence. Moreover, by the generalized curvature $\alpha$, the
marginal gain of any element at a later stage is at least $(1-\alpha)$ times
its marginal gain at the empty set.

Combining these properties yields a standard greedy recurrence, which leads to
the bound
\[
\Phi(S_{\mathrm{TAG}}, N_{\mathrm{TAG}}) \ge \frac{1}{\alpha} \left(1 - e^{-\gamma \alpha}\right) \Phi(S^{\star}, N^{\star}).
\]
The use of lazy evaluation affects only computational efficiency and does not
alter the sequence of selected elements. Hence, the approximation guarantee remains unchanged.
\end{proof}

\section{Experimental Evaluation} \label{Sec:EE}
\subsection{Dataset Description and Setup}
For experimental evaluation, we used three datasets: Trajectory, Social Network, and Billboard. The trajectory and social network datasets used in our study have also been used by existing studies \cite{10.1145/3308558.3313635,9099985}. All the datasets are summarized in Table \ref{tab:dataset_description}. Next, since outdoor advertising companies such as LAMAR and JCDecaux do not disclose actual billboard rental prices, prior work assumes that slot cost is proportional to influence. Following this convention, we set the cost of a billboard slot $bs$ as $\left\lfloor \delta \cdot \frac{\mathcal{I}(bs)}{10} \right\rfloor$, where $\delta \in [0.8, 1.1]$ is a scaling factor. Next, we evaluate three influence probability models: \emph{Uniform}, \emph{Weighted Cascade}, and \emph{Trivalency}. In the uniform model, all edges have probability \(p_c = 0.1\); in the weighted cascade model, an edge \((u_i,u_j)\) has probability \(1/\deg(u_i)\); and in the trivalency model, edge probabilities are sampled uniformly from \(\{0.1, 0.01, 0.001\}\). For \textbf{user selection cost}, we adopt the degree-proportional cost model \cite{sharma2024minimizing}, where the cost of selecting a user \(u\) is
\(\mathcal{C}(u) = k \cdot (|\mathcal{V}| / \sum_{v \in \mathcal{V}} \deg(v)) \cdot \deg(u)\). We set \(k = 1000\). All experiments were run on an HP Z4 workstation with an Xeon 3.50 GHz CPU and 64 GB RAM. 

% \begin{table}[h!]
% \centering
% \caption{Summary of Datasets Used in Experiments}
% \label{tab:dataset_description}
% \small
% \begin{tabular}{@{}lp{8.2cm}@{}}
% \toprule
% \textbf{Dataset} & \textbf{Description} \\
% \midrule
% Trajectory Dataset &
% Global-scale Foursquare check-in data collected over 22 months (April 2012--January 2014). After filtering only U.S. check-ins and merging two datasets, the final dataset contains 124,539 check-ins from 51,318 unique users. Similarly, the Canada Dataset contains 2,10,650 check-ins with 43,560 unique users. \\[0.4em]

% Social Network Dataset &
% Two snapshots of user friendship networks (before and after the check-in period).
% After filtering for U.S. users, the dataset includes  129,864 friendships
% corresponding to the trajectory users. Similarly Canada Dataset contains 50,538 friendships. \\[0.4em]

% Billboard Dataset &
% Crawled from the LAMAR inventory.
% Includes 716 billboards (1,031,040 slots) from New York City and
% 1,483 billboards (2,135,520 slots) from Los Angeles.
% The combined U.S. dataset contains 3,166,560 billboard slots. In Canada, the Dataset contains 1230 billboards, i.e., 1771200 slots. \\[0.4em]
% \bottomrule
% \end{tabular}
% \end{table}

\begin{table}[t]
\centering
\caption{Summary of Datasets Used in Experiments}
\label{tab:dataset_description}
\small
\begin{tabular}{@{}lp{7.6cm}@{}}
\toprule
\textbf{Dataset} & \textbf{Statistics (US / Canada)} \\
\midrule
Trajectory Dataset &
124,539 check-ins, 51,318 users (US);
210,650 check-ins, 43,560 users (Canada). \\[0.3em]

Social Network Dataset &
129,864 friendships (US);
50,538 friendships (Canada). \\[0.3em]

Billboard Dataset &
3,166,560 billboard slots from 2199 billboards (US);
1,771,200 billboard slots from 1,230 billboards (Canada).\cite{ali2022influential,zhang2020towards,ali2023influential} \\
\bottomrule
\end{tabular}
\end{table}

%%%%%%%%%%%%%%%%%%%%%%%%%%%%%%%%% Algorithm Vs. Influence %%%%%%%%%%%%%%%%%%%%%%%%%%%%%%%%%%%%%%%%
\begin{figure*}[!ht]
\centering
\setlength{\tabcolsep}{0.1pt} % tighter spacing between columns
\renewcommand{\arraystretch}{0.9} % tighter spacing between rows
\begin{tabular}{cccc}
\includegraphics[scale=0.13]{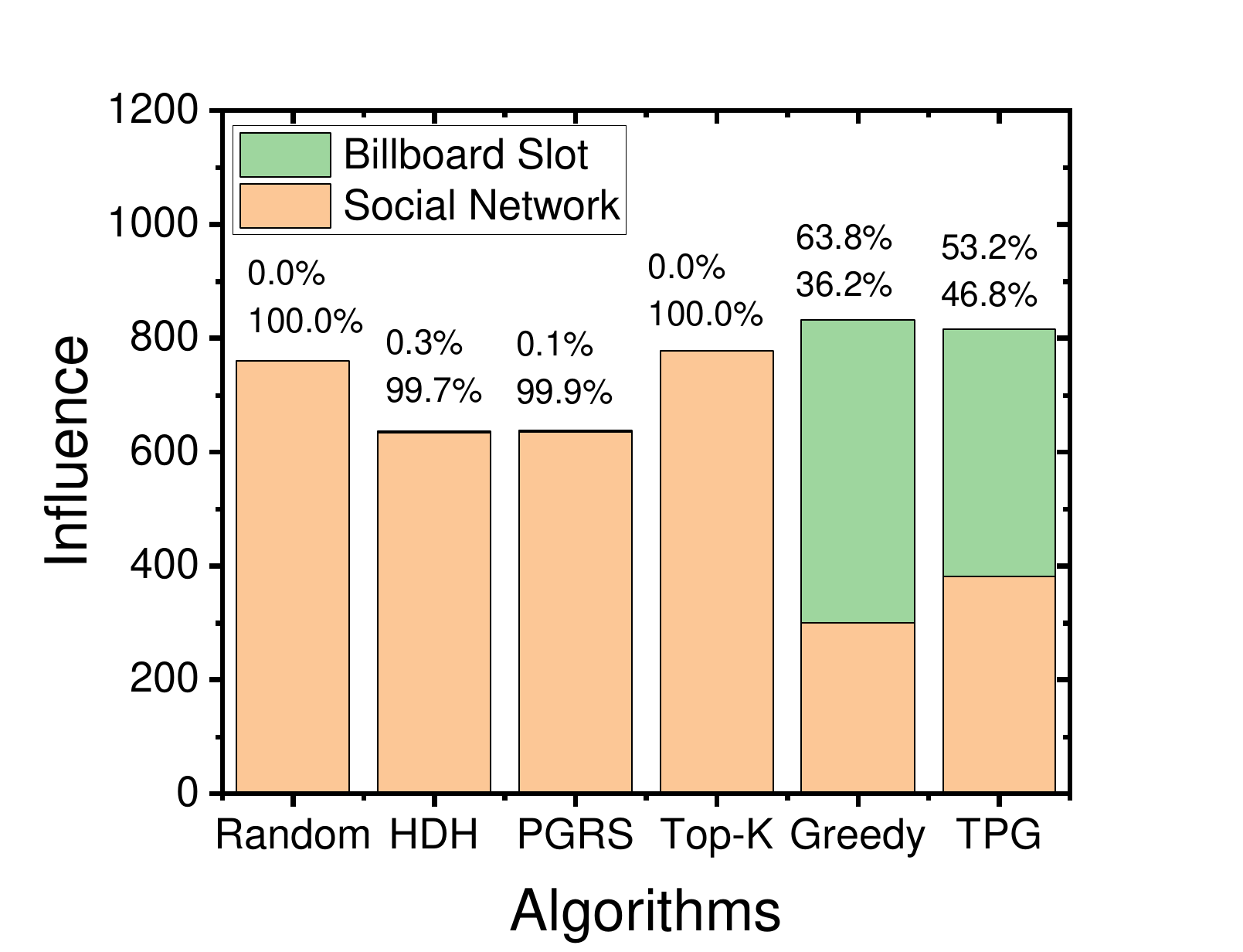} & \includegraphics[scale=0.13]{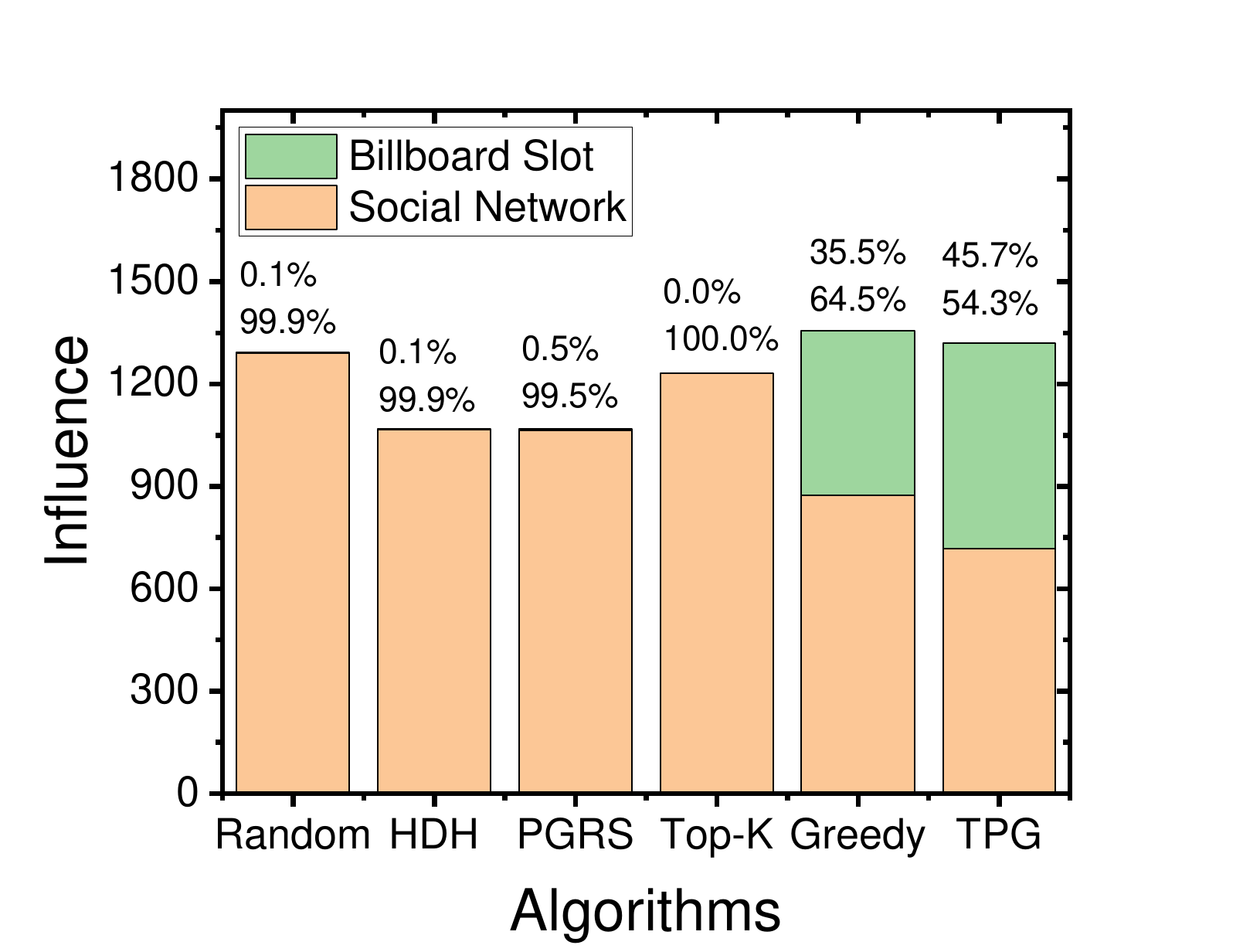} & \includegraphics[scale=0.13]{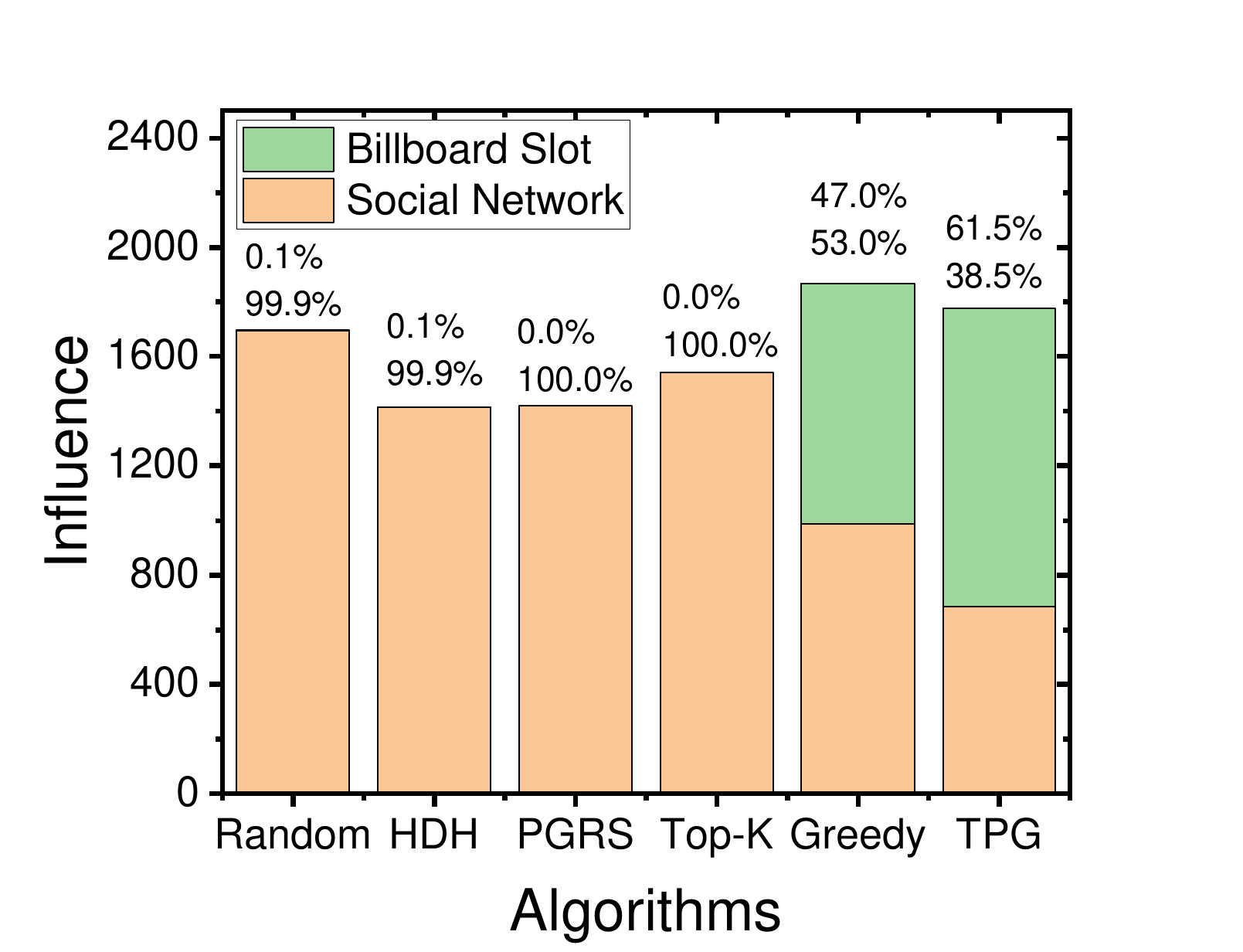} &
\includegraphics[scale=0.13]{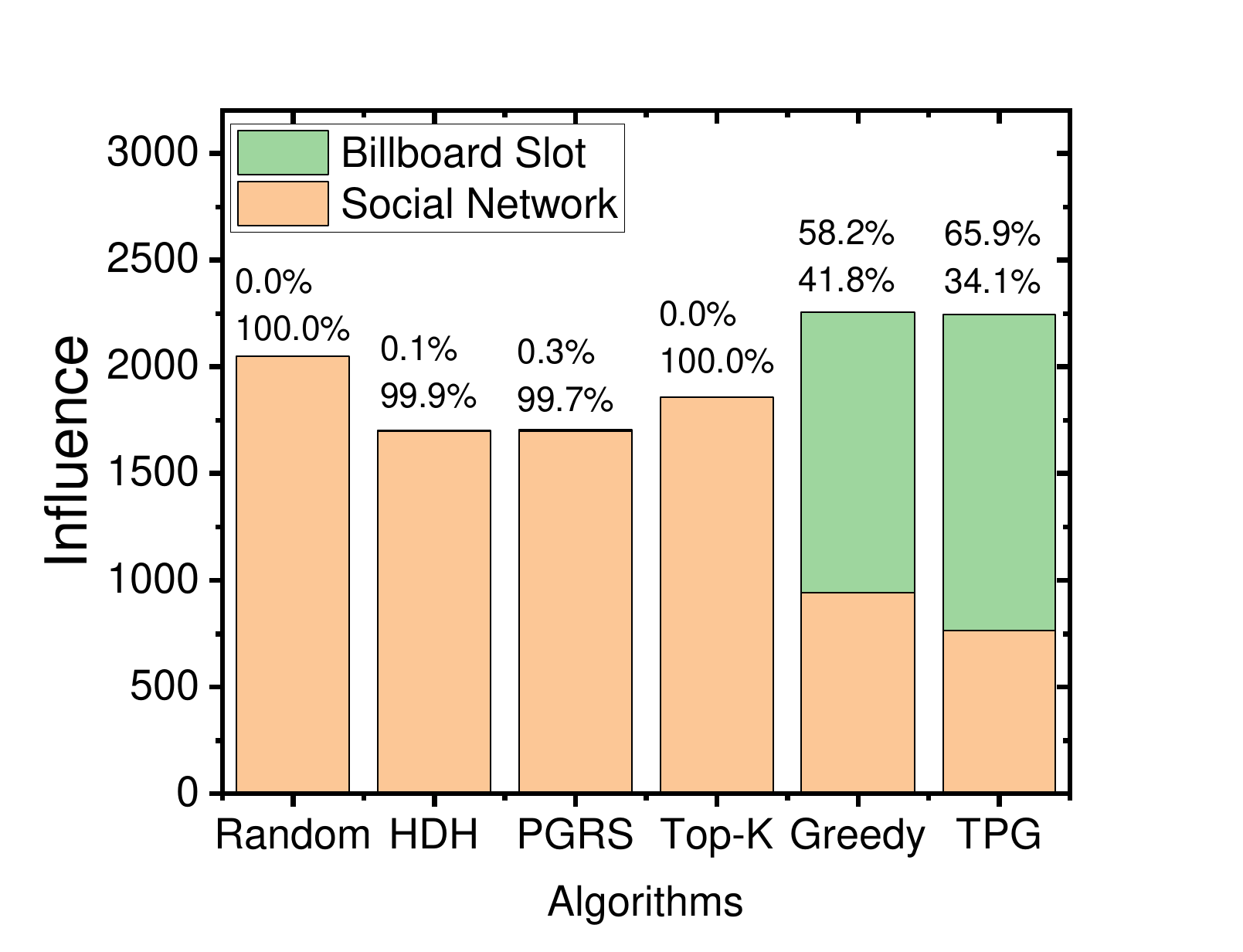} \\
\tiny{(a) Budget = $500$} &  \tiny{(b) Budget = $1000$} & \tiny{(c) Budget = $1500$} & \tiny{(d) Budget = $2000$} \\
\includegraphics[scale=0.13]{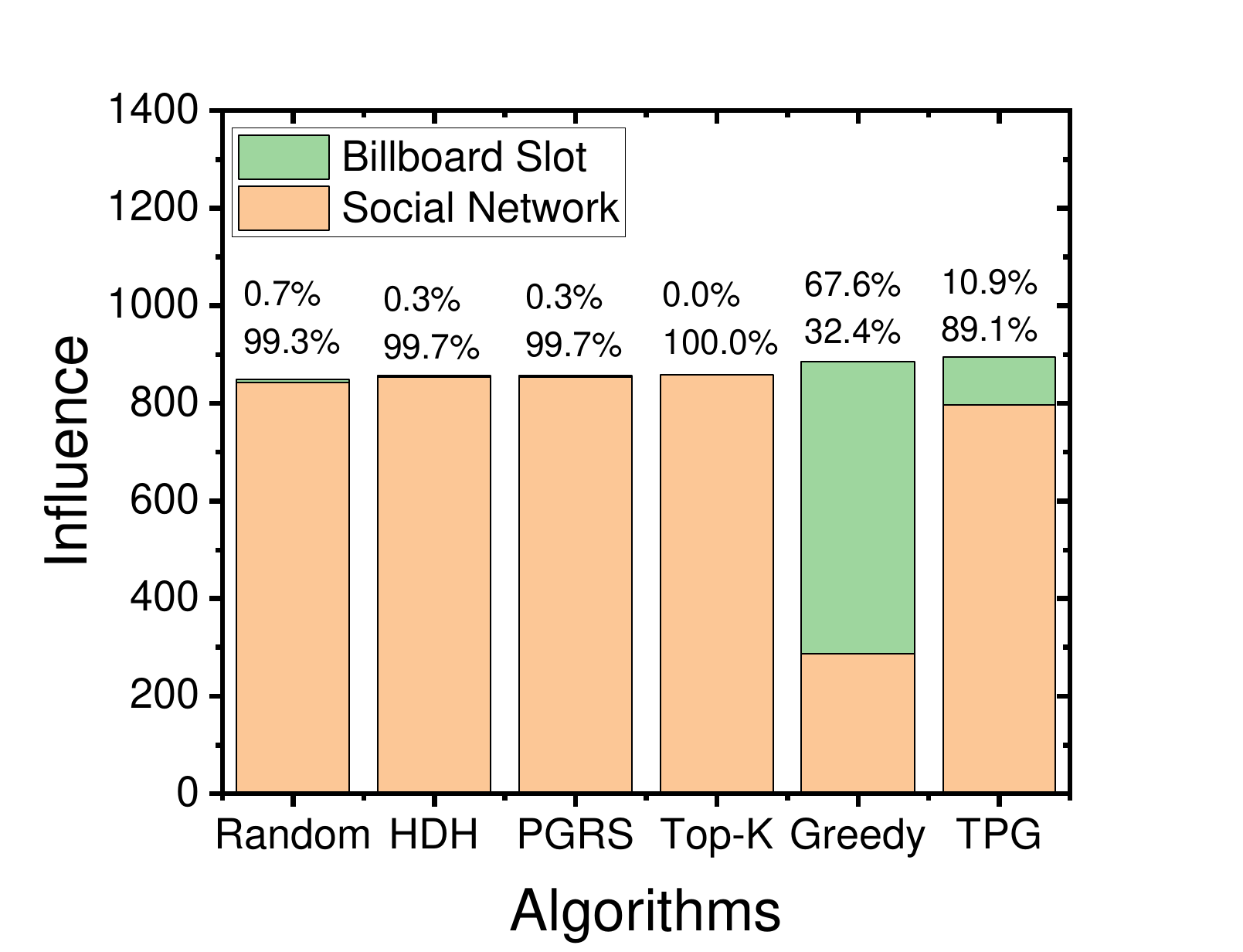} & \includegraphics[scale=0.13]{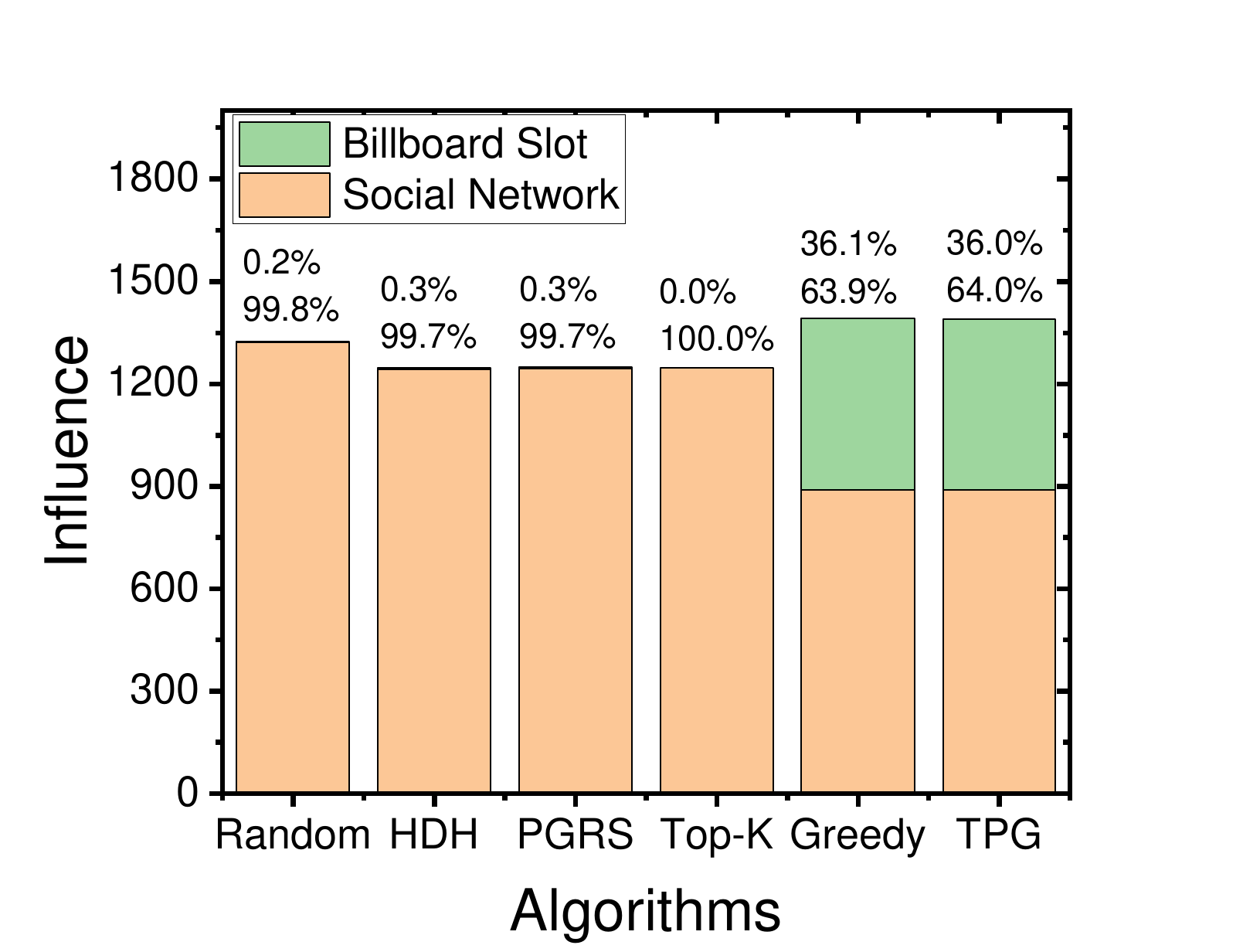} &
\includegraphics[scale=0.13]{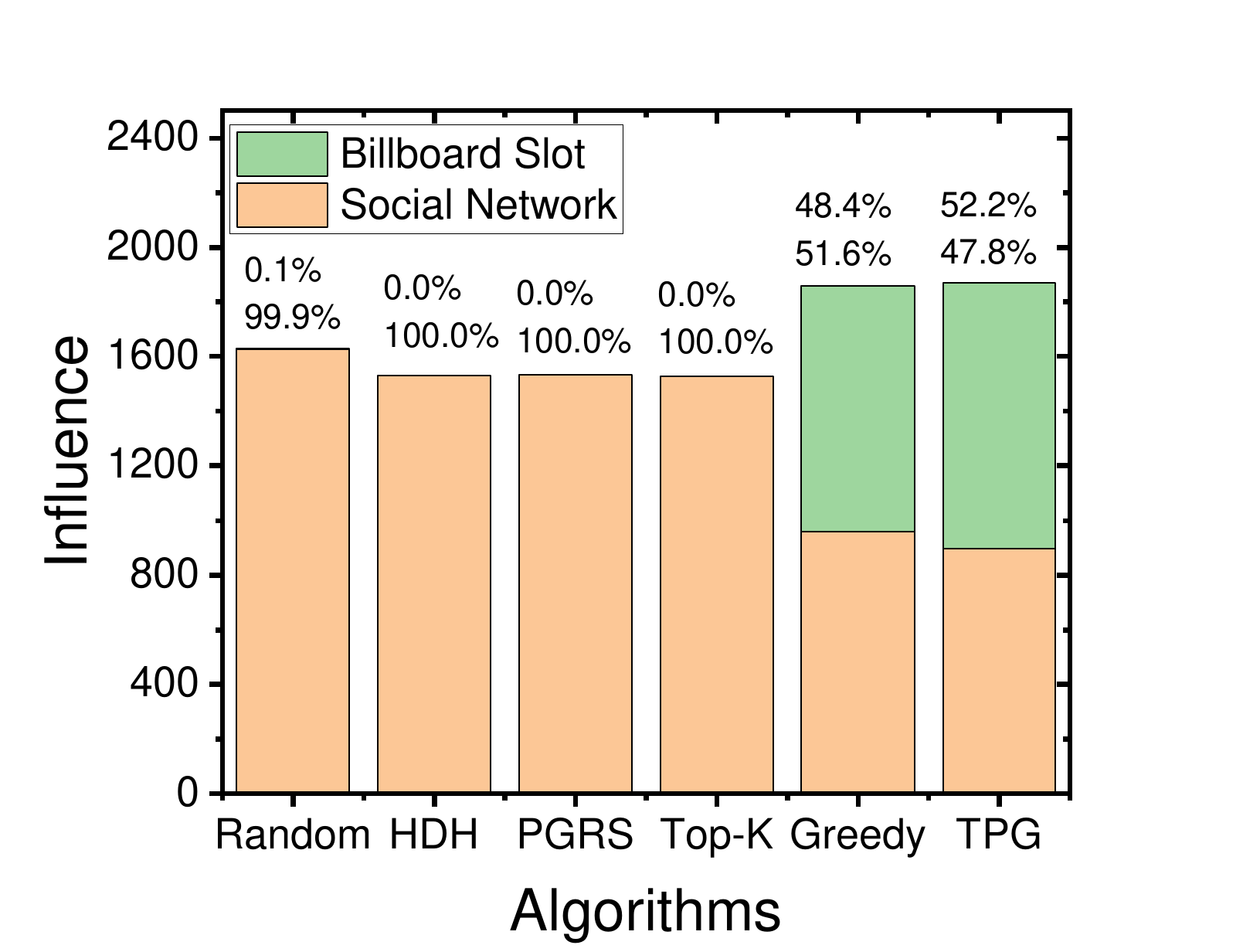} & \includegraphics[scale=0.13]{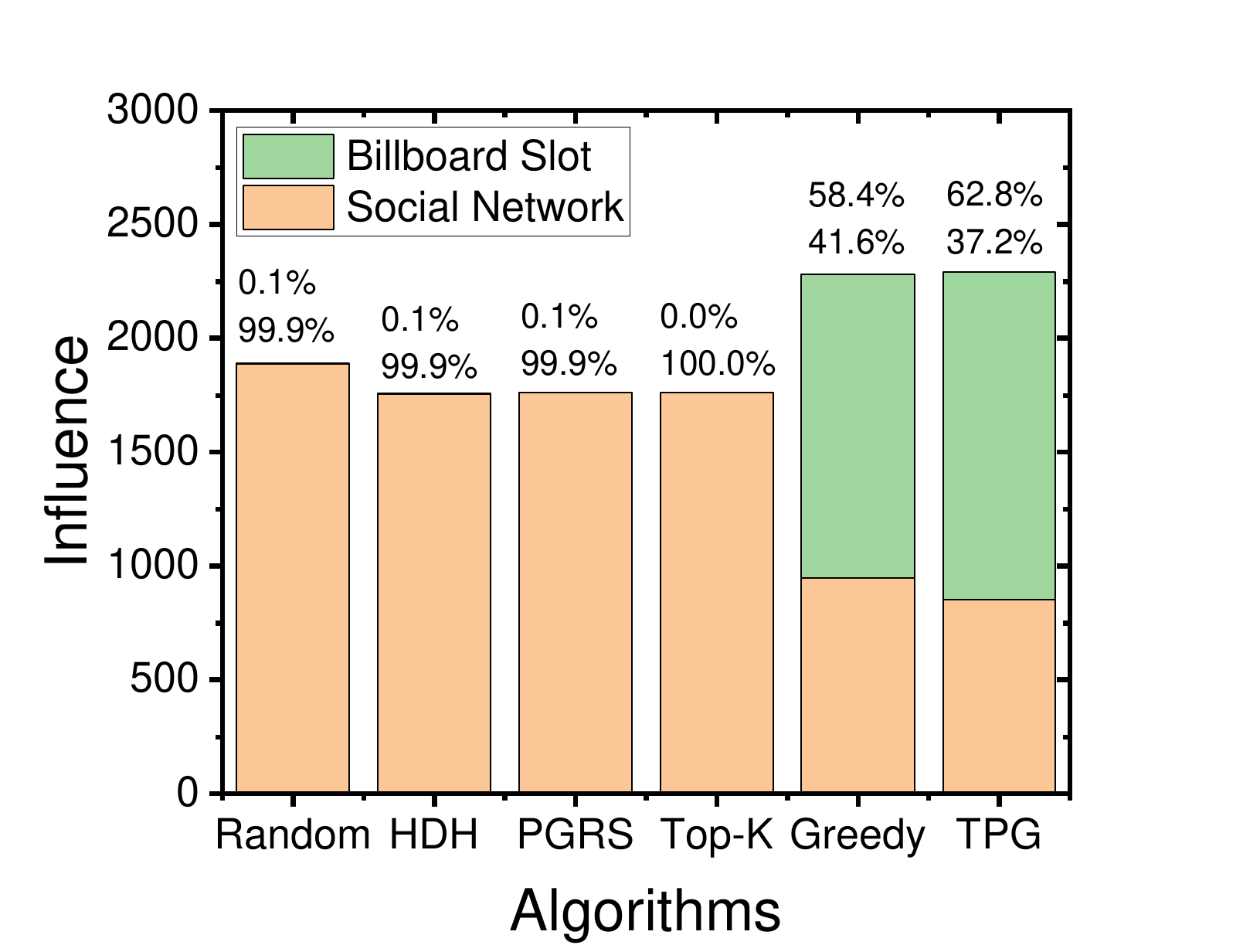} \\
\tiny{(e) Budget = $500$}  & \tiny{(f) Budget = $1000$} & \tiny{(g) Budget = $1500$} & \tiny{(h) Budget = $2000$} \\
\includegraphics[scale=0.13]{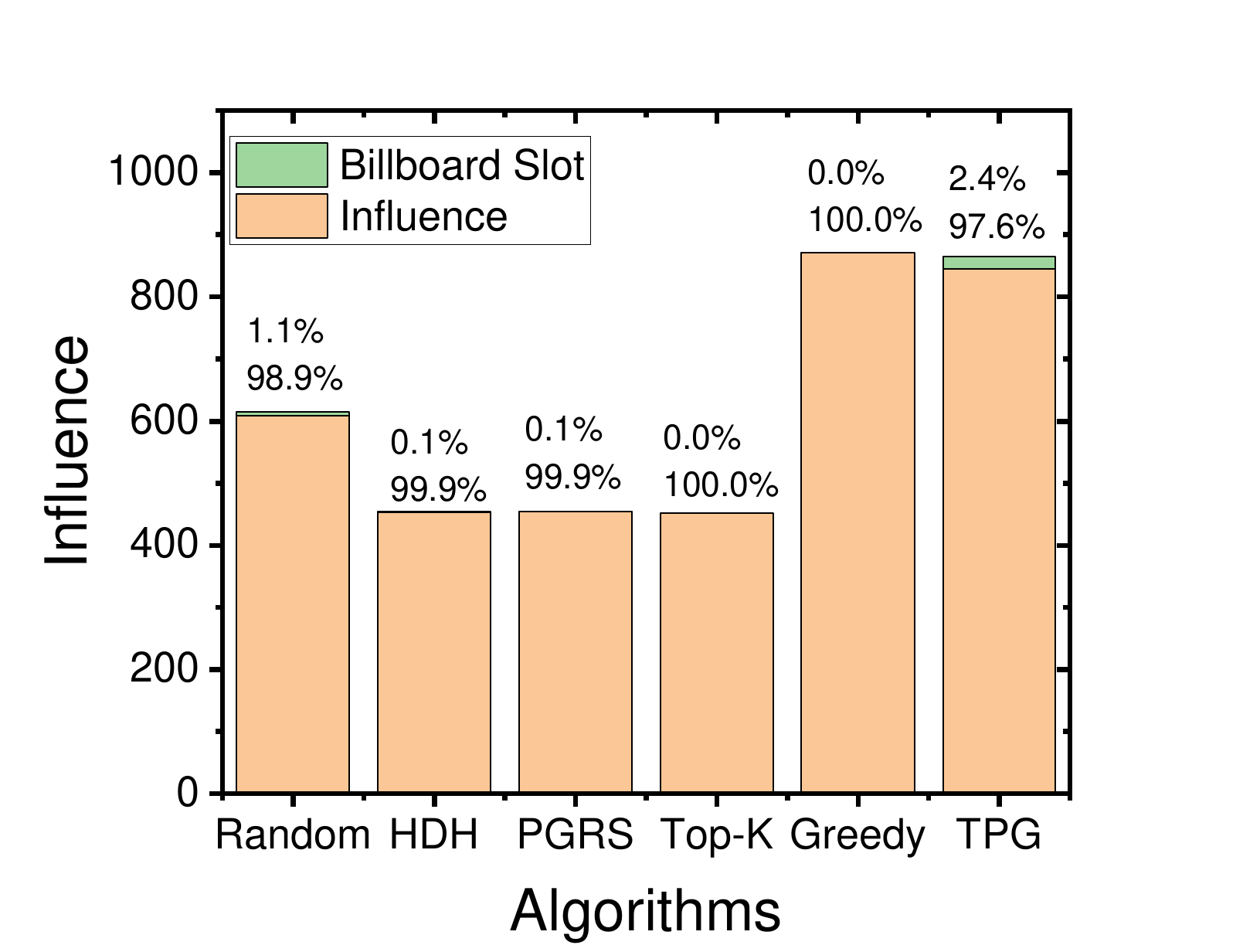} & \includegraphics[scale=0.13]{Algo-Vs-Influence/Uniform/CA/InfulenceVAlgoCAUniform500.pdf} & \includegraphics[scale=0.13]{Algo-Vs-Influence/Uniform/CA/InfulenceVAlgoCAUniform500.pdf} & \includegraphics[scale=0.13]{Algo-Vs-Influence/Uniform/CA/InfulenceVAlgoCAUniform500.pdf} \\
\tiny{(i) Budget = $500$} & \tiny{(j) Budget = $1000$} & \tiny{(k) Budget = $1500$} & \tiny{$(\ell)$ Budget = $2000$} \\ 
\end{tabular}
\caption{Varying Algorithms Vs. Influence in Trivalency $(a,b,c,d)$, in Weighted Cascade $(e,f,g,h)$, in Uniform probability setting $(i,j,k,\ell)$ in CA Dataset}
\label{Fig:Plot1}
\end{figure*}

%%%%%%%%%%%%%%%%%%%%%%%%%%%%%%%%% Algorithm Vs. Budget Split CA %%%%%%%%%%%%%%%%%%%%%%%%%%%%%%%%%%%%%%%%

\begin{figure*}[!ht]
\centering
\setlength{\tabcolsep}{0.1pt} % tighter spacing between columns
\renewcommand{\arraystretch}{0.9} % tighter spacing between rows
\begin{tabular}{cccc}
\includegraphics[scale=0.13]{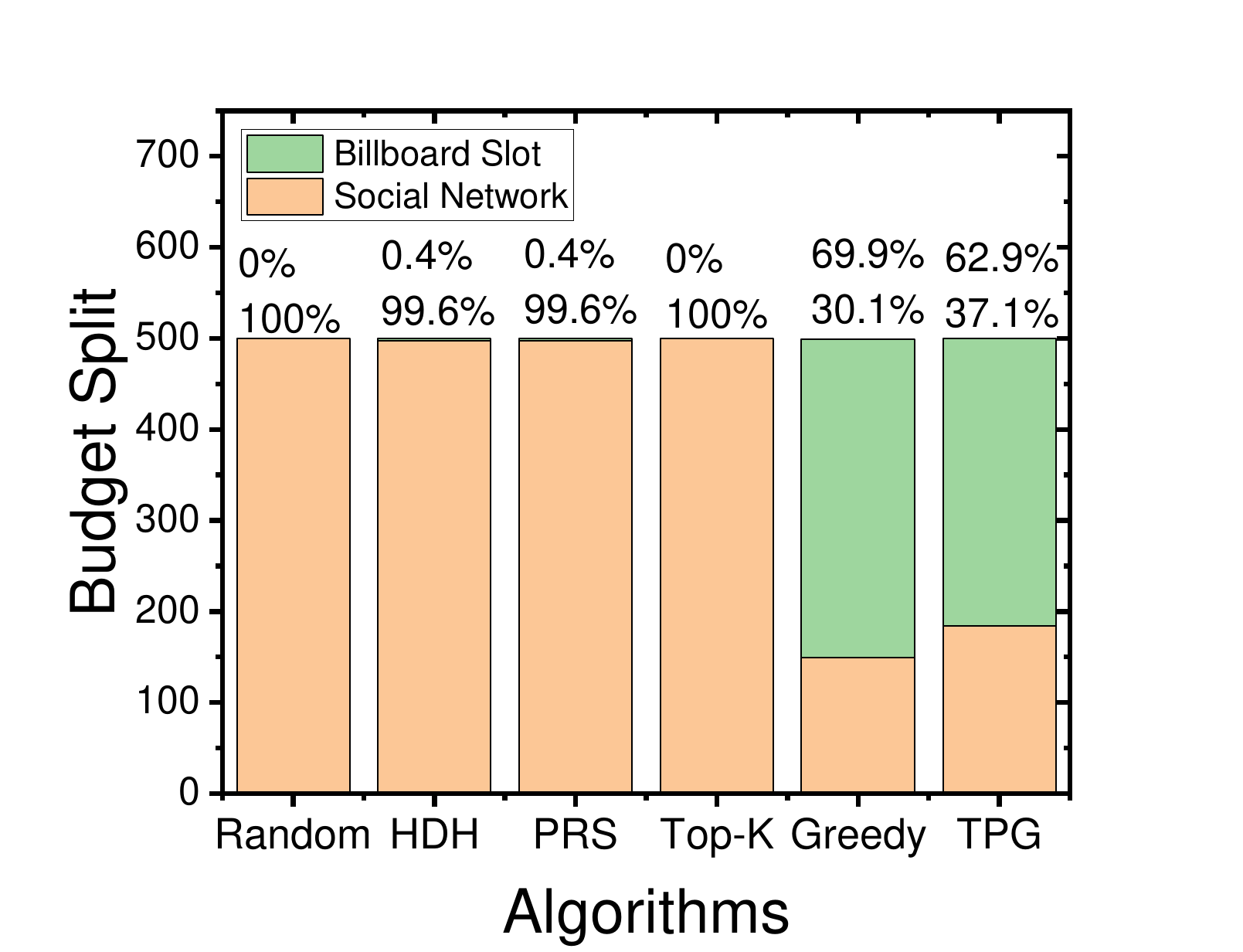} & \includegraphics[scale=0.13]{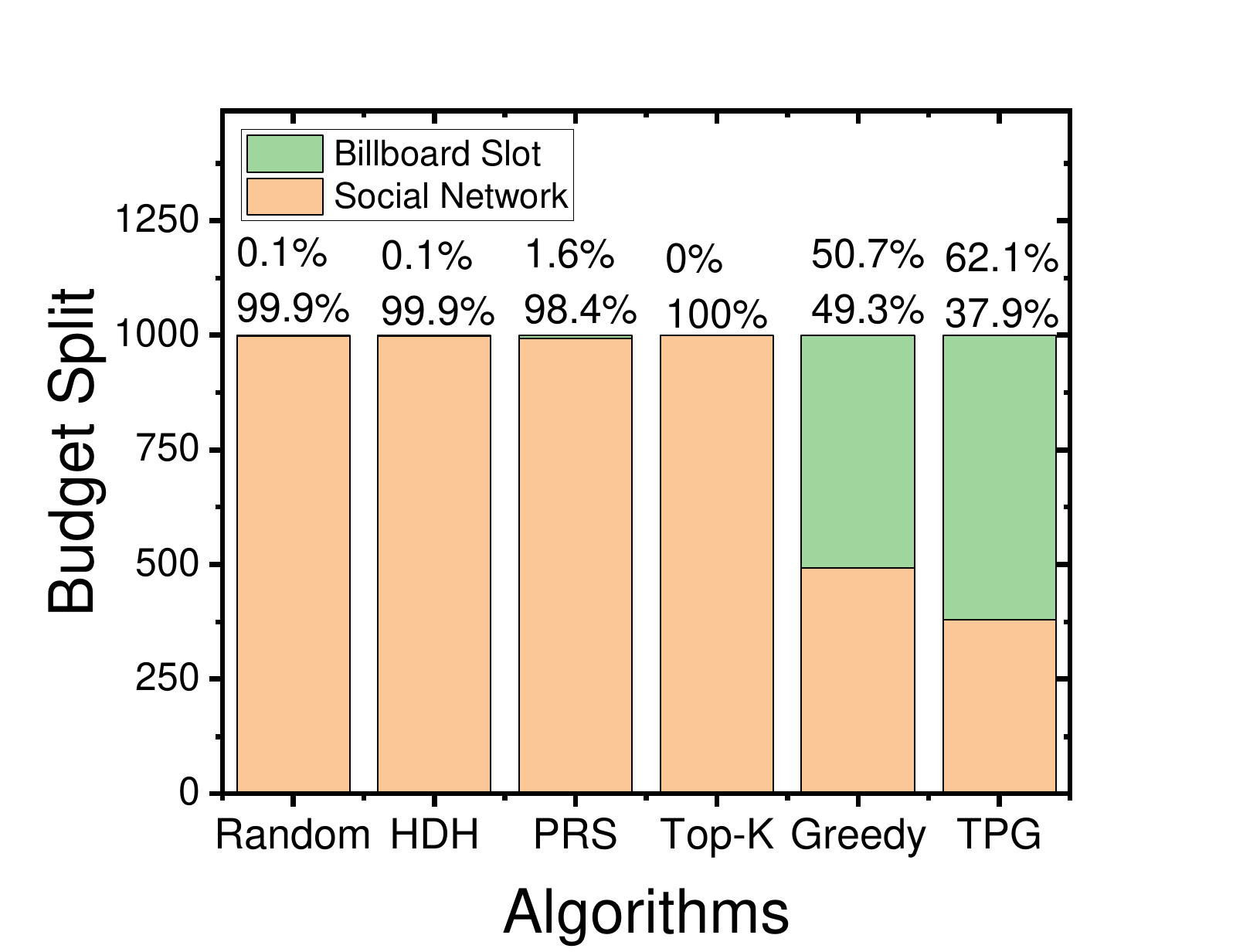} & \includegraphics[scale=0.13]{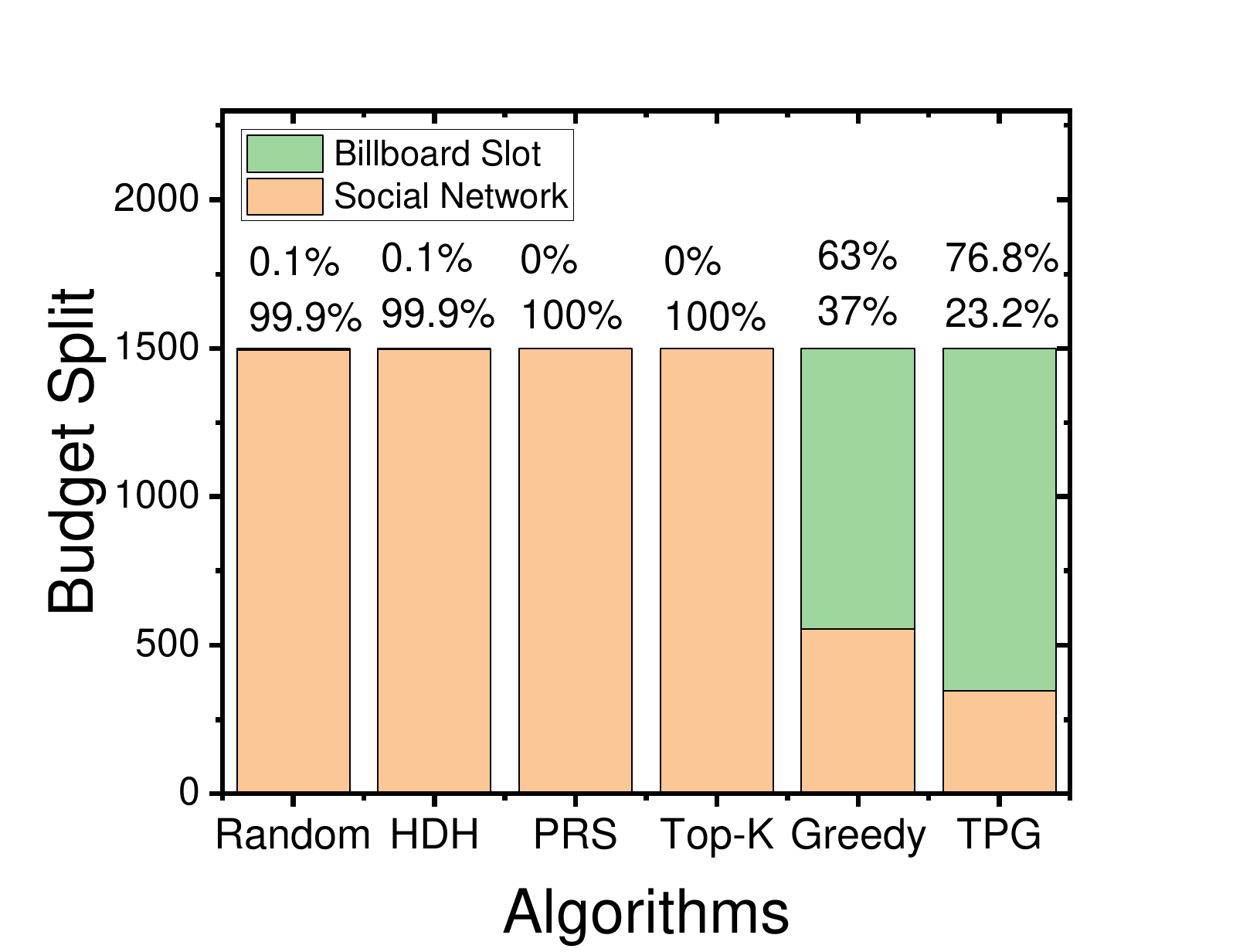} &
\includegraphics[scale=0.13]{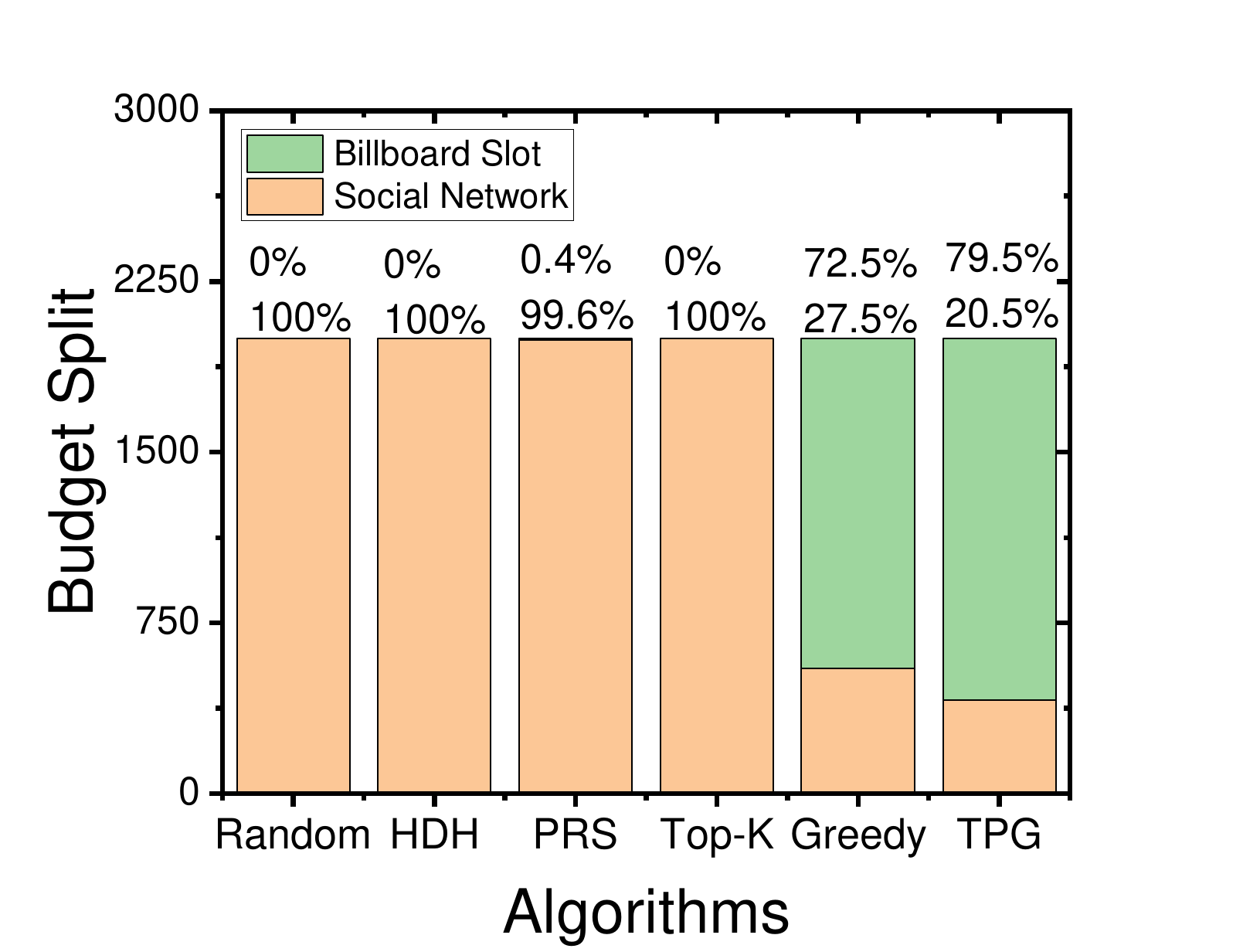} \\
\tiny{(a) Budget = $500$} &  \tiny{(b) Budget = $1000$} & \tiny{(c) Budget = $1500$} & \tiny{(d) Budget = $2000$} \\
\includegraphics[scale=0.13]{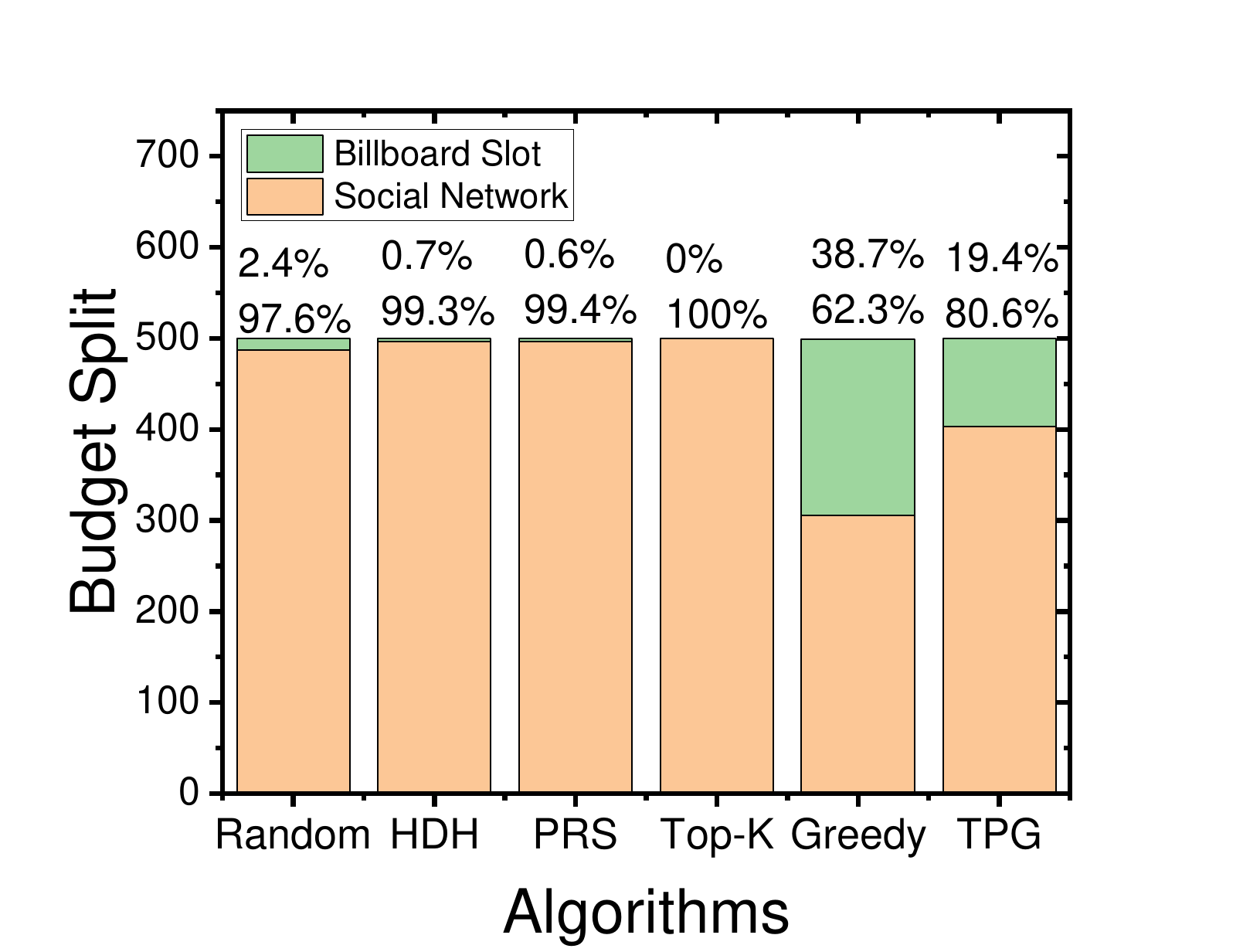} & 
\includegraphics[scale=0.13]{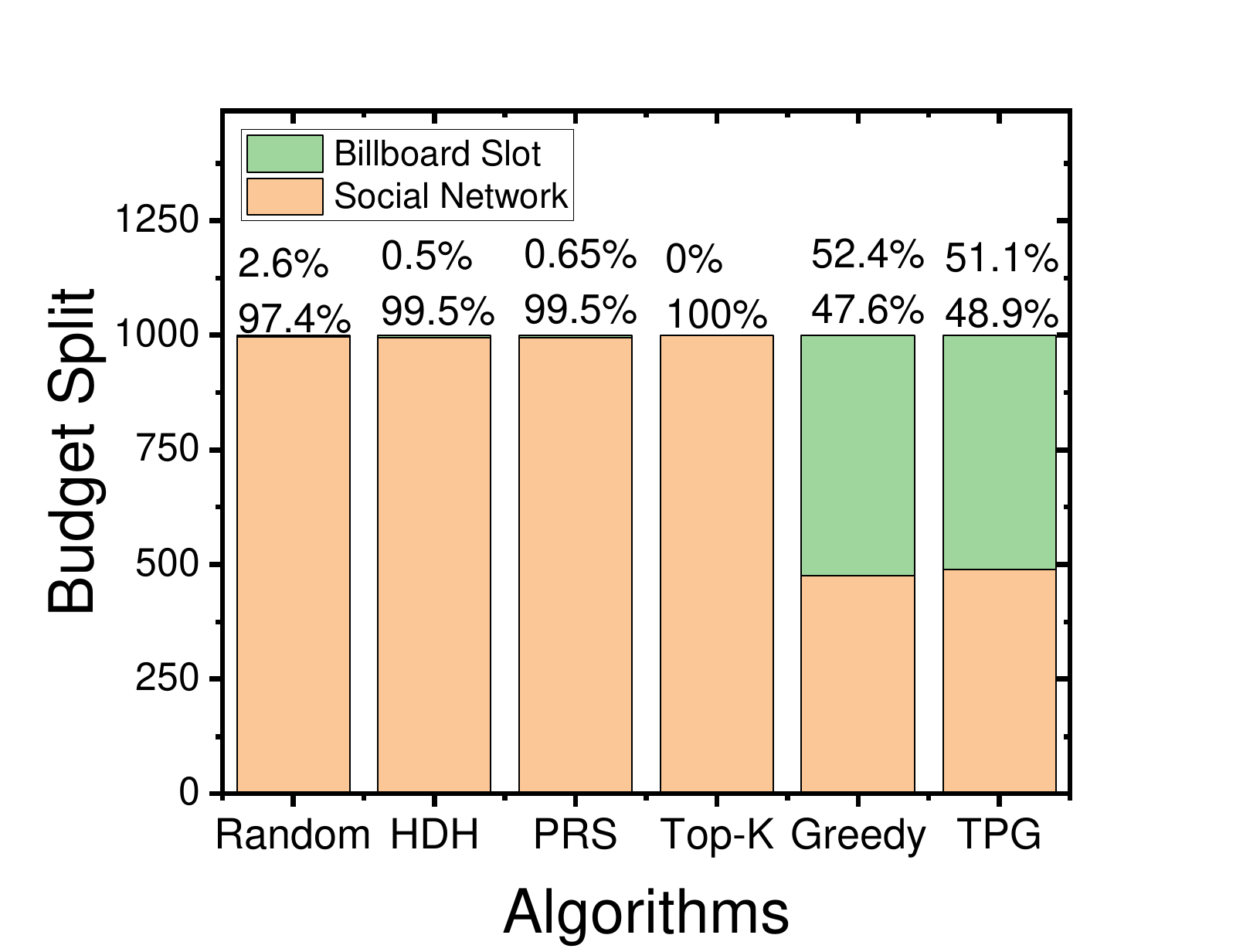} &
\includegraphics[scale=0.13]{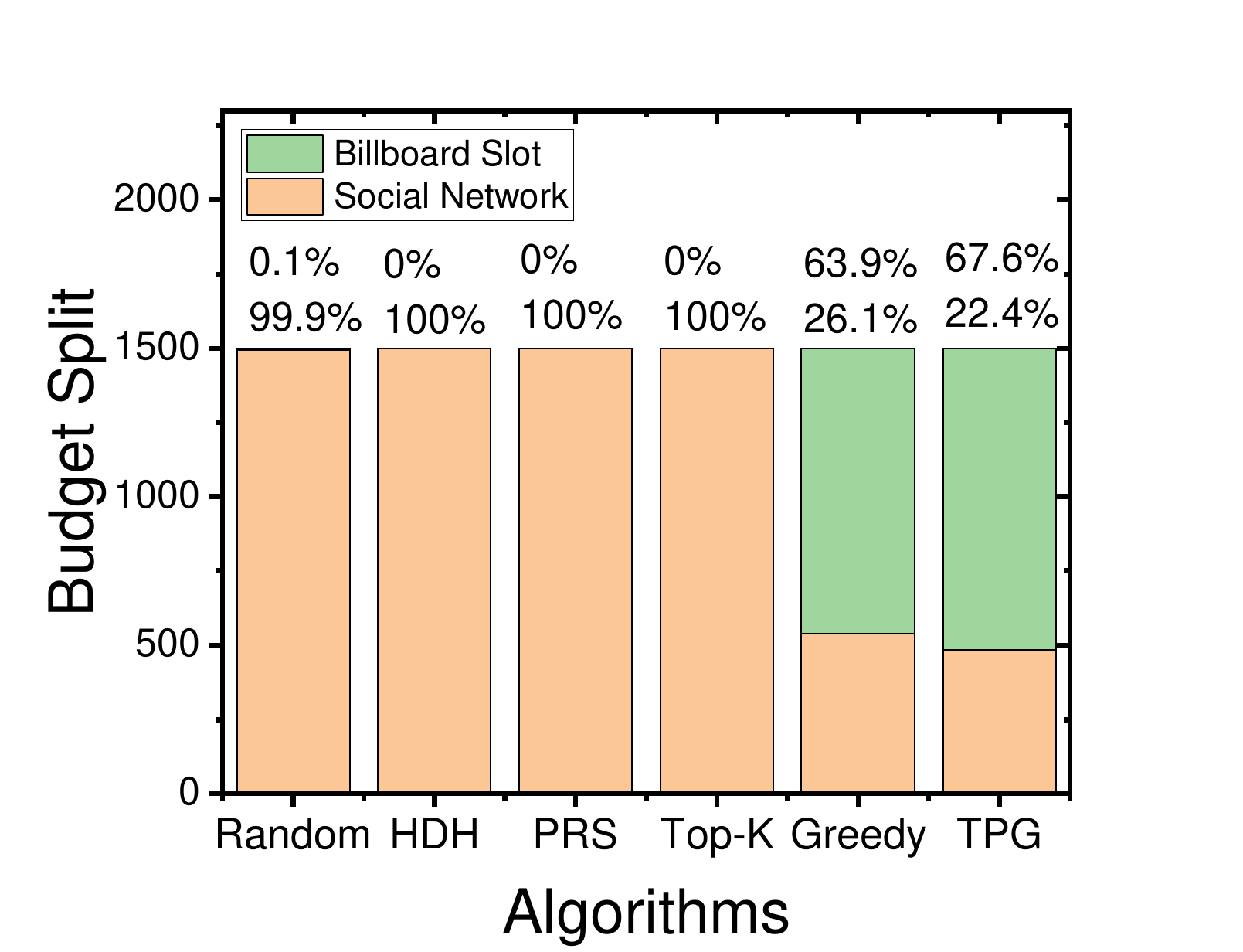} & 
\includegraphics[scale=0.13]{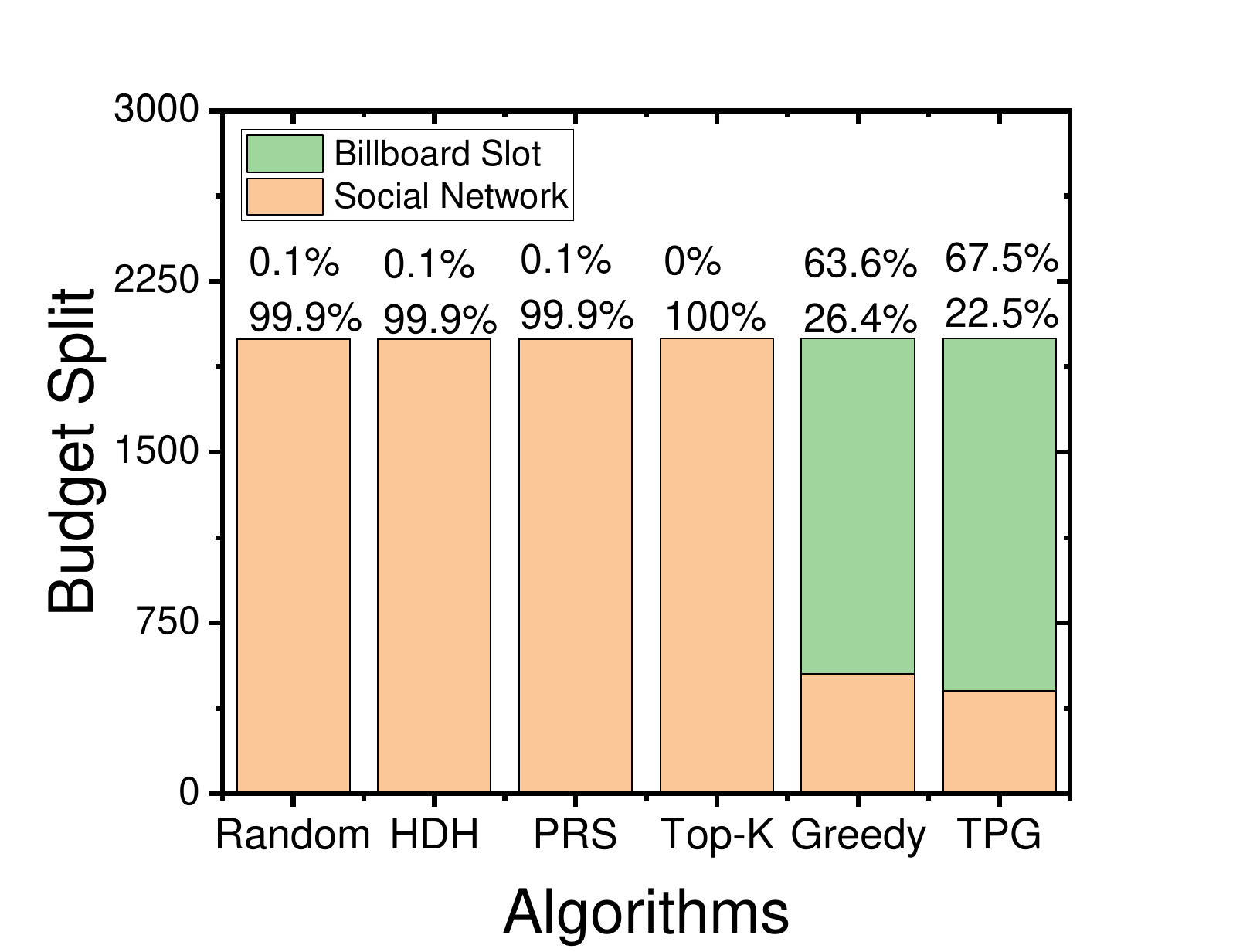} \\
\tiny{(e) Budget = $500$}  & \tiny{(f) Budget = $1000$} & \tiny{(g) Budget = $1500$} & \tiny{(h) Budget = $2000$} \\
\includegraphics[scale=0.13]{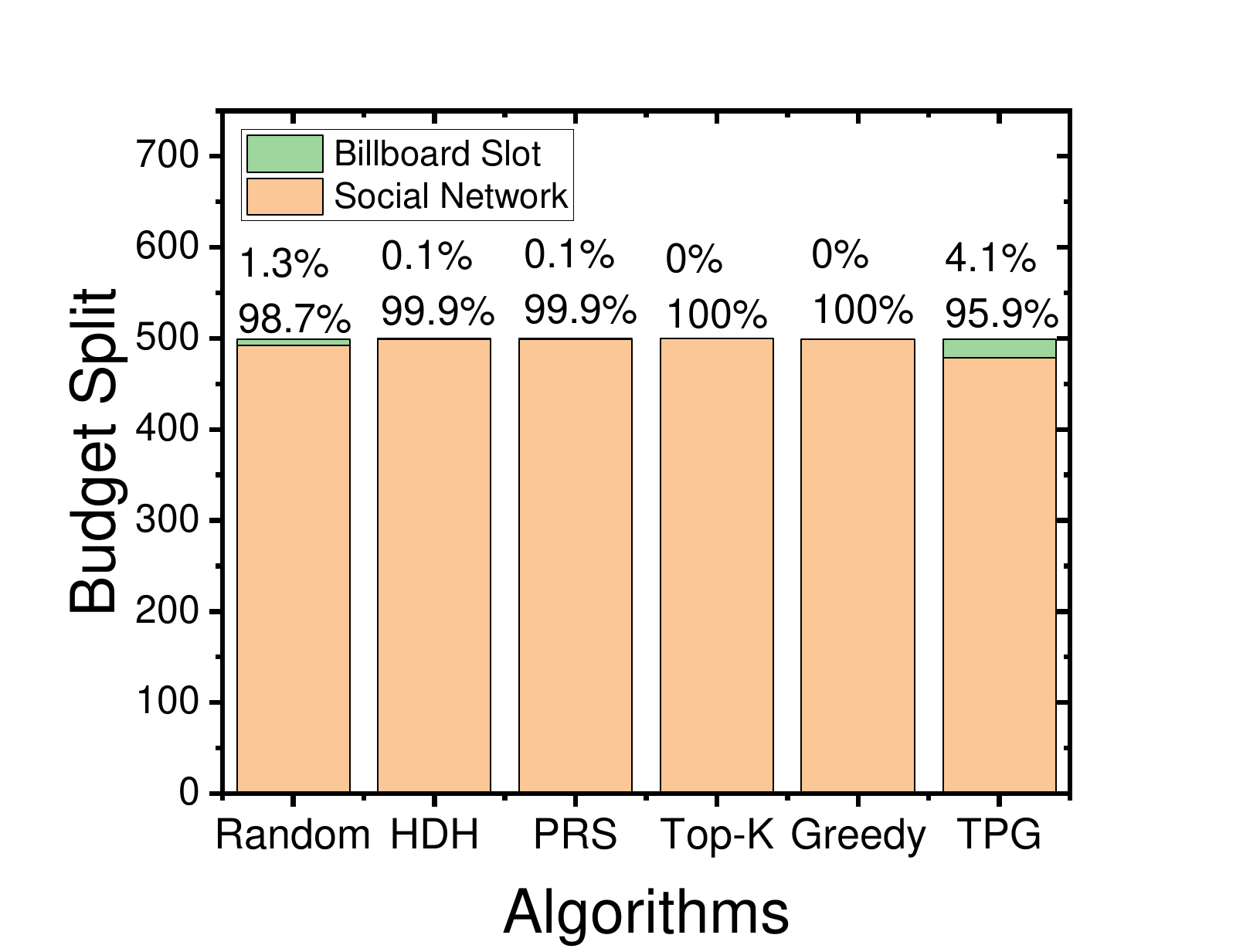} & \includegraphics[scale=0.13]{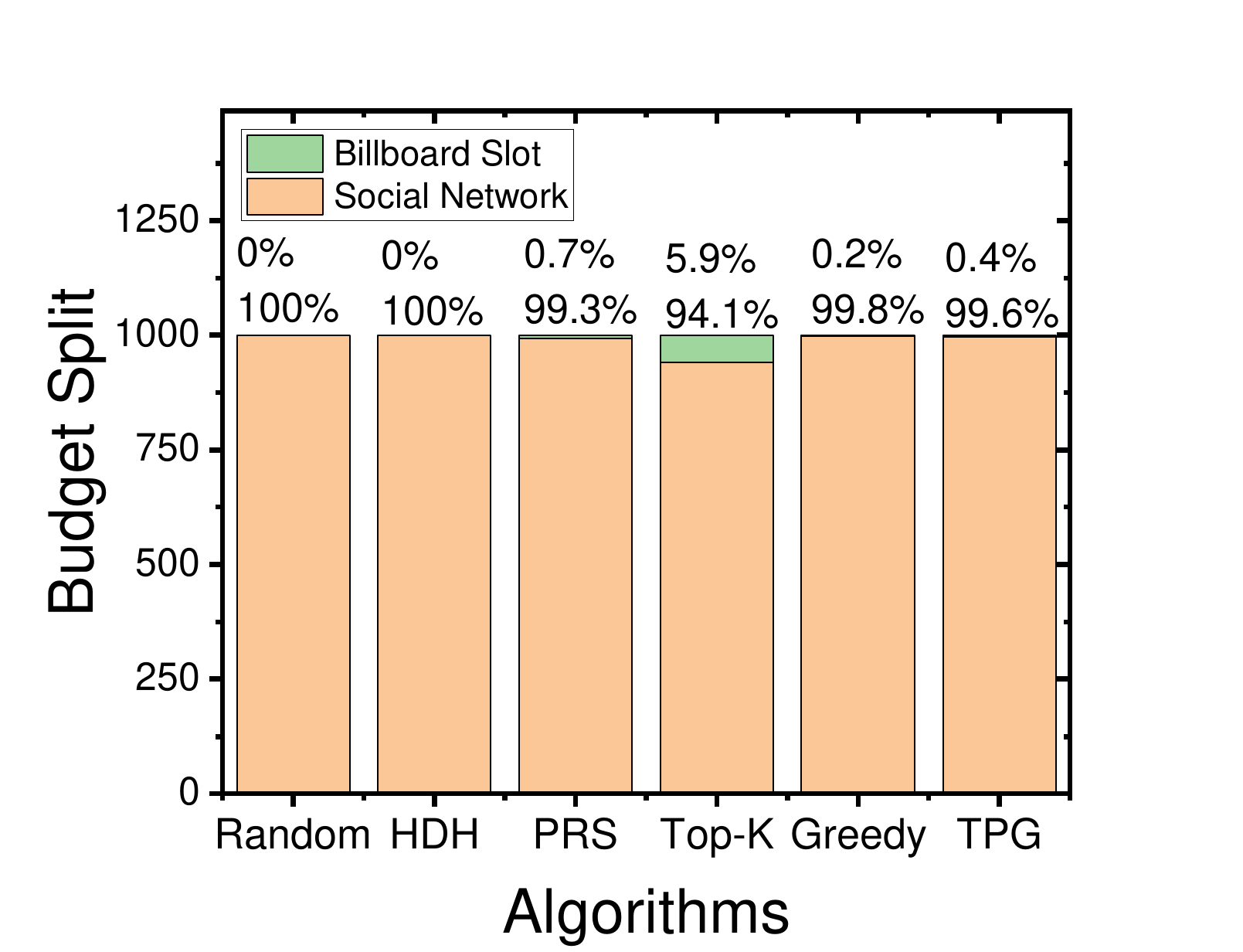} & \includegraphics[scale=0.13]{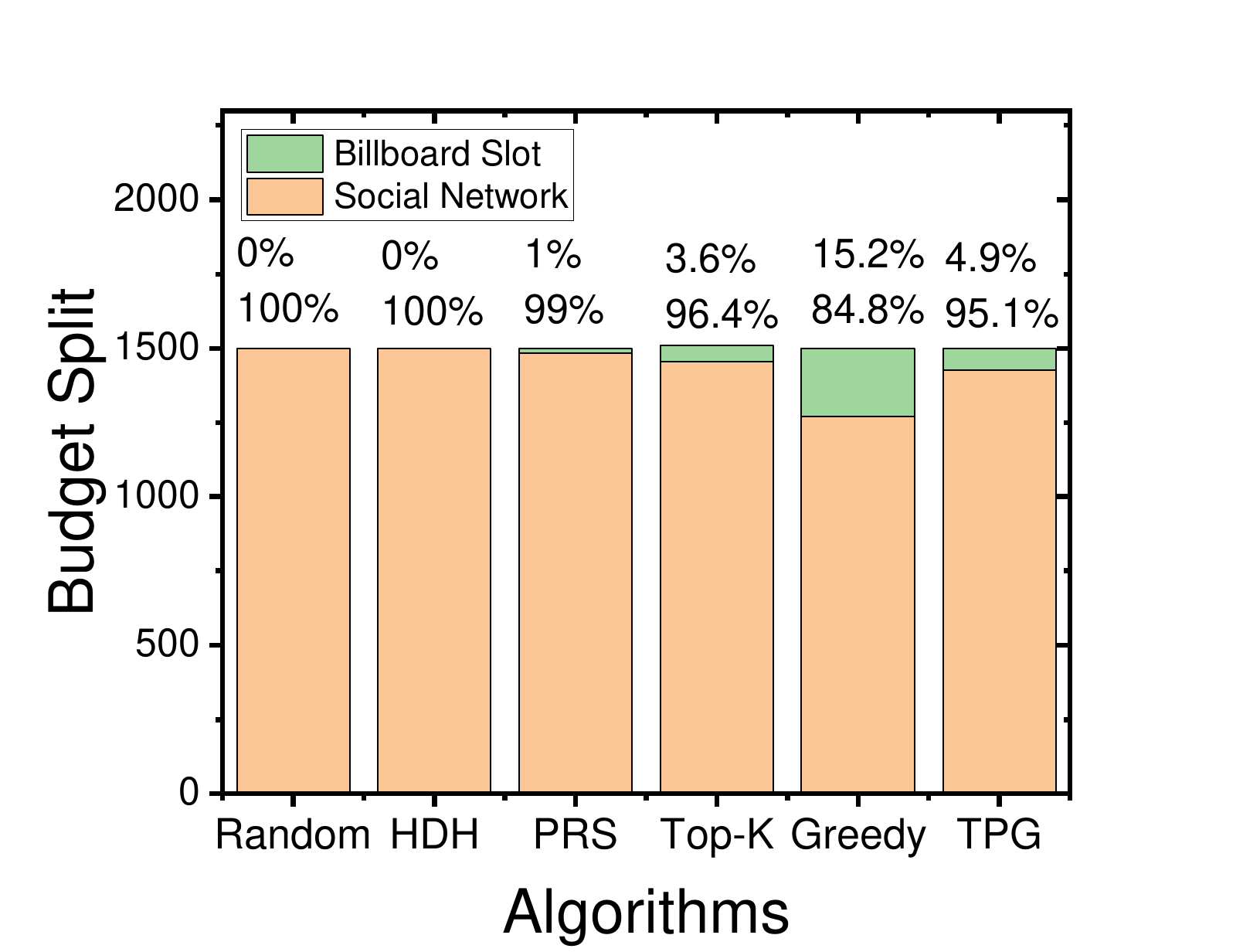} & \includegraphics[scale=0.13]{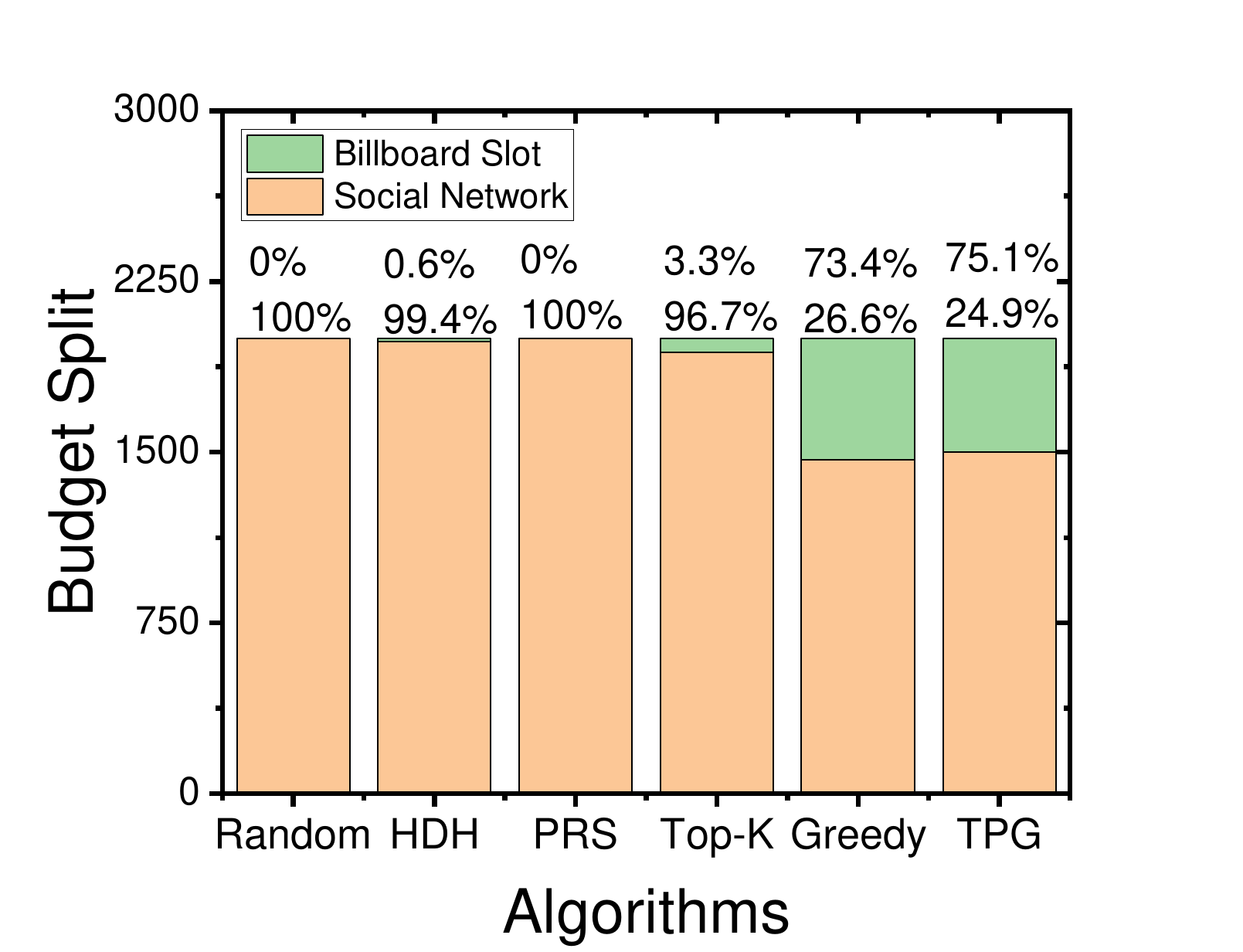} \\
\tiny{(i) Budget = $500$} & \tiny{(j) Budget = $1000$} & \tiny{(k) Budget = $1500$} & \tiny{$(\ell)$ Budget = $2000$} \\ 
\end{tabular}
\caption{Varying Algorithms Vs. Budget split percentage in Trivalency $(a,b,c,d)$, in Weighted Cascade $(e,f,g,h)$, in Uniform probability setting $(i,j,k,\ell)$ in CA Dataset}
\label{Fig:Plot2}
\end{figure*}
%%%%%%%%%%%%%%%%%%%%%%%%%%%%%%%%% Algorithm Vs. Budget Split USA %%%%%%%%%%%%%%%%%%%%%%%%%%%%%%%%%%%%%%%%

\begin{figure*}[!ht]
\centering
\setlength{\tabcolsep}{0.1pt} % tighter spacing between columns
\renewcommand{\arraystretch}{0.9} % tighter spacing between rows
\begin{tabular}{cccc}
\includegraphics[scale=0.13]{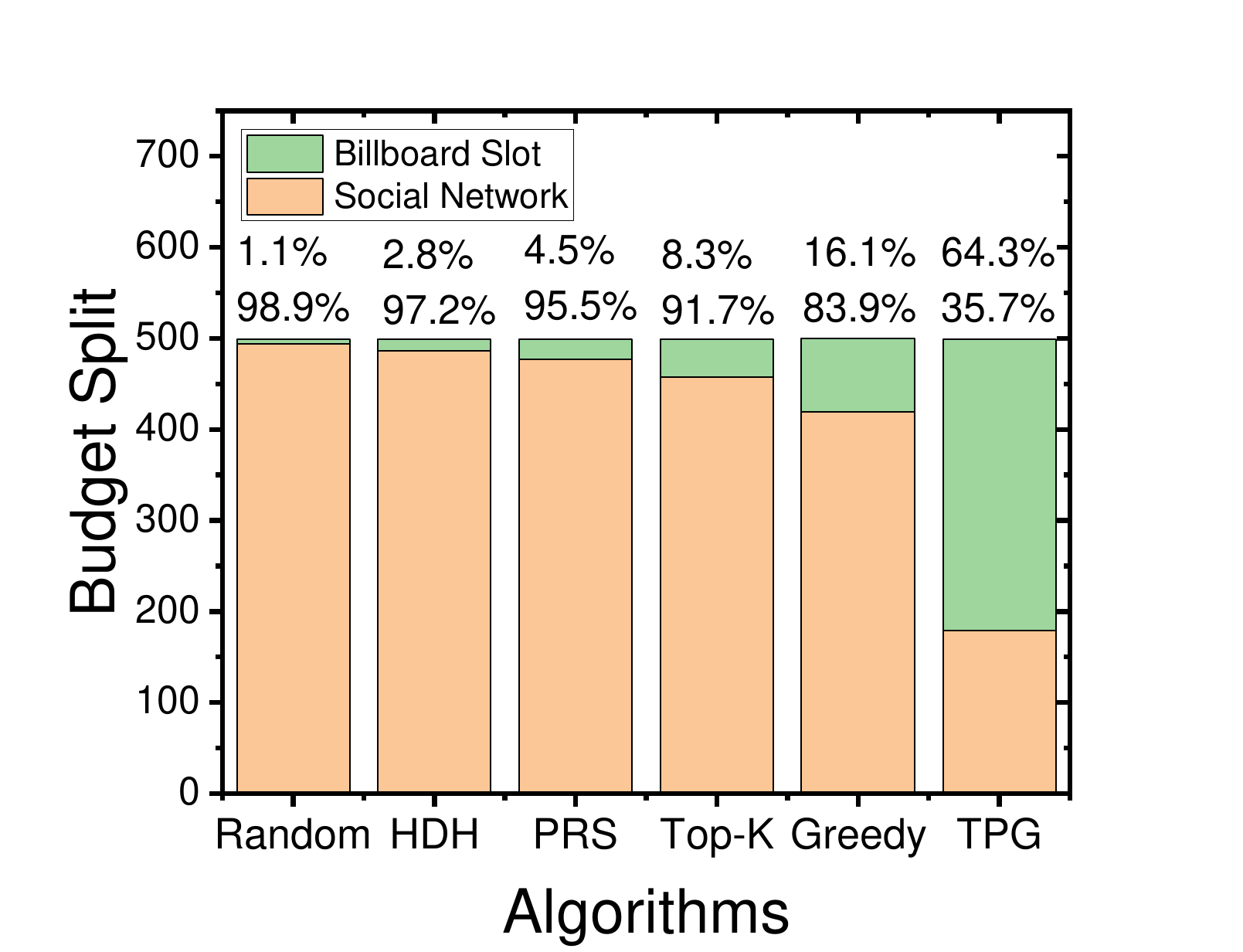} & \includegraphics[scale=0.13]{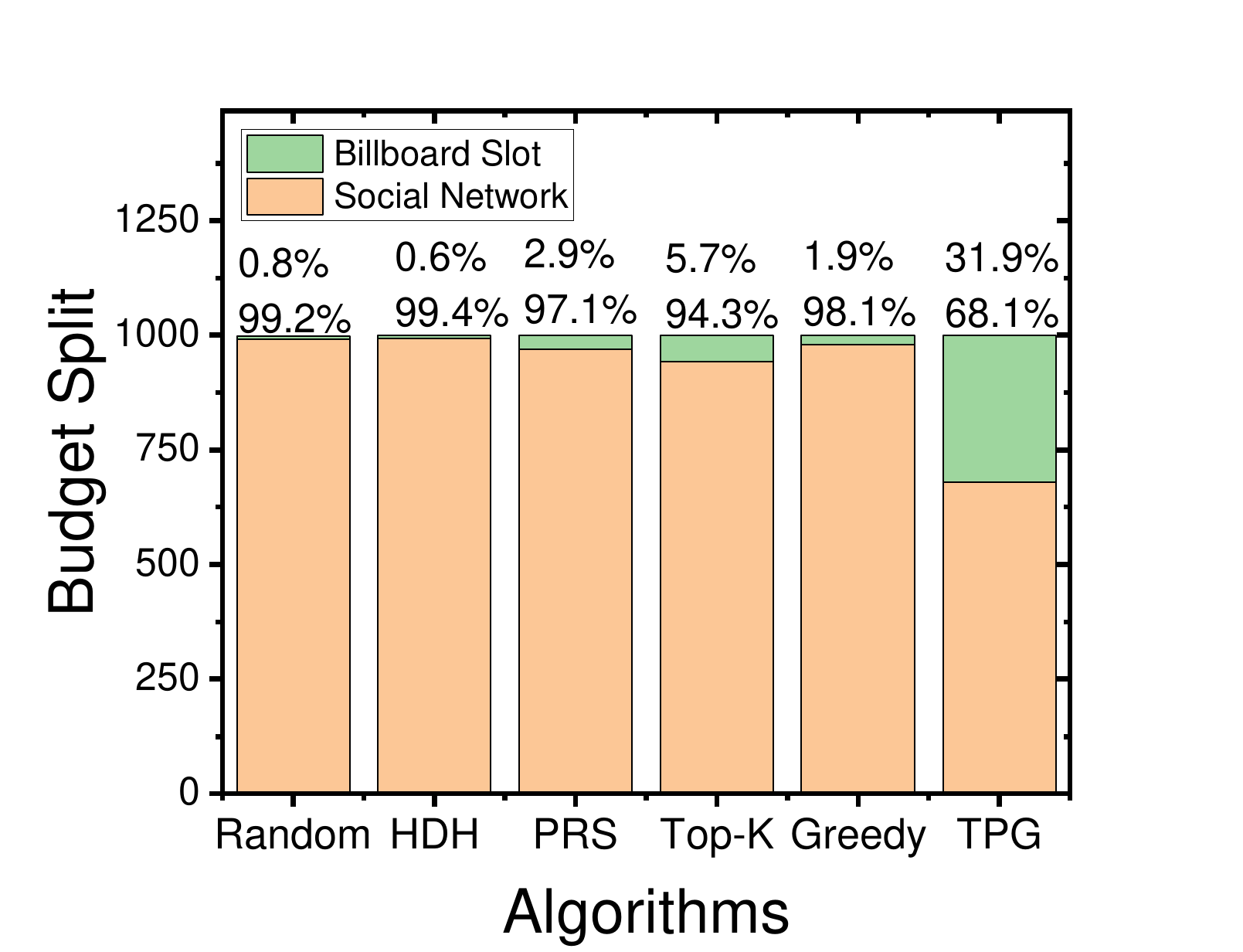} & \includegraphics[scale=0.13]{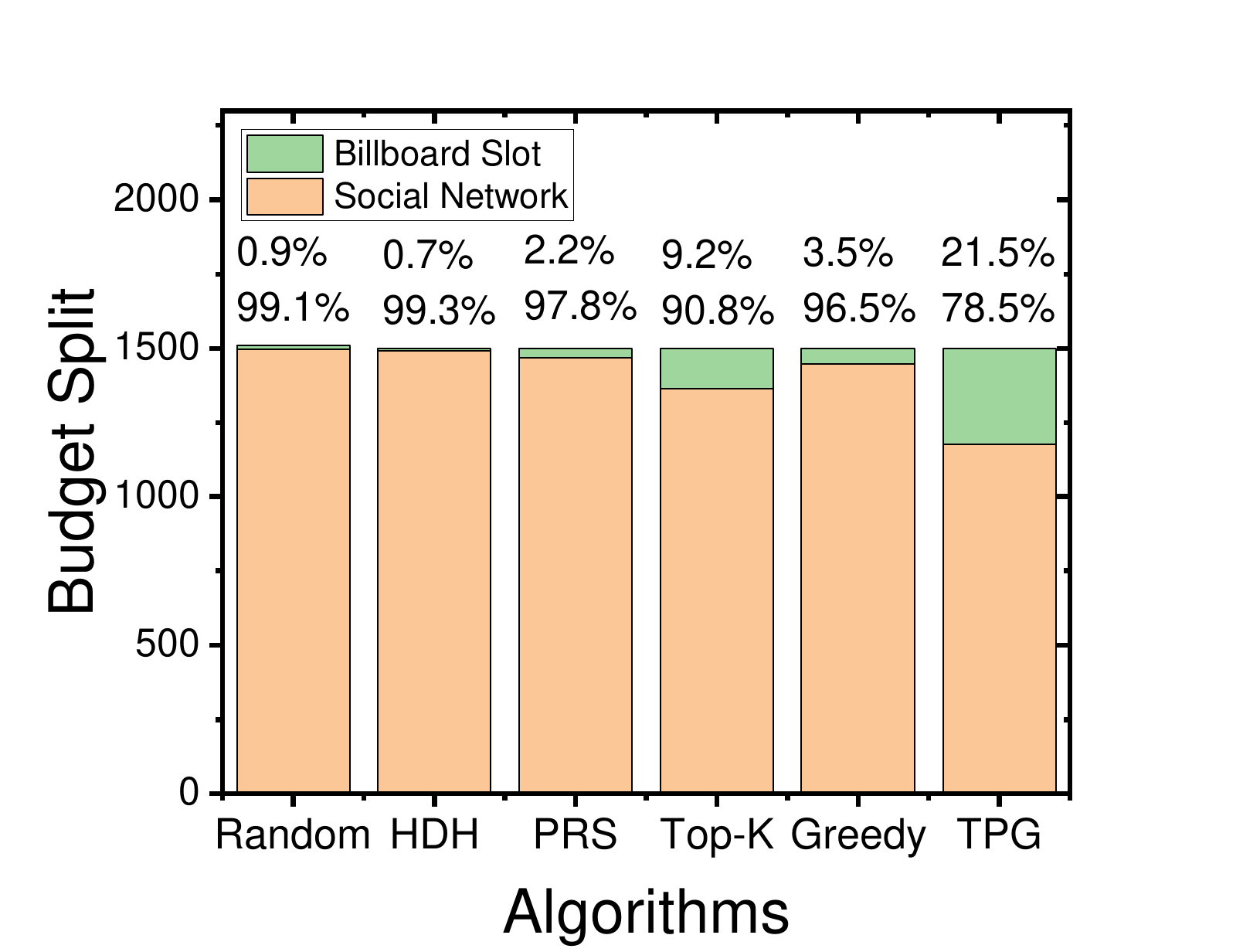} &
\includegraphics[scale=0.13]{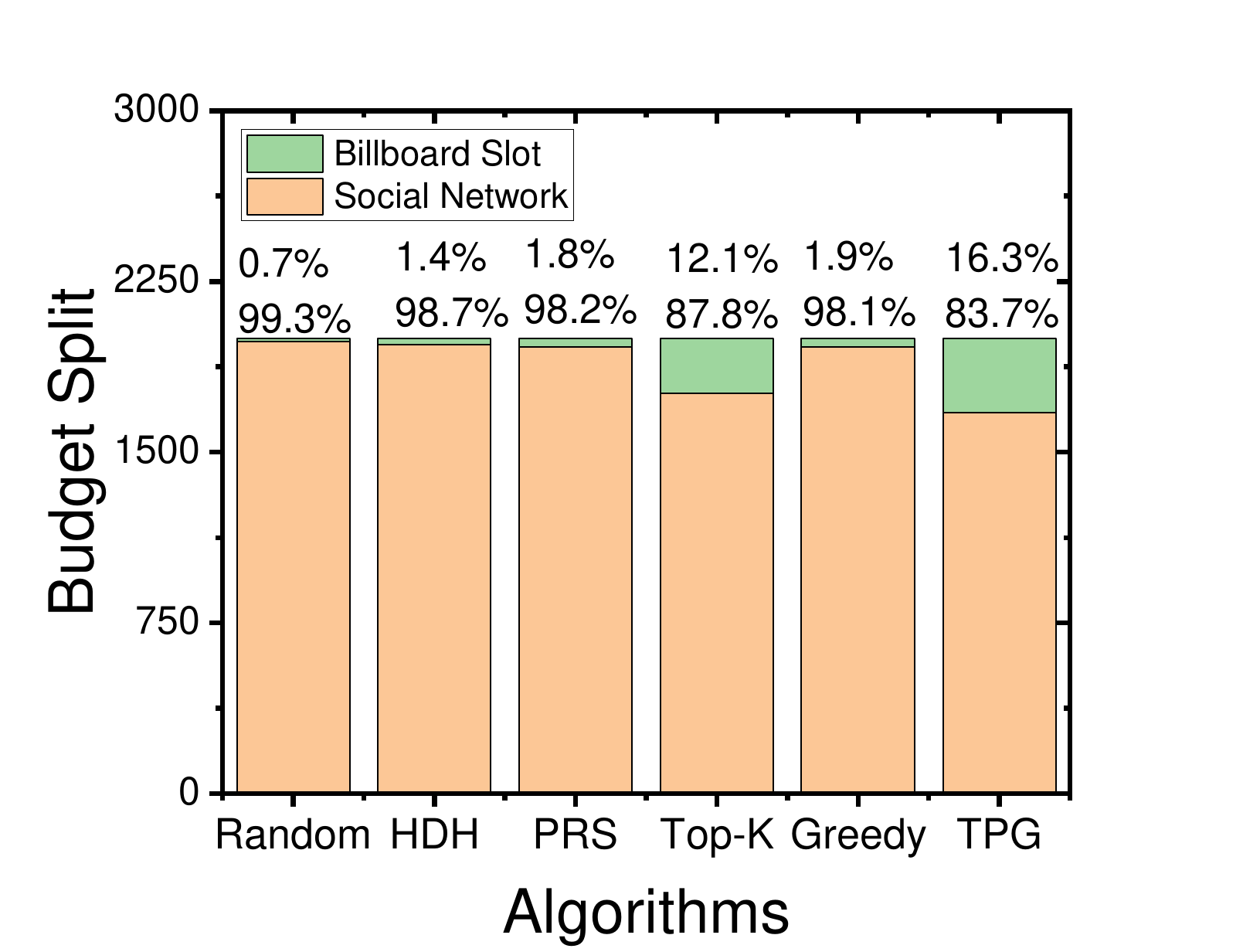} \\
\tiny{(a) Budget = $500$} &  \tiny{(b) Budget = $1000$} & \tiny{(c) Budget = $1500$} & \tiny{(d) Budget = $2000$} \\
\includegraphics[scale=0.13]{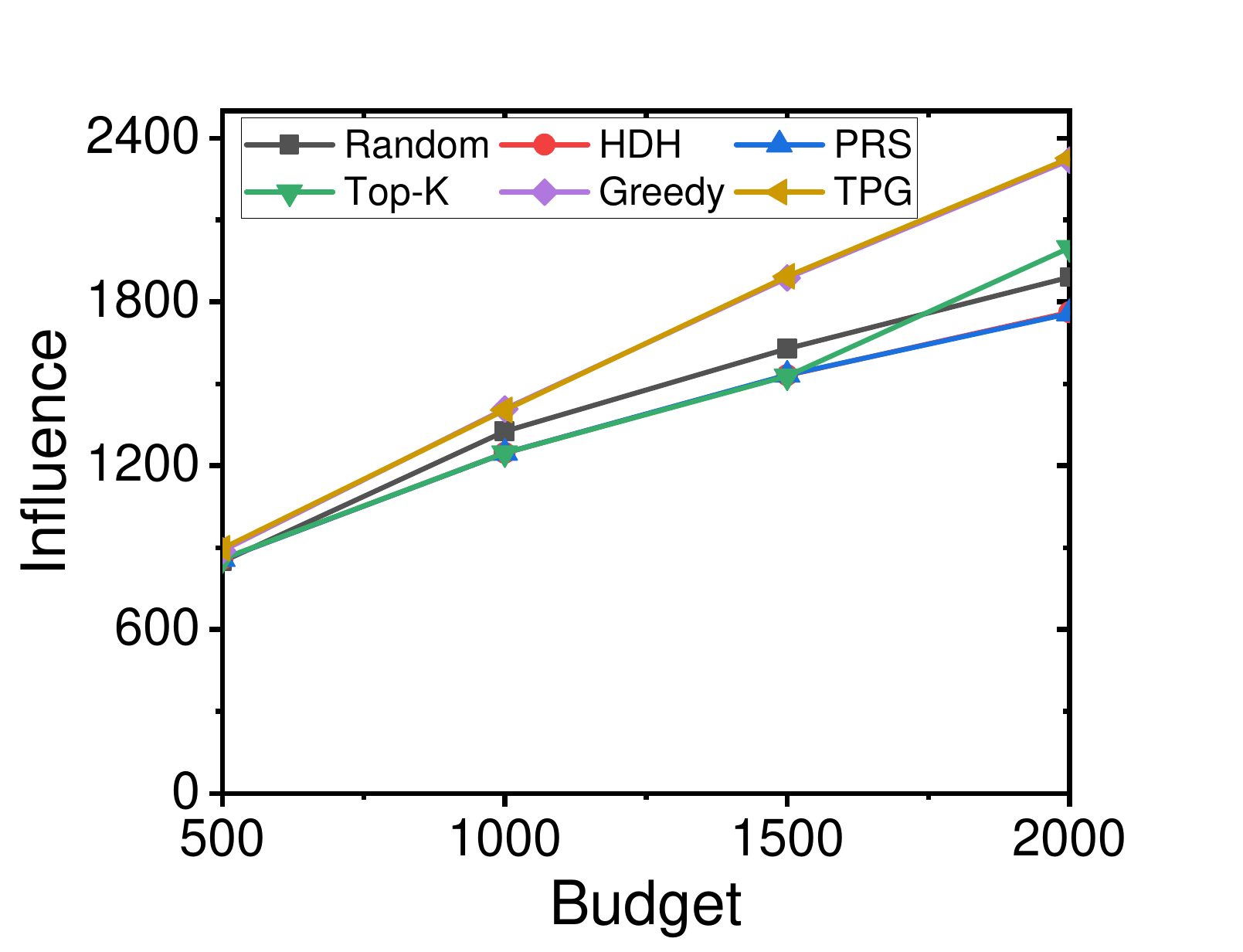} & 
\includegraphics[scale=0.13]{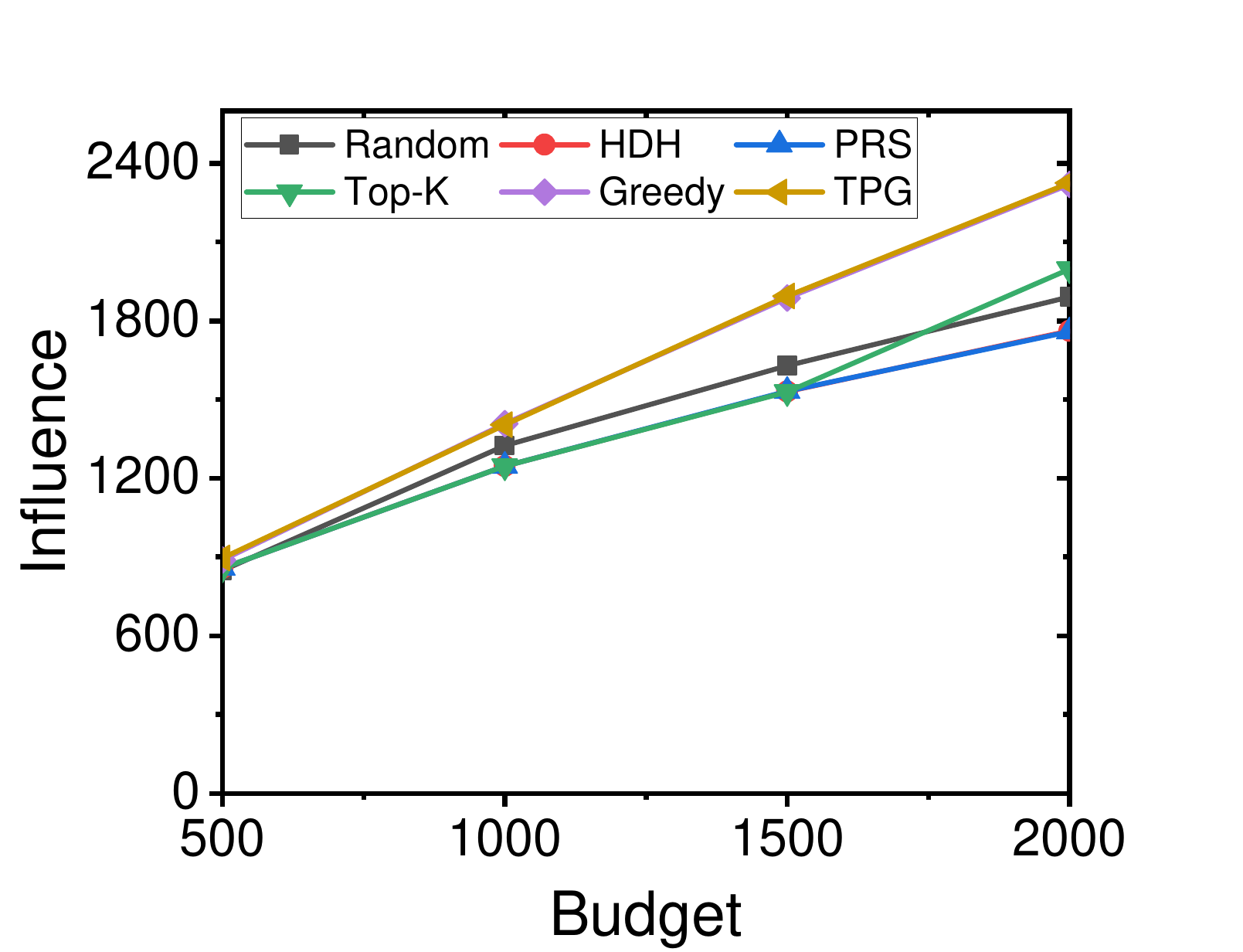} &
\includegraphics[scale=0.13]{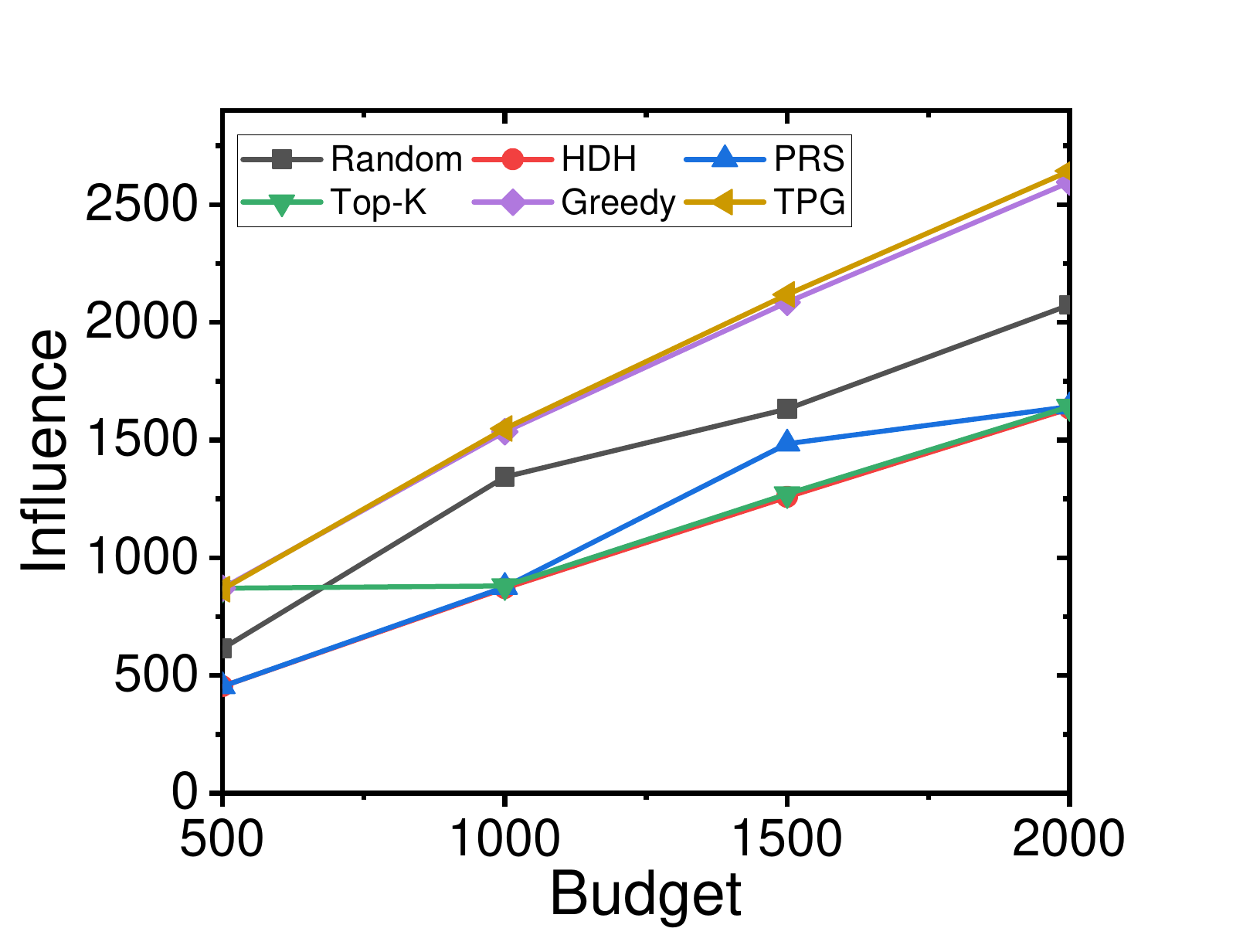} & 
\includegraphics[scale=0.13]{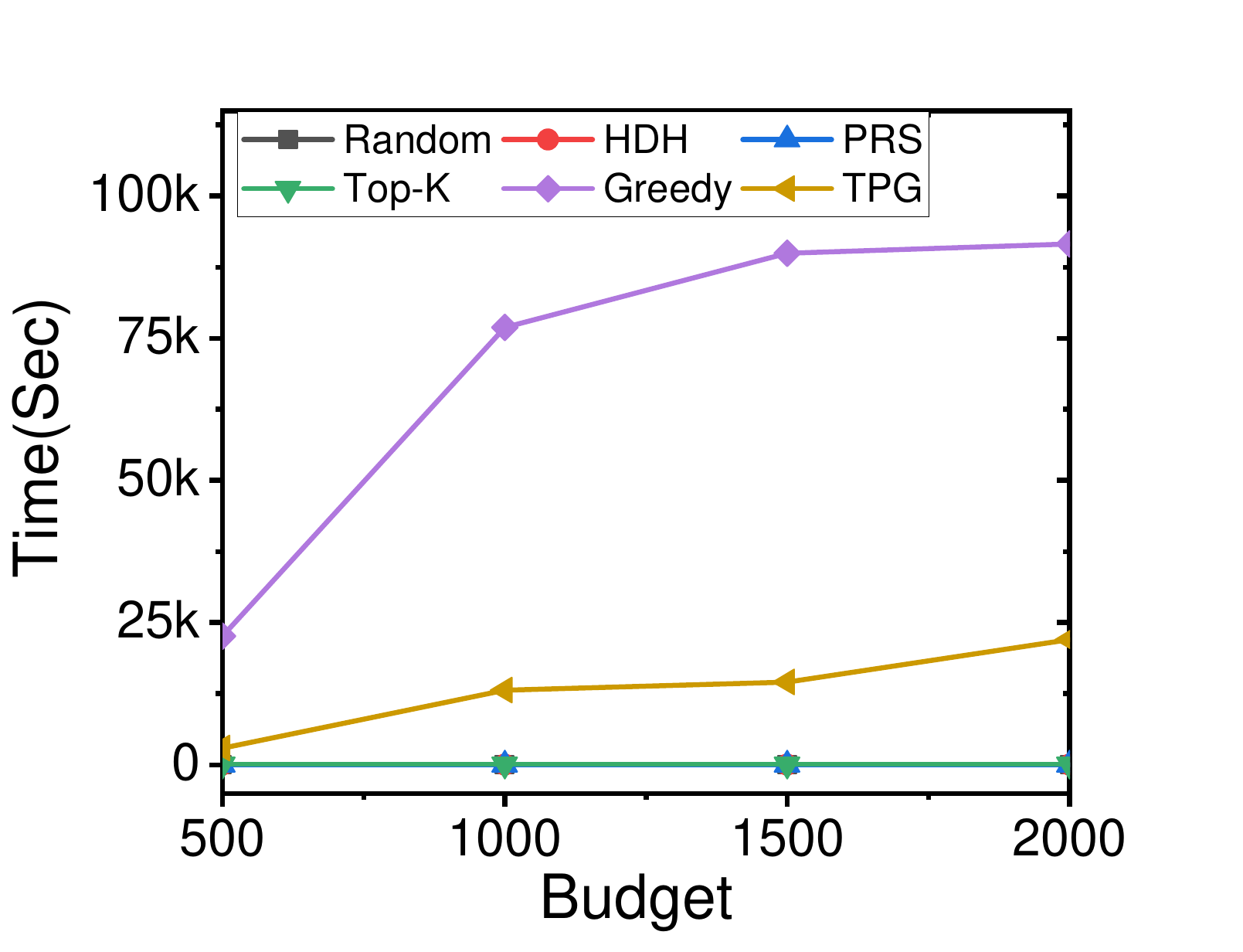} \\
\tiny{(e) Trivalency}  & \tiny{(f) Weighted Cascade} & \tiny{(g)  Uniform} & \tiny{(h) Trivalency} \\
\includegraphics[scale=0.13]{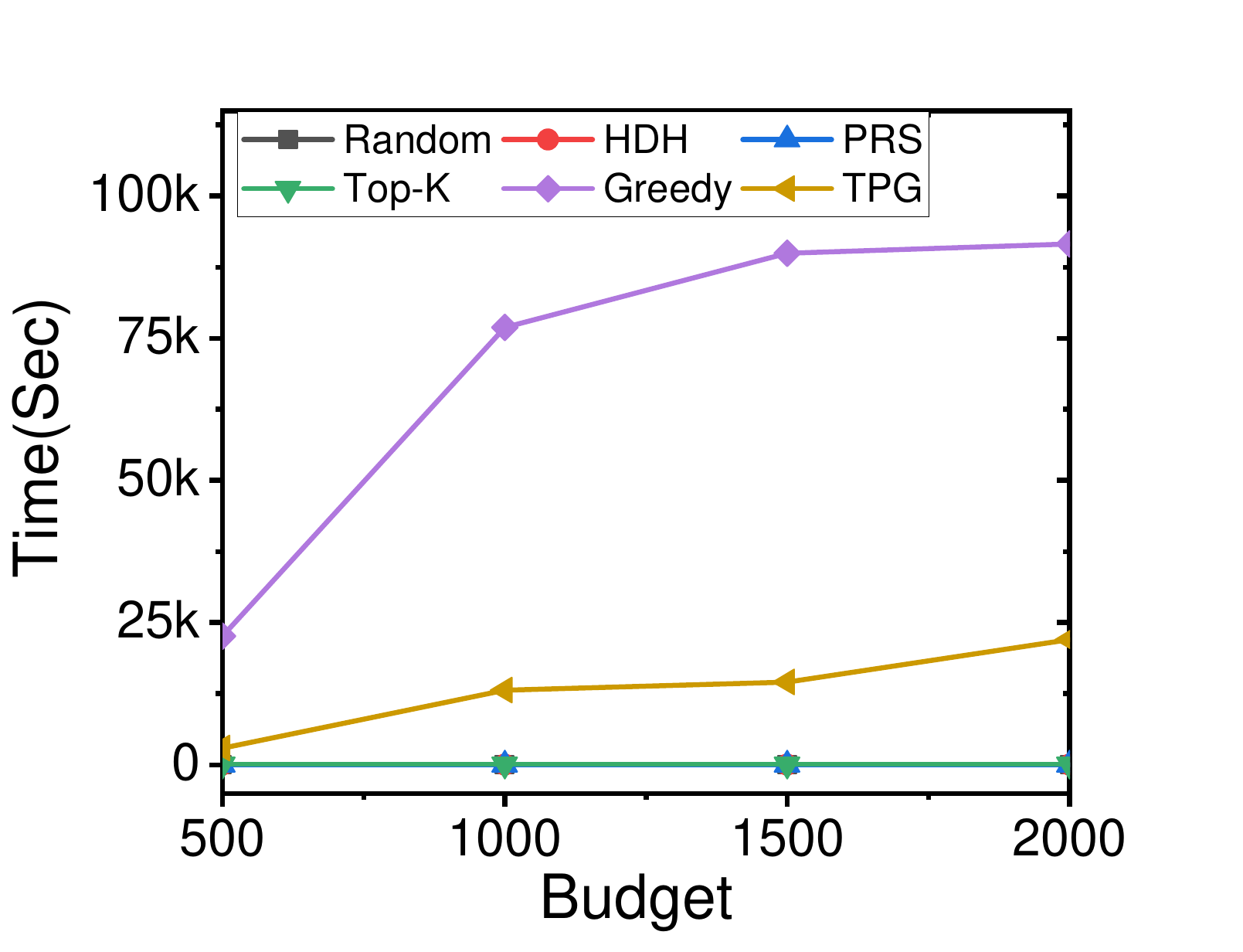} &
\includegraphics[scale=0.13]{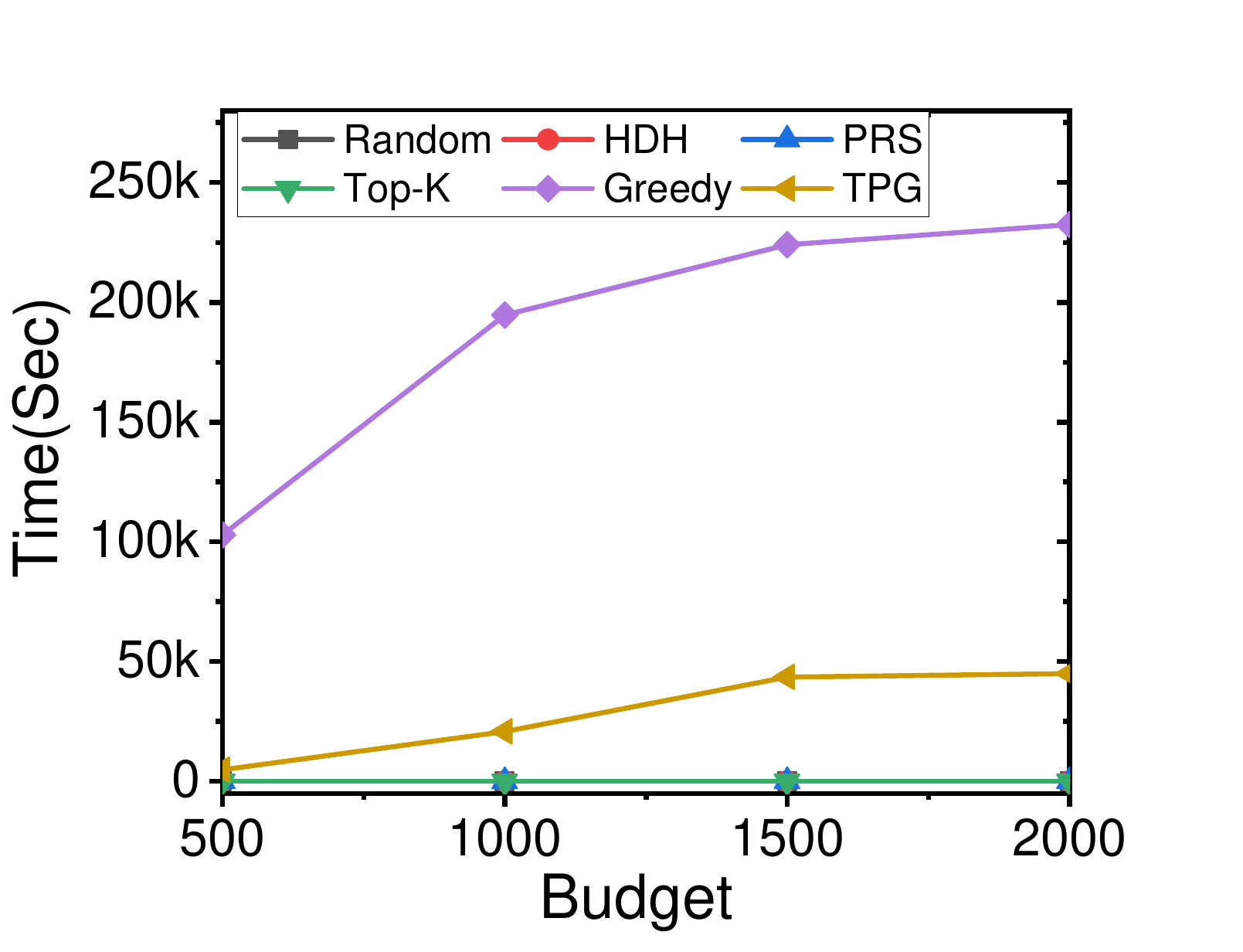} & 
\includegraphics[scale=0.13]{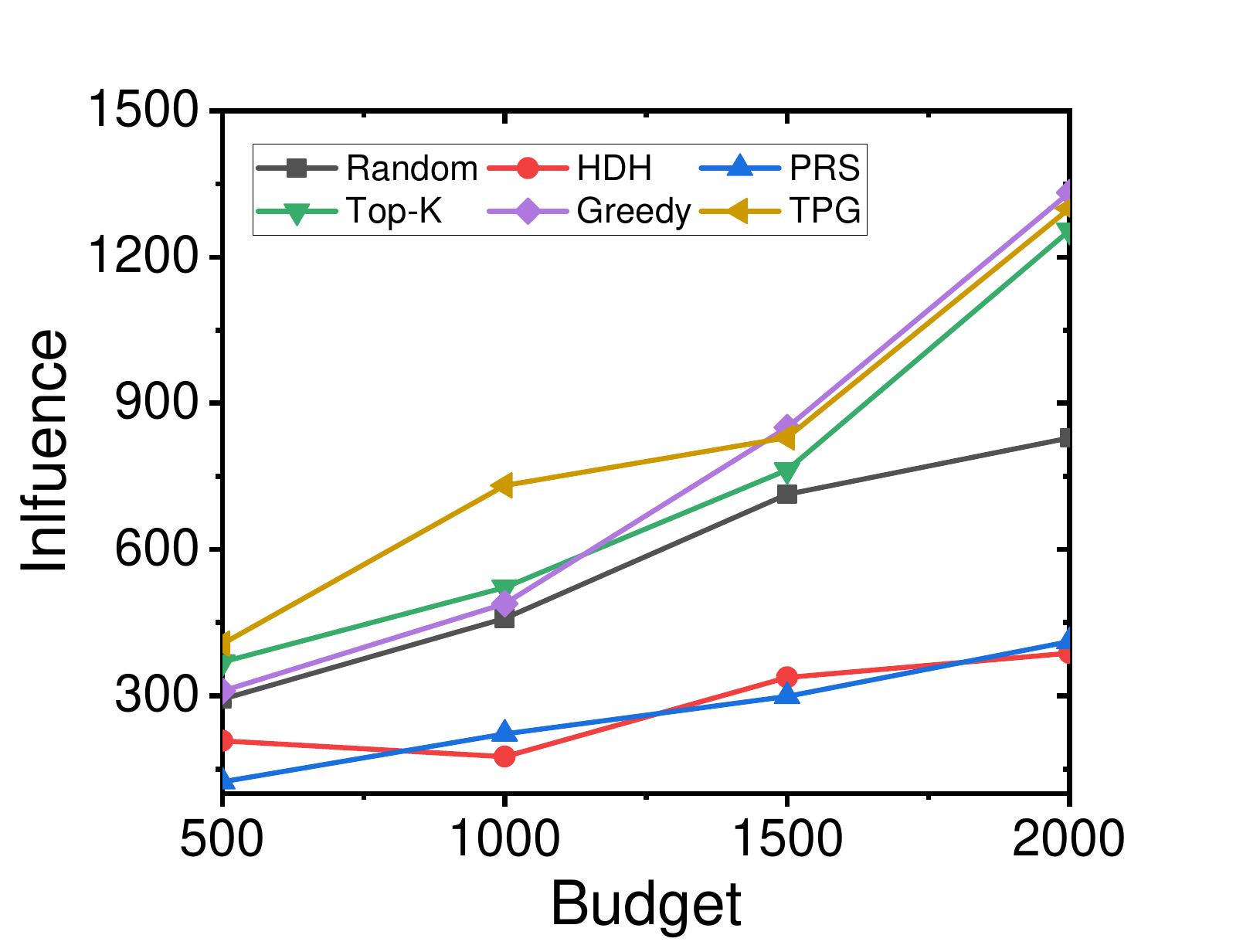} & \includegraphics[scale=0.13]{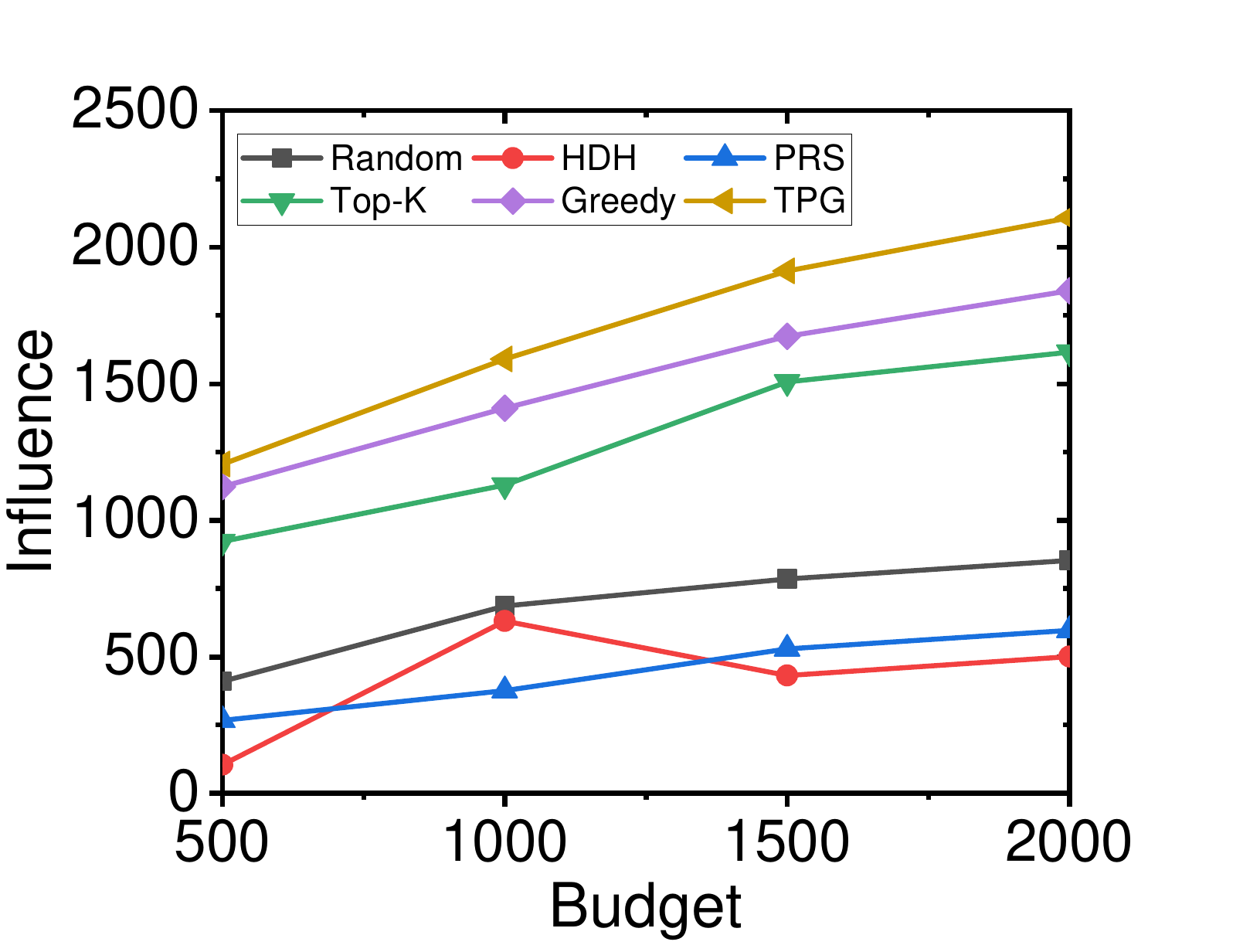} \\
\tiny{(i) Weighted Cascade } & \tiny{(j) Uniform} & \tiny{(k) Trivalency} & \tiny{$(\ell)$ Weighted} \\ 

\includegraphics[scale=0.13]{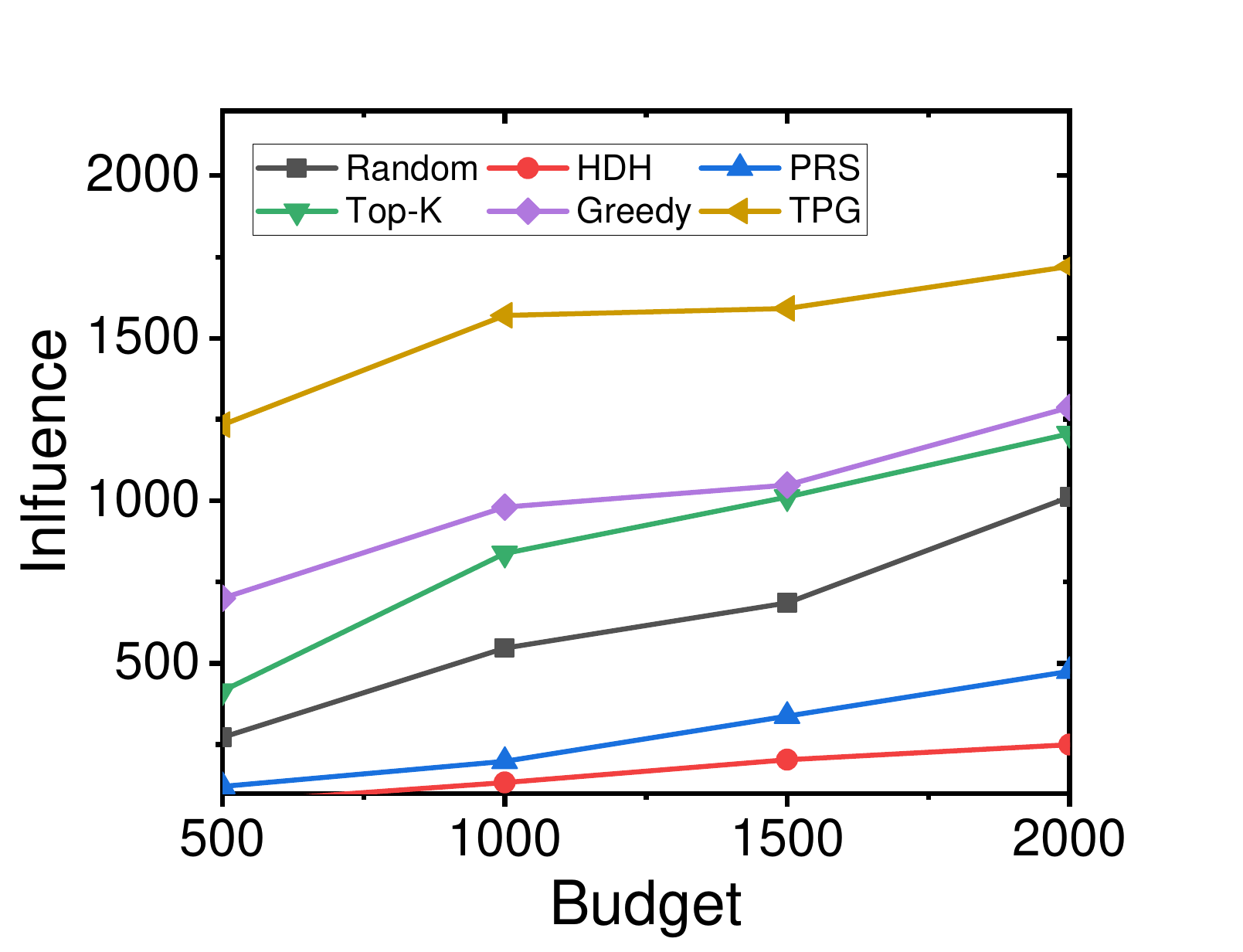} &
\includegraphics[scale=0.13]{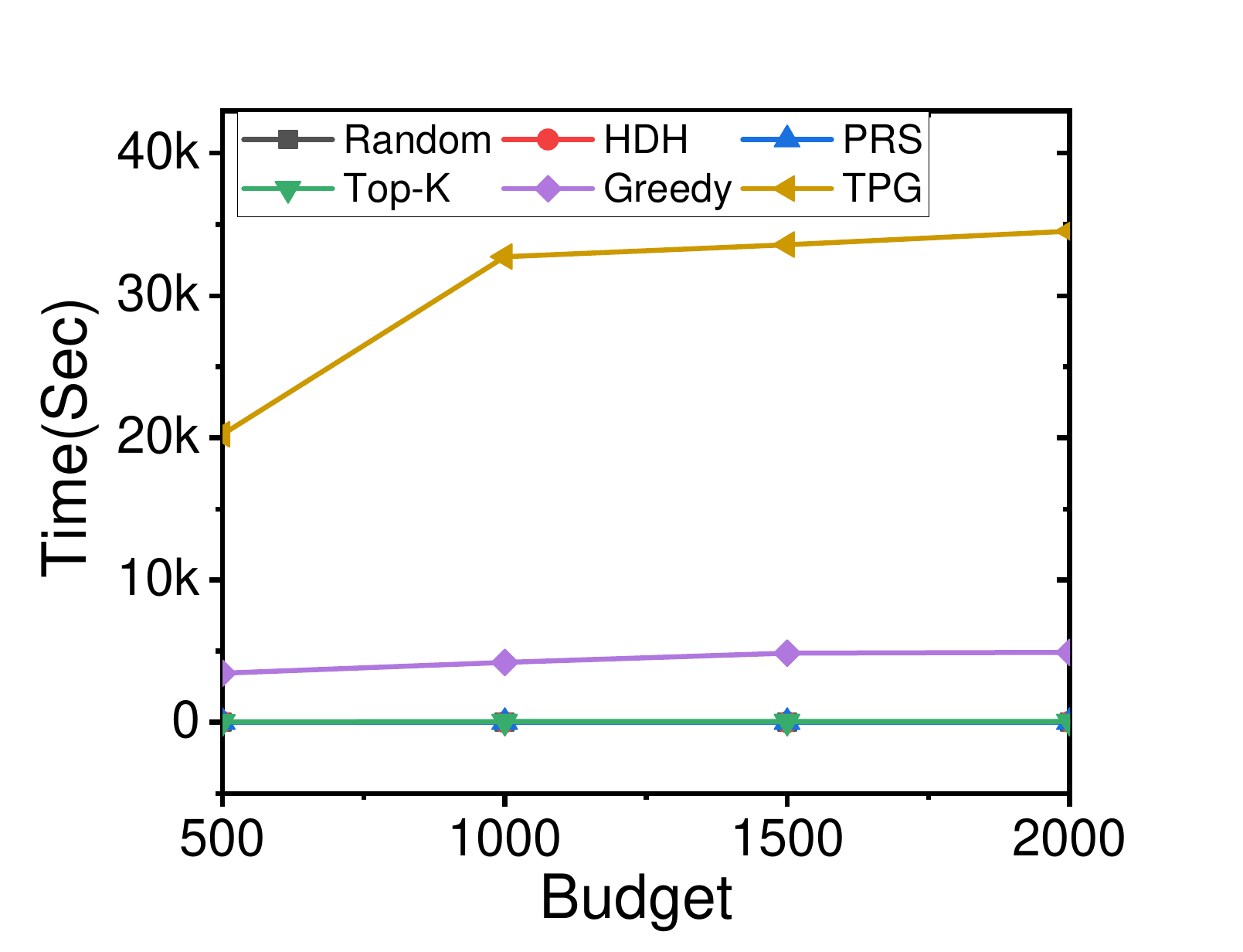} & 
\includegraphics[scale=0.13]{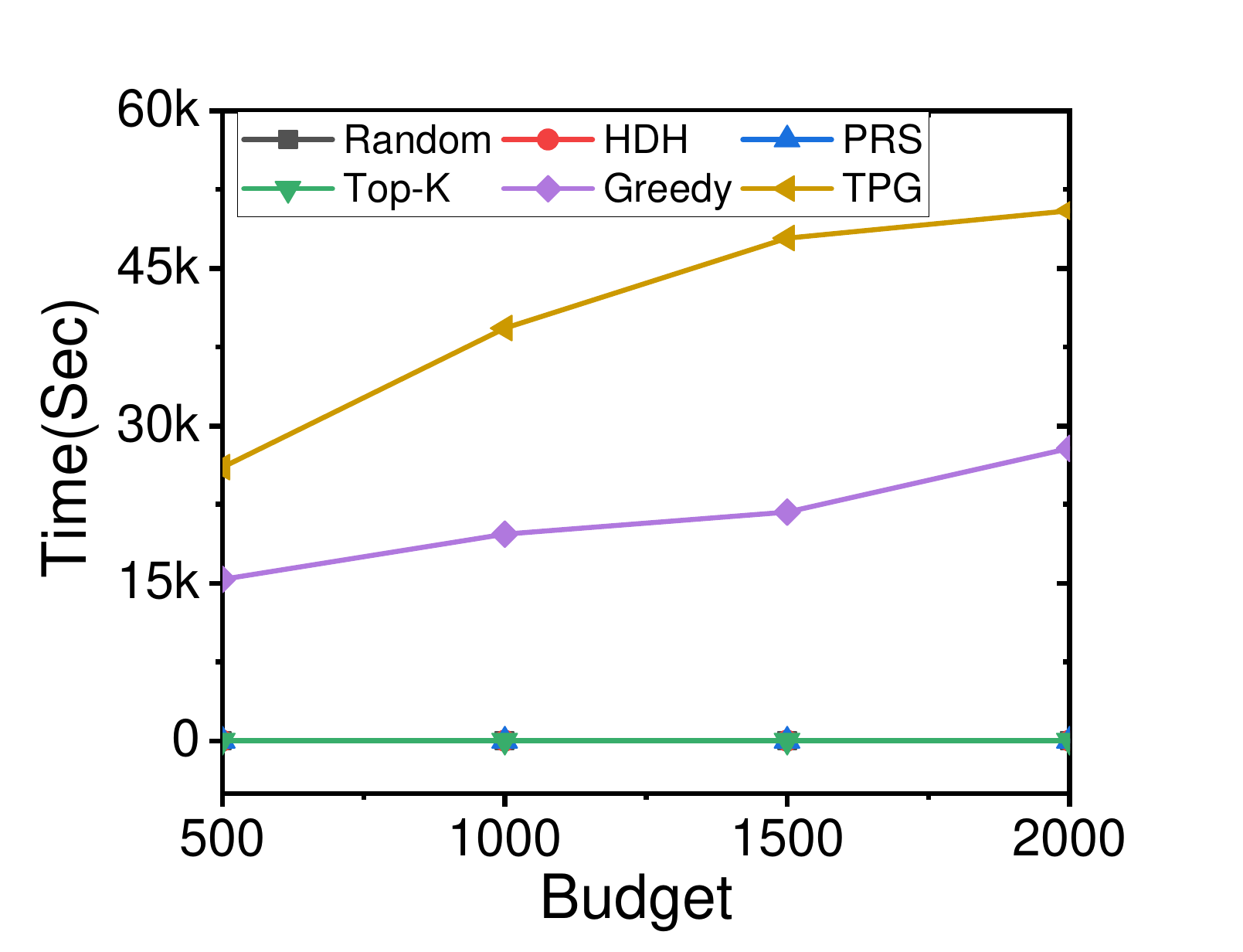} & 
\includegraphics[scale=0.13]{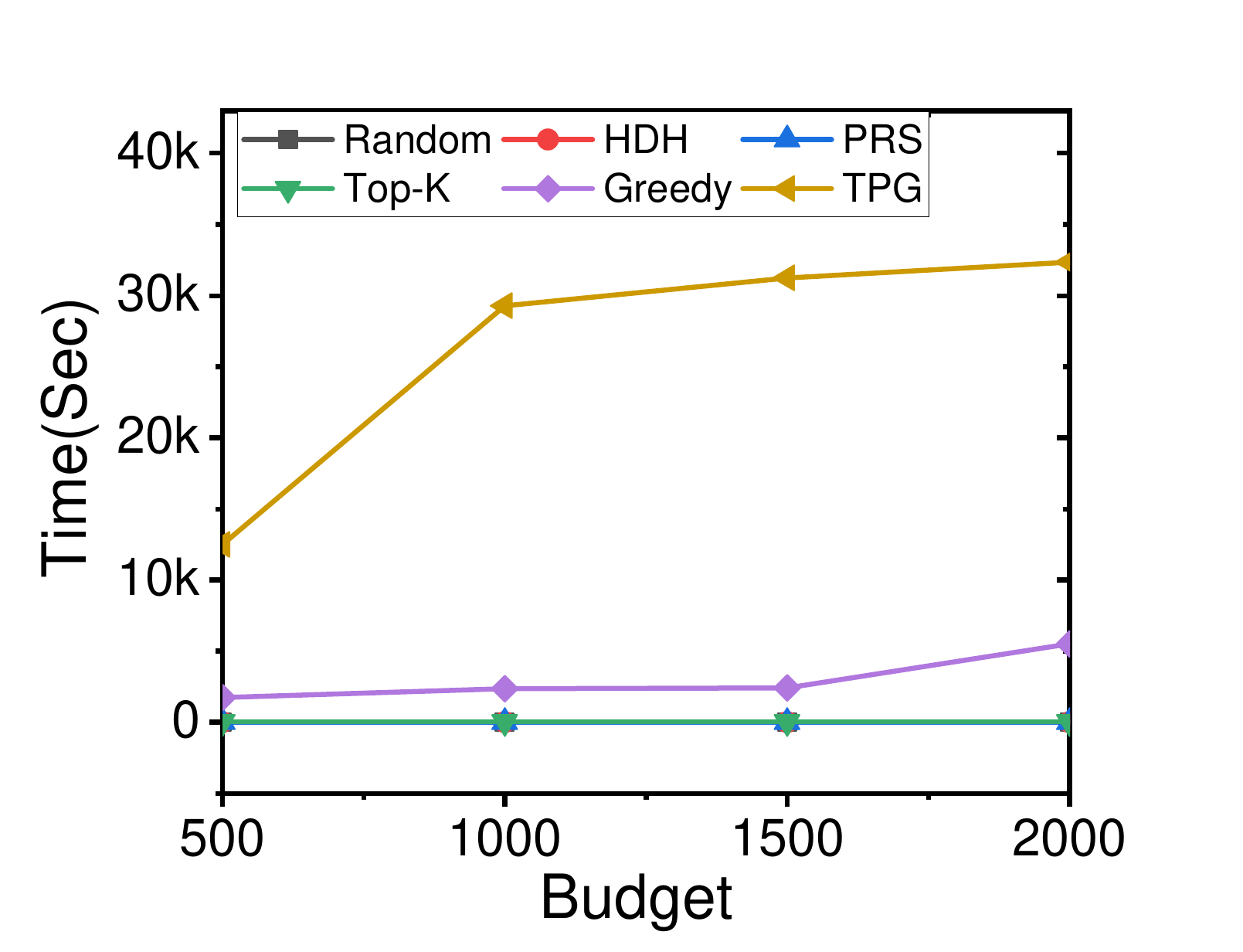} \\
\tiny{($m$) Uniform} & \tiny{($n$) Trivalency} & \tiny{($o$) Weighted Cascade} & \tiny{$(p)$ Uniform} \\
\end{tabular}
\caption{Varying Algorithms Vs. Budget split percentage in Trivalency $(a,b,c,d)$, Budget Vs. Influence for Trivalency $(e)$, Weighted Cascade $(f)$, Uniform $(g)$, Budget Vs. Time for Trivalency $(h)$, Weighted Cascade $(i)$, Uniform $(j)$ for CA dataset. Budget Vs. Influence for Trivalency $(k)$, Weighted Cascade $(\ell)$, Uniform $(m)$, Budget Vs. Time for Trivalency $(n)$, Weighted Cascade $(o)$, Uniform $(p)$ for USA dataset.}
\label{Fig:Plot3}
\end{figure*}
%%%%%%%%%%%%%%%%%%%%%%%%%%%%%%%%% Budget Vs. Time %%%%%%%%%%%%%%%%%%%%%%%%%%%%%%%%%%%%%%%%

\subsection{\textbf{Baseline Methods}} 
In \textbf{Random Allocation (RA)} approach, billboard slots and social media seeds are selected randomly. It randomly selects nodes without considering any influence maximization criterion and stops when the budget is exhausted. In \textbf{Top-$k$ Allocation}, most influential billboard slots and seed nodes are selected till their respective demand and budget constraints are satisfied. In \textbf{High-Degree Heuristic (HDH)} approach, billboard slots are sorted based on high impression count, and social media seeds are sorted by highest out-degree. It greedily selects nodes until the budget is exhausted. The \textbf{Page Rank Based Selection (PRS/PGRS)} selects influential nodes based on their global importance by computing PageRank scores and choosing the top-ranked nodes until the budget is exhausted.

\subsection{\textbf{Goals of our Experiments.}} \label{Sec:Research_Questions}
The following research questions are our focus in this study.\textbf{RQ1}: Varying budgets, how the influence of social networks and billboards varies. \textbf{RQ2}: Varying budgets, how the budget split ratio in social networks and billboard advertisements varies. \textbf{RQ3}: Varying budgets, how the overall influence and run time vary.

\subsection{\textbf{Experimental Results and Discussions.}} 
We now present the experimental results of the proposed methods and address the research questions outlined in the experimental goals.
\subparagraph{Varying Algorithms, Budget Vs. Influence.} Figure \ref{Fig:Plot1} shows the effectiveness of the proposed solution approach. We have three main observations. \textbf{First}, with the budget increasing from $500$ to $2000$, the overall influence of both the proposed and baseline methods increases almost $3\times$ times in both the CA and USA datasets. \textbf{Second}, in the probability settings of trivalency (Figure \ref{Fig:Plot1} $(a,b,c,d)$), weighted cascade (Figure \ref{Fig:Plot1} $(e,f,g,h)$) and uniform (Figure \ref{Fig:Plot1} $(i,j,k,\ell)$), most of the influence of the baseline algorithms comes from the seed nodes of the social network. However, the total influence proposed from `Randomized Greedy' and `TPG' is distributed across the billboard slots and seed nodes, with a significant portion coming from the billboards. \textbf{Third}, in the Figure \ref{Fig:Plot3}$(e,k)$ for trivalency, Figure \ref{Fig:Plot3}$(f,\ell)$ weighted cascade, and Figure \ref{Fig:Plot3}$(g,m)$ for uniform probability setting show the trade-off between budget and influence for CA and USA dataset. The proposed `Randomized Greedy' and `TPG' achieve approximately $2\times$ to $3\times$ more influence than the baselines.

\subparagraph{Varying Algorithms Vs. Budget Split Percentage.} Figure \ref{Fig:Plot2} $(a,b,c,d)$ in trivalency, Figure \ref{Fig:Plot2} $(e,f,g,h)$ in the weighted cascade and in a uniform probability setting \ref{Fig:Plot2} $(i,j,k,\ell)$ presents the percentage of budget split with varying budgets for the CA dataset. We have three main observations.\textbf{ First}, for baseline methods, a large portion of the budget is allocated to social network advertising across all probability settings. With budgets ranging from $500$ to $2000$, the budget split share increases in both `Randomized Greedy' and `TPG' for the billboard slots. This happens because the cost of slots is lower than that of seeds, and slots provide better influence than seeds in a limited budget. \textbf{Second}, in the baseline approaches, in most cases, almost $94\%$ to $99\%$ of the total budget is allocated to social networks, and only $1\%$ to $5\%$ to billboard advertisements. Among the baselines, `Random' and `Top-$k$' perform well compared to the `HDH' and `PRS' approaches.\textbf{ Third}, the budget split results suggest that, in the CA dataset, advertisers benefit from allocating a larger share of the budget to billboard advertising than to social networks. In contrast, the USA dataset exhibits an opposite trend: as the budget increases from $500 to $2000, the share allocated to billboards decreases while social network spending increases across all probability settings, as illustrated in Figure~\ref{Fig:Plot3}(a–d) for the trivalency model.

\subparagraph{Efficiency Test.}
Figure \ref{Fig:Plot3}(h,n), Figure \ref{Fig:Plot3}(i,o), and Figure \ref{Fig:Plot3}(j,p) show the efficiency of the proposed and baseline methods in trivalency, weighted cascade, and uniform probability settings, respectively. We have three main observations.\textbf{ First}, with the increase in budget from $500$ to $2000$, the computational time of all the proposed and baseline methods increases. In the CA dataset, the proposed `Randomized Greedy' method takes longer than `TPG'. However, we have a different observation on the USA dataset. The `TPG' takes longer than the `Randomized Greedy', because in the worst case, it behaves like an incremental greedy. On the other hand, `Randomized Greedy' takes sampling to select a seed or slots, which takes less time than the normal incremental greedy. \textbf{Second}, all the baseline methods require much less runtime than the proposed solution methodologies. \textbf{ Third}, in the CA and USA dataset, the weighted cascade has a higher runtime than the uniform and trivalency settings. Weighted Cascade is computationally more expensive since edge activation probabilities depend on node degrees and must be computed dynamically during diffusion, unlike trivalency and uniform models, which use fixed probabilities.

\subparagraph{Scalability Test.}
Figure \ref{Fig:Plot1}, Figure \ref{Fig:Plot2} and Figure \ref{Fig:Plot3}$(a,b,c,d)$  shows the scalability of the proposed solution approach. To demonstrate scalability, we vary the budget from $500$ to $2000$ for both the proposed and baseline methods. We have two main observations. First, with an increase in budget, the influence also increases. The proposed solution approach achieves stronger influence in all trivalency, weighted cascade, and uniform probability settings than the baseline methods. Second, among the baseline methods, `Top-$k$' has the greatest influence; however, it quickly exhausts the total budget. The `Random' approach gives the worst results than others because of the random selection of slots or nodes. The experimental results on real-world datasets from Canada and the USA demonstrate the scalability of our proposed solutions in the real-world scenarios shown in Figures \ref{Fig:Plot1}, \ref{Fig:Plot2}, and \ref{Fig:Plot3}.

\subparagraph{Additional Discussion.}
The additional parameters used in our experiments are $\epsilon$, $R$, and $\lambda$. First, $\epsilon$ is the sampling parameter that decides the random sampling size used in the selection of slots or seeds. When the $\epsilon$ value increases, the sample set size and run time decrease, but the quality of the solution also decreases. Second, $\lambda$ is the distance parameter, which signifies the distance range (in meters) that influenced the trajectories of the ad space. With the increase of $\lambda$ from $25m$ to $100m$, the influence of the proposed and baseline methods also increases, as one billboard slot can influence more than one user at a time. Third, $R$ is the number of simulations used in the independent cascade model. In our experiments, we use $\epsilon = 0.01, \lambda = 100m, R =1000$ as the default setting. We have experimented with different values of $\delta, \lambda, R$; however, due to space limitations, not all are reported.
% \vspace{-0.5in}
\section{Concluding Remarks} \label{Sec:Conclusion}
In this paper, we study the budget-splitting problem: given a budget, in what proportions should it be split to maximize total influence. We introduce the notion of the \texttt{Interaction Effect} and, based on this notion, we make a noble discrete optimization formulation of our problem. We show that the problem is NP-hard, and the objective function is non-negative, monotone, and non-bisubmodular. We have proposed `Randomized Greedy' and `TPG' approaches for solving our problem, which provide an approximation guarantee of $\frac{1}{\alpha}\left(1 - e^{-\gamma \alpha} \right)$ where $\alpha$ and $\gamma$ denote the generalized curvature and bisubmodularity ratio, respectively. An analysis was conducted to understand the time, space requirements, and performance guarantee. The reported experimental results show that the proposed solution approach yields greater influence and a better budget split ratio than the baseline methods, with reasonable computational overhead.

%% Bibliography
\bibliography{lipics-v2021-sample-article}
\end{document}